\documentclass[11pt, oneside]{article} 
\usepackage[margin=2cm]{geometry}                		
\usepackage[parfill]{parskip}    		
\usepackage{graphicx}			
\usepackage[utf8]{inputenc}
\usepackage[english]{babel} 
\usepackage{amsmath,amssymb,amsthm,bm}
\usepackage[margin=2cm]{geometry}
\usepackage[color=yellow]{todonotes}
\usepackage{mathrsfs}
\usepackage{url}	
\usepackage{esint}
\usepackage[colorlinks]{hyperref}
\usepackage{enumerate}
\usepackage{outlines}[enumerate]
\usepackage{framed,comment,enumerate}
\usepackage{upgreek}
\usepackage{pgfplots}
\usepackage{float}
\usepackage{tikz,tikz-cd}
\usepackage{rotating}
\usepackage{graphicx}
\usepackage{caption}
\usepackage{multicol}
\usepackage{cancel}
\usepackage{subcaption}
\usepackage[title]{appendix}
\usepackage{tikz-cd}
\usepackage{pgf, pgffor}
\usepackage{subcaption}
\usepackage{stmaryrd}
\usetikzlibrary{decorations.pathmorphing,arrows.meta,decorations.markings}

\DeclareFontFamily{U}{MnSymbolC}{}
\DeclareSymbolFont{MnSyC}{U}{MnSymbolC}{m}{n}
\DeclareFontShape{U}{MnSymbolC}{m}{n}{
    <-6>  MnSymbolC5
   <6-7>  MnSymbolC6
   <7-8>  MnSymbolC7
   <8-9>  MnSymbolC8
   <9-10> MnSymbolC9
  <10-12> MnSymbolC10
  <12->   MnSymbolC12}{}
\DeclareMathSymbol{\intprod}{\mathbin}{MnSyC}{'270}

\extrafloats{100}
\numberwithin{equation}{section}
\newtheorem{theorem}{Theorem}[section]
\newtheorem{lemma}{Lemma}[section]
\newtheorem{corollary}{Corollary}[section]
\newtheorem{proposition}{Proposition}[section]
\newtheorem{remark}{Remark}[section]

\theoremstyle{definition}
\newtheorem{definition}{Definition}[section]
\newtheorem{assumption}[theorem]{Assumption}

\usepackage[color=yellow]{todonotes}

\makeatletter
\@tfor\next:=abcdefghijklmnopqrstuvwxyzABCDEFGHIJKLMNOPQRSTUVWXYZ\do{
\def\command@factory#1{%
\expandafter\def\csname b#1\endcsname{\mathbf{#1}}
\expandafter\def\csname fk#1\endcsname{\mathfrak{#1}}
\expandafter\def\csname bb#1\endcsname{\mathbb{#1}}
\expandafter\def\csname cl#1\endcsname{\mathcal{#1}}
\expandafter\def\csname bcl#1\endcsname{\mathbfcal{#1}}
}
\expandafter\command@factory\next
}

\newcommand{\dede}[2]{\frac{\delta #1}{\delta #2}}

\newcommand{\wt}[1]{\widetilde{#1}}

\newcommand{\mb}[1]{\mathbf{#1}}

\newcommand{\mc}[1]{\mathcal{#1}}

\newcommand{\mcal}[1]{\mc{#1}}
\newcommand{\scp}[2]{\left<#1\,,\,#2\right>}
\newcommand{\qv}[2]{\left[#1\,,\,#2\right]}

\newcommand{\ad}{\operatorname{ad}}
\newcommand{\Ad}{\operatorname{Ad}}
\newcommand{\dd}{\mathrm{d}}

\def\dt{\mathrm{d}t}
\def\dW{\mathrm{d}W}

\def\p{{\partial}}

\def\rmd{{\rm d}}

\def\bu{{\mathbf{u}}}

\makeatother

\def\p{\partial}

\pgfplotsset{compat=1.16}

\begin{document}
	\title{\vspace{-2em}
    Variational closures for composite homogenised fluid flows}
    \author{Theo Diamantakis$^1$, Ruiao Hu$^1$\footnote{Corresponding author, email: ruiao.hu15@imperial.ac.uk} \, and James-Michael Leahy$^2$ \vspace{0.4cm}\\
1. Department of Mathematics, Imperial College London, UK\\
2. Department of Mathematics, Imperial College London, UK and PhysicsX, UK
}
	\date{\today}
	\maketitle
\begin{abstract}
Homogenisation theory has seen recent applications \cite{CGH2017, FP2021, DH2023} in deriving stochastic transport models for fluid dynamics. In this work, we first derive the stochastic Lagrange-to-Euler map that underpins stochastic transport noise in fluid dynamics as the homogenisation limit of a parameterised flow map decomposing into rapidly fluctuating and slow components. Specifically, we prove convergence of this parameterised flow map to a scale-separated limit under the assumptions of a weak invariance principle for the rapidly fluctuating component and path continuity for the slow component. In this limit, the rapidly fluctuating component converges to a stochastic flow of diffeomorphisms that transforms the full flow dynamics into an SDE-governed stochastic flow through composition, while the slow component requires closure.

Our second contribution formulates two distinct variational closures for the slow component of the homogenised flow that exploit the composite structure of the stochastic flow. For the first closure, the critical points of a new variational principle satisfy a system of random-coefficient PDEs, which can be transformed into a system of stochastic PDEs via the coadjoint action of the stochastic flow map obtained from homogenising the fluctuating component. We show that these equations coincide with the stochastic Euler-Poincaré equations previously derived in \cite{Holm2015}. For the second closure, we modify the assumptions on the slow component and the associated variational principle to derive averaged models inspired by previous work on mean flow dynamics such as the Generalized Lagrangian Mean.
\end{abstract}


\textbf{MSC Classification (2020):} 60L20, 60L90, 60H15, 49S05, 35Q31, 70Hxx.

\tableofcontents

\section{Introduction}

When modelling complex and often turbulent Geophysical Fluid Dynamics (GFD), one inevitably encounters physical processes that exist over a wide range of spatial and temporal scales. To accurately represent these processes in finite-resolution numerical simulations, parameterisation schemes are required. By making these schemes stochastic, model uncertainties can be introduced naturally. Stochastic parameterisation schemes have seen operational use in numerical weather forecasting \cite{BMP1999, BTB2015}, where they have improved ensemble forecasting reliability and probabilistic skill scores. To guide the design of stochastic parameterisation schemes, one design principle is to preserve the geometric structures of the underlying deterministic fluid models.

The geometric structures of ideal fluid models are revealed by Lie-group invariant variational principles. As first noted by Arnold \cite{Arnold1966}, given a domain $\mcal{D}\subset \mathbb{R}^d$, solutions of the incompressible Euler fluid equations constitute time-dependent geodesic flow $g_t$ on the manifold of volume-preserving diffeomorphisms ($g_t\in \operatorname{SDiff}(\mcal{D})\, \forall t$). Building on Arnold's idea of modelling fluid dynamics as flow on the Lie group of diffeomorphisms, a modern approach was developed in \cite{HMR1998}, in which ideal fluid dynamics are derived from a Lie-group symmetry reduced Euler--Poincar\'e variational principle involving advected quantities. These advected quantities consist of fluid properties such as volume density, mass density, and heat that follow the Lagrangian trajectories of fluid parcels and contribute to the overall fluid dynamics. The presence of advected quantities introduces potential energies such that solutions are no longer free geodesic curves but rather forced geodesics in the manifold of diffeomorphisms. Nevertheless, solutions remain as smooth, time-parameterised curves on the manifold of diffeomorphisms.

One stochastic parameterisation approach for ideal fluid models replaces the smooth time-dependent curves of diffeomorphisms with Stratonovich stochastic processes. In this case, it is natural to consider the time-dependent stochastic flow of diffeomorphisms $g_t \in \operatorname{Diff}(\mcal{D})$ as a Stratonovich stochastic differential equation (SDE)
\begin{align}
    \dd g_t(X) := u_t(g_t(X)) \dd t + \sum_k \xi_k(g_t(X))\circ \dd W^k_t\,, \quad g_0(X) = X \,, \quad X \in \mcal{D}\,,\label{eqn:eulerian decomp g}
\end{align}
where $\xi_k$ is a collection of prescribed noise vector fields and $u_t$ is the drift velocity vector field whose evolution is to be determined. Example closures for the dynamics of $u_t$ include \cite{MR2004}, which obtains the evolution of $u_t$ by formally taking the time derivative of \eqref{eqn:eulerian decomp g} and using momentum balance arguments similar to the deterministic case while postulating an It\^o decomposition of forces, velocity fields, pressure fields, and Wiener processes. Another closure is the Location Uncertainty (LU) framework \cite{Memin2014, debussche2024variational}, which derives the dynamics of $u_t$ through a stochastic Reynolds transport theorem with additional forces required for energy conservation. For the scope of this paper, we consider the stochastic parameterisation framework known as Stochastic Advection by Lie Transport (SALT) \cite{Holm2015}, which derives the dynamics of $u_t$ through variational principles. To date, there are several derivations of SALT using stochastic variational principles, including Clebsch \cite{Holm2015, H2021} and Hamilton-Pontryagin \cite{Arnaudon2018, GH18a} approaches. These capabilities in structure-preserving stochastic parameterisations have led to numerous works on implementing SALT parameterisation in GFD models \cite{CCHOS18a, CCHOS18, HP2023, CLLLP2023, HP2025} and analysing SALT SPDEs \cite{CFH19, CL2023}. Additionally, these stochastic parameterisations have driven new developments in stochastic data assimilation methods for GFD models using particle filters \cite{CCHPS2021}.

The solution to the SDE \eqref{eqn:eulerian decomp g} can be expressed as the composition of two semi-martingale flows, $\Xi_t , \bar{g}_t \in \operatorname{Diff}(\mcal{D})$, such that $g_t$ is defined by
\begin{align}
    g_t(X) = (\Xi_t \circ \bar{g}_t)(X)\,, \quad X \in \mcal{D}\,. \label{eqn:lag decomp g}
\end{align}
Here, $\Xi_t$ is a prescribed flow of diffeomorphisms satisfying $\dd \Xi_t(X) = \sum_{k}\xi_k(\Xi_t(X))\circ \dd W^k_t$, and $\bar{g}_t$ is to be determined from the dynamics of $u_t$. The approach of decomposing the flow map is popular in studying multi-scale and multi-physics systems, such as wave-current interactions \cite{HHS2023} and plasma dynamics \cite{H2001}. In \cite{CGH2017}, the Eulerian decomposition \eqref{eqn:eulerian decomp g} of $g_t$ is motivated by the homogenisation of a fast-slow decomposition $g_t = \Xi_{t/\varepsilon} \circ \bar{g}^\varepsilon_t$, where $\varepsilon \in \mathbb{R}^+$ is a small parameter and $\Xi_{t/\varepsilon}$ is assumed to be a rapidly fluctuating map of the form $\Xi_{t/\varepsilon}(X) = X + \zeta_{t/\varepsilon}(X)$ for a smooth function $\zeta_{t/\varepsilon} : \mcal{D} \rightarrow \mcal{D}$ acting on the local coordinates $X$ of the domain. The present work extends the analysis in \cite{CGH2017, diamantakis2024levy} by considering $g_t$ as the homogenisation limit of the $\varepsilon$-parameterised flow $g_t^\varepsilon$ defined by the composition $g^\varepsilon_t = \Xi^\varepsilon_t \circ \bar{g}^\varepsilon_t$, where we do not assume a linear decomposition of the fluctuating map $\Xi^\varepsilon_t$. Additionally, all three flows have non-trivial dependencies on $\varepsilon$. As we shall explain in Section \ref{sec:homog}, the $\varepsilon$-dependence of $\bar{g}^\varepsilon_t$ is crucial for the consistency of models arising from stochastic variational closures of the dynamics of $g_t$.

Complementary to this approach, the application of homogenisation theory to derive stochastic fluid dynamics with transport noise is not limited to homogenising slow-fast equations of Lagrangian particle trajectories. For example, \cite{flandoli2022additive, FP2021} applies stochastic homogenisation to a slow-fast system of 2D Euler-like stochastic PDEs to derive the 2D Euler equation with transport noise. This is extended in \cite{debussche2024second,DH2023} for a general class of PDE systems, from which stochastic Navier-Stokes equations in 2D and 3D can be obtained.

\paragraph{Main contributions.} 
The two main contributions of this work are as follows. First, we use the theory of rough flows to establish sufficient conditions on the flows of diffeomorphisms $\Xi^\varepsilon$ and $\bar{g}^\varepsilon$ to rigorously obtain a stochastic flow $g$ in the $\varepsilon \rightarrow 0$ limit of the composition $g^\varepsilon = \Xi^\varepsilon \circ \bar{g}^\varepsilon$. The precise convergence result is stated in Theorem \ref{thm:convergence}. Second, in Section \ref{sec:vp}, we derive stochastic closure dynamics for the drift velocity $u_t$ in \eqref{eqn:eulerian decomp g} through variational principles that exploit the geometric structure of the homogenised flow $g = \Xi \circ \bar{g}$ obtained in Section \ref{sec:homog}. These variational principles do not require explicit stochastic constraints and directly reveal the connections between stochastic and random-coefficient Euler--Poincar\'e equations based on flow map representations. This work opens a new avenue for analysing stochastic Euler--Poincar\'e equations, as the analysis of stochastic PDEs based on their flow maps is well established among analysts.

\paragraph{Outline of the paper}

We summarise the main content and results in the following sections. The homogenisation contribution in Section \ref{sec:homog} and the variational closures contribution in Section \ref{sec:vp} may be read independently according to the reader's particular interests.

\begin{itemize}
    \item In Section \ref{sec:homog}, we construct a stochastic flow of diffeomorphisms $\Xi$ as the limit of a fast, chaotic $\varepsilon$-dependent flow $\Xi^\varepsilon$ using deterministic homogenisation. This is achieved by assuming an iterated Weak Invariance Principle (WIP) for the chaotic dynamics and utilising the continuity of the rough flow associated with $\Xi^\varepsilon$. For the mean flow $\bar{g}^\varepsilon$, we assume it is $\varepsilon$-dependent and converges to a stochastic flow $\bar{g}$. When the fast flow $\Xi^\varepsilon$ is composed with the mean flow $\bar{g}^\varepsilon$, the composite flow $g^\varepsilon = \Xi^\varepsilon \circ \bar{g}^\varepsilon$ describes a fluid with $\mcal{O}(\varepsilon^{-1})$ scale separation in the velocity fields defined by the flows $\Xi^\varepsilon$ and $\bar{g}^\varepsilon$. By proving the continuous dependence of $g^\varepsilon$ on the pathwise noise data, we establish convergence of the composite flow $g^\varepsilon$ to a stochastic flow of diffeomorphisms $g = \Xi \circ \bar{g}$. This ensures that the stochastic vector field associated with the flow $g$, $dg_tg_t^{-1} = \Xi_{t *}\bar{u}_t \, \dd t + \xi_k \circ \dd W^k_t + \frac{1}{2}\Gamma^{kl}\qv{\xi_k}{\xi_l} \, \dt$, operates on a single timescale. This coincides with the stochastic Lagrangian trajectories of fluid particles \eqref{eqn:eulerian decomp g} upon identifying the drift coefficient $u_t = \Xi_{t *}\bar{u}_t$ with an additional drift contribution $\frac{1}{2}\Gamma^{kl}\qv{\xi_k}{\xi_l}$ induced by the noise properties.
    
    \item In Section \ref{sec:vp}, we derive closure dynamics for the mean vector field $\bar{u}_t$ associated with the mean flow $\bar{g}$. First, we construct variations of the stochastic flow of diffeomorphisms $g = \Xi \circ \bar{g}$ and fluid advected quantities $a_t = a_0 g_t^{-1}$ that respect the flow composition structure. Starting from a Lagrangian that depends on the non-singular drift velocity $u_t = \Xi_{t*} \bar{u}_t$ and using these variations, we derive a stochastic closure for the dynamics of $u_t$ that coincides with the stochastic Euler--Poincar\'e equations obtained through the SALT approach. Additionally, we show that these stochastic Euler--Poincar\'e equations are equivalent to random-coefficient Euler--Poincar\'e equations for the dynamics of $\bar{u}_t$ via a flow map transformation arising from a time-dependent Lagrangian that depends on $\bar{u}_t$. The random dynamics of $\bar{u}_t$ thus justify the initial $\varepsilon$-dependence assumption of $\bar{g}^\varepsilon$ made in Section \ref{sec:homog}. When the flow of diffeomorphisms $\Xi$ is an isometry at all times, the resulting stochastic equations preserve the deterministic energy. We further study possible deterministic closures for $\bar{u}_t$ by proposing averaged models in Section \ref{sec:det vp}, which may be motivated by relaxing the convergence assumptions outlined in Section \ref{sec:homog}.
    
    \item In Section \ref{sec:examples}, we consider the illustrative example of the incompressible Euler equations to demonstrate the results of Section \ref{sec:vp}. We show the equivalence between the stochastic incompressible Euler equations for $u_t$ and the random-coefficient counterpart for $\bar{u}_t$. In this case, the mean dynamics can be interpreted as introducing stochastic time dependence to the Riemannian metric, which can then be averaged to produce a deterministic PDE model. The relative equilibrium solutions of a stochastic point vortex model studied in \cite{diamantakis2024levy} can be viewed as an application of the mean formulation proposed here.

    \item In Section \ref{prevwork}, we compare the present work with previous studies on homogenisation of Lagrangian particle flow maps. We clarify the differences in modelling approaches between the current work and \cite{CGH2017}, address a gap in the proof of the homogenisation limit in \cite{CGH2017} by reformulating the analysis using the construction presented in Section \ref{sec:homog}, and discuss interpretations of these modelling approaches in the context of stochastic parameterisation schemes.
    
    \item Section \ref{sec:conclusion} contains concluding remarks and future work.

    \item Appendix \ref{app} contains a review of rough path theory used for the homogenisation analysis in Section \ref{sec:homog}. We have also included a finite-dimensional example of rigid body rotation to demonstrate the correspondence between random-coefficient and stochastic Euler--Poincar\'e equations obtained in Section \ref{sec:vp} for a left-invariant system defined on the Lie algebra $\mathfrak{so}(3)$.
\end{itemize}

\section{Deterministic homogenisation of flow maps}\label{sec:homog}

In this section, we derive the stochastic Lagrangian flow map of SALT \eqref{eq:SALTansatz} as the homogenised limit of a parameterised flow map $g^\varepsilon$ that decomposes into rapidly fluctuating and slow components. Specifically, we construct $g^\varepsilon$ as a composition $g^\varepsilon = \Xi^\varepsilon \circ \overline{g}^\varepsilon$, where $\Xi^\varepsilon$ represents the rapidly fluctuating component and $\overline{g}^\varepsilon$ represents the slow component, with two distinct time scales distinguished by a parameter $\varepsilon$. We then prove that in the limit $\varepsilon \rightarrow 0$, corresponding to infinitely large scale separation, this parameterised composite flow converges to the stochastic flow of diffeomorphisms satisfying the ansatz \eqref{eq:SALTansatz}.

Let $(\clD, \bg)$ denote a smooth compact $d$-dimensional boundaryless Riemannian manifold with metric $\mathbf{g}$. Let $T>0$ and $(\Omega,\mathcal{F},\bbP)$ denote a probability space supporting an $K$-dimensional Brownian motion $W=(W^1,\ldots, W^K)$ with identity covariance.  Let $\bbF$ denote the filtration generated by $W$. Let $n\in \bbN$, $\fkX_{C^{n}}(\clD)$ denote the space of $n$-times continuously differentiable vector fields on $\clD$ and $\textnormal{Diff}_{C^n}(\clD)$ denote the space of $C^n$-diffeomorphisms (see, Definition \ref{def:diffeo}). 

The base assumption of the Stochastic Advection by Lie Transport (SALT) modelling approach \cite{Holm2015} is that the Lagrange-to-Euler map for a fluid is governed by an $\bbF$-adapted stochastic flow of diffeomorphisms $g: \Omega \times [0,T]\rightarrow \textnormal{Diff}_{C^n}(\clD)$ such that for each $X\in \clD$, $g(X):\Omega\times [0,T]\rightarrow \mathcal{D}$ is the unique strong solution of the SDE
\begin{align}\label{eq:SALTansatz}
    \dd g_t(X) = u_t(g_t(X)) \dd t + \sum\limits_{k=1}^K \xi_k(g_t(X))\circ \dd W^k_t \,, \quad g_0(X)=X\,,
\end{align}
where  $u: \Omega\times [0,T]\rightarrow \fkX_{C^n}(\clD)$ is an $\bbF$-progressively measurable vector field, $\xi \in \fkX_{C^{n+2}}(\clD)^K$ is a collection of $K$ vector fields, and the stochastic differential is understood in the Stratonovich sense. Indeed, under the aforementioned conditions on the datum $(u,\xi)$, it is well-known that there exists an $\bbF$-adapted 
$C^n$ stochastic flow of diffeomorphisms satisfying \eqref{eq:SALTansatz} \cite[Theorem 3.4]{kunita1990stochastic, kunita1996stochastic}. In particular, strong and pathwise uniqueness hold for the SDE \eqref{eq:SALTansatz} for each $X\in \clD$. The goal of this section is to show that this Lagrange-to-Euler model arises as a (deterministic) homogenised limit of a multi-scale factorisation of a parameterised Lagrange-to-Euler map. 

Towards, this end, let $\clM$ denote a smooth compact Riemannian manifold isometrically embedded in $\bbR^K$ and $h\in \mathfrak{X}_{C^2}(\clM)$ denote a $C^2$-vector-field on $\clM$. Let $\phi: \bbR\rightarrow \operatorname{Diff}_{C^{2}}(\mcal{M})$ denote the solution flow of the ODE $\dot{\lambda}=h(\lambda)$. Assume that $\phi$ has a closed attracting set $\Omega\subset \clM$, and there exists a $\phi_t$-invariant ergodic SRB measure $\bbP\in \mathcal{P}(\clM)$ supported on $\Omega$ such that the following centering condition holds
$$\int_{\Omega} \lambda \bbP(\dd \lambda)=0\,,$$
where the integral is understood in the embedding space $\bbR^K$.\footnote{One can avoid embedding the manifold in $\mathbb{R}^K$ using observables, however, we avoid this for simplicity.} Let $\varepsilon\in (0,1)$ denote a parameter that will dictate the scale at which certain quantities fluctuate in time. Let $\bar{u}^{\varepsilon}:\Omega\times [0, T]\rightarrow \fkX_{C^n}(\clD)$ denote a measurable map and $\sigma \in \fkX_{C^{n+2}}(\clD)^K$. We postulate a parameterised Lagrange-to-Euler map $g^{\varepsilon}: \Omega\times [0,T]\rightarrow  \operatorname{Diff}_{C^n}(\mcal{D})$ that factorises as
\begin{align}\label{eq:CoMdef}
g^\varepsilon_t(\omega)=\Xi^\varepsilon_{t}(\omega)\circ \overline{g}^\varepsilon_t(\omega)\,,
\end{align}
where $\Xi^{\varepsilon}, \bar{g}^{\varepsilon}: \Omega\times [0, T] \rightarrow \operatorname{Diff}_{C^n}(\mcal{D})$ are  random flows of (at least) $C^n$-diffeomorphisms satisfying
\begin{align}
\dot{\Xi}^{\varepsilon}_t(\omega,\bar{X})&=\varepsilon^{-1}\sum\limits_{k=1}^K\sigma_k(\Xi_t^{\varepsilon}(\omega,\bar{X}))\lambda_t^{\varepsilon,k}(\omega)\,, \quad 
\Xi^{\varepsilon}_0(\bar{X})=\bar{X}\in \clD\,,\label{eq:Xi_eps}\\
\dot{\bar{g}}_t^{\varepsilon}(\omega, X) &= \bar{u}_t^{\varepsilon}(\omega,\bar{g}_t^{\varepsilon}(\omega, X)), \quad \bar{g}_0^{\varepsilon}(X)=X\in \clD\,,\label{eq:bar_g_eps}\\
\dot{\lambda}_t^{\varepsilon}(\omega) &= \varepsilon^{-2} h(\lambda^\varepsilon_t(\omega))\,, \quad \lambda_0=\omega \in \Omega \,,\label{eq:lambda_eps}
\end{align}
Henceforth, we will drop the $\omega$-dependence to simplify notation. 
 
For a given label $\bar{X} \in \clD$, the coupled system $(\Xi^{\varepsilon}(\bar{X}),\lambda^{\varepsilon})$ forms a fast-slow skew-product system in which $\Xi^{\varepsilon}(\bar{X})$ is the ``slow'' variable and $\lambda^{\varepsilon}$ is the ``fast'' variable \cite{melbourne2011note}. The flow $\bar{g}^{\varepsilon}$ is generated from a vector field $\bar{u}^{\varepsilon}$, which itself depends on $\varepsilon$ and $\omega$. The trajectories $\bar{g}^{\varepsilon}(X)$ are the slowest of the system $(\Xi^{\varepsilon}(\bar{X}), \bar{g}^{\varepsilon}(X), \lambda^{\varepsilon})$. In this section, we do not specify the specific model for the vector field $\bar{u}^{\varepsilon}$, but in Section \ref{sec:vp}, we propose a variational principle for the limiting vector field $\bar{u}$ (see, also, Remark \ref{rem:why_gbar_is_random}).

Letting $X^\varepsilon_t=g^\varepsilon_t(X)$, $\bar{X}_t^{\varepsilon}=\bar{g}^{\varepsilon}_t(X)$ and applying the chain rule  (c.f., \cite{CGH2017, H2019}), we find
\begin{align}\label{eq:X_eps}
\dot{X}^\varepsilon_t = T{\Xi^\varepsilon_t}\dot{\bar{X}}^\varepsilon_t + \varepsilon^{-1}\sum\limits_{k=1}^K\sigma_k(\Xi_t^{\varepsilon}(\bar{X}^\varepsilon_t))\lambda_t^{\varepsilon,k}=(\Xi^\varepsilon_{t *}\bar{u}^\varepsilon_t)(X^\varepsilon_t) + \varepsilon^{-1}\sum\limits_{k=1}^K\sigma_k(X^\varepsilon_t) \lambda_t^{\varepsilon,k}\,,
\end{align}
such that
\begin{align}\label{eq:g_eps}
\dot{g}^{\varepsilon}_t(X)=\Xi^{\varepsilon}_{t*}\bar{u}^{\varepsilon}_t(g^{\varepsilon}(X)) + {\varepsilon^{-1}}\sum\limits_{k=1}^K \sigma_k(g^{\varepsilon}_t(X))\lambda_t^{\varepsilon}\,,  \quad g_0^{\varepsilon}(X)=X\in \clD\,.
\end{align}
Here, we use the notation $\phi^* v:=T\phi^{-1} \cdot v \circ \phi $ for the pullback of a vector field $v\in \mathfrak{X}_{C^n}(\clD) \simeq T_e \operatorname{Diff}_{C^n}(\clD)$ along the diffeomorphism  $\phi \in \operatorname{Diff}_{C^n}(\mcal{D})$ and  denote by $ \phi_*v=(\phi^{-1})^*v $ the push forward. 

The decompositions \eqref{eq:CoMdef} and \eqref{eq:X_eps} are similar to the factorisation assumption in the Generalized Lagrangian Mean (GLM) theory  \cite{andrews_mcintyre_1978, SOWARD_ROBERTS_2010, Holm2002, GV2018, gilbert2024geometric}. In these works, a decomposition of $X_t = \bar{X}_t + \zeta( \bar{X}_t)$ is made such that a perturbation map $\zeta : \clD \rightarrow \clD$ evaluated at $\bar{X}_t$ is assumed to be mean zero at all times. Normally, the validity of summing Lagrangian particles requires assuming a flat configuration space. Like the geometric GLM theory \cite{SOWARD_ROBERTS_2010, GV2018, gilbert2024geometric} that avoids a flat space decomposition of the decomposition map, the average of the vector field associated with the fast map $\Xi^{\varepsilon}$ over the invariant measure $\bbP$ is zero for $\varepsilon>0$:
$$
\int_{\Omega}\dot{\Xi}_t^{\varepsilon}\Xi_t^{\varepsilon;-1}(\omega, x)\bbP(d\omega)=\sum_{k=1}^{K}\sigma_k(x)\int_{\Omega}\lambda^{\varepsilon,k}_t(\omega)\bbP(d\omega)=0.
$$
However, unlike GLM, the map $\bar{g}^{\varepsilon}$ is assumed to be random.   From the modelling perspective, the main differentiating factor between $\Xi^{\varepsilon}$ and $\bar{g}^{\varepsilon}$ is the speed at which they fluctuate, which is evidenced in the parameterisation by the fact that $\dot{\Xi}_{t}^{\varepsilon} \Xi_{t}^{\varepsilon; -1}$ is of order $\mc{O}(\varepsilon^{-1})$ and  $\dot{\bar{g}}_{t}^{\varepsilon} \bar{g}_{t}^{\varepsilon; -1}$  is of order $\mc{O}(1)$. 

In the next subsection, under the assumption that the dynamical system induced by $\dot{\lambda}=h(\lambda)$ possesses a weak invariance principle (see Assumption \ref{asm:weak_invariance} (ii))  and $\bar{u}^{\varepsilon}$ has additional structure such that $\bar{u}^{\varepsilon}\rightarrow_\bbP \bar{u}$ in $C([0,T] ; \fkX_{C^{n}}(\bbT^d))$  as $\varepsilon\rightarrow 0$ (see Assumption \ref{asm:u_bar_eps}), we will show that $g^{\varepsilon}\rightarrow_{\bbP} g$ in $C([0,T] ; \operatorname{Diff}_{C^{n-2}}(\bbT^d))$ as $\varepsilon\rightarrow 0$ (see Theorem \ref{thm:convergence}), where $g$ satisfies \eqref{eq:SALTansatz}  with
$$
u_t:=\Xi_{t*}\bar{u}_t + \frac{1}{2}\sum_{k,l=1}^K\Gamma^{kl}[\xi_k, \xi_l] \quad \textnormal{and} \quad \xi :=  \sigma \sqrt{\Sigma} \,.
$$
Here, $\Sigma \in \operatorname{Sym}_+ (\bbR^K)$ is a symmetric positive-definite matrix and  $\Gamma\in \mathfrak{s} \mathfrak{o}(K)$ is an  anti-symmetric matrix that arises from the weak invariance principle.

\begin{remark}
It is not that surprising that \eqref{eq:SALTansatz} is the limit of a composition of maps. Indeed,  any stochastic flow of the form \eqref{eq:SALTansatz} can be factorised into a composition of a stochastic flow and a random flow. Indeed, let $\beta: \Omega\times [0,T]\rightarrow \fkX_{C^n}(\clD)$ denote a $\bbF$-progressively measurable vector field. Define a stochastic flow of diffeomorphisms $\Xi^{(\beta)}: \Omega \times [0,T]\rightarrow \textnormal{Diff}_{C^n}(\clD)$ by 
$$
\dd\Xi^{(\beta)}_t(\bar{X})= \beta_t(\Xi_t^{(\beta)}(\bar{X}))\dd t + \sum\limits_{k=1}^K \xi_k(\Xi_t^{(\beta)}(\bar{X}))\circ  \dd W^k_t\,,  \quad \Xi^{(\beta)}_0(\bar{X})=\bar{X}\in \clD\,,
$$
and a random flow of diffeomorphisms $\bar{g}:  \Omega \times [0,T]\rightarrow \textnormal{Diff}_{C^n}(\clD)$ by
$$
\dot{\bar{g}}_t(X) = \bar{u}_t(\bar{g}_t(X)), \quad \bar{u}_t=\Xi_t^* (u_t- \beta_t)\,, \quad \bar{g}_0(X)=X\in \clD\,.
$$
Applying the Stratonovich It\^o-Wentzel formula \cite{ventzel1965equations, de2020implications, len2023geometric}, we find  
\begin{align*}
\begin{split}
\dd (\Xi_{t}^{(\beta)}\circ \bar{g}_t(X))&= (T \Xi_{t}^{(\beta)}) \dot{\bar{g}}_t (X) + \dd \Xi_{t}^{(\beta)} \circ \bar{g}_t(X)  = (T \Xi_{t}^{(\beta)})  \bar{u}_t(\bar{g}_t(X)) + \dd \Xi_{t}^{(\beta)} \circ \bar{g}_t(X)\\
&= \Xi_{t *} \bar{u}_t(\Xi_{t}^{(\beta)}\circ \bar{g}_t(X)) \dd t + \beta_t(\Xi_{t}^{(\beta)}\circ \bar{g}_t(X)) \dd t+  \sum\limits_{k=1}^K \xi_k(\Xi_{t}^{(\beta)}\circ \bar{g}_t(X))\circ  \dd W^k_t\\
&=u_t(\Xi_{t}^{(\beta)}\circ \bar{g}_t(X)) \dd t + \sum\limits_{k=1}^K \xi_k(\Xi_{t}^{(\beta)}\circ \bar{g}_t(X))\circ \dd W^k_t\,,
\end{split}
\end{align*}
and hence $g=\Xi^{(\beta)}\circ \bar{g}$ by strong uniqueness of the SDE \eqref{eq:SALTansatz}. Since  $\beta$ is arbitrary, the factorisation of $g$ into a composition of a  stochastic and random flow map is not unique. However, there is a natural choice of $\beta$ that arises from a homogenisation limit perspective under the modelling assumption \eqref{eq:Xi_eps}, \eqref{eq:bar_g_eps} and \eqref{eq:lambda_eps}. 
\end{remark}

\subsection{Statement and proof of homogenised limit}\label{sec:proof_of_homog}
We will now assume $\mcal{D}= \bbT^d$ for simplicity. First notice that we can re-write \eqref{eq:X_eps} as 
\begin{align*}
\dot{X}^{\varepsilon}_t &= a^{\varepsilon}_t(X^{\varepsilon}_t)\dt  + \varepsilon^{-1}b(X_t^{\varepsilon}, \lambda_t^{\varepsilon}) \,,\\
\dot{\lambda}_t^{\varepsilon} &= \varepsilon^{-2} h(\lambda^\varepsilon_t)\,, \quad \lambda_0=\omega \in \Omega \,.
\end{align*}
where $$a^{\varepsilon}_t:=\Xi_{t*}^{\varepsilon}\bar{u}_t^{\varepsilon} \quad \textnormal{and} \quad b^{\varepsilon}(x,\lambda)=\sum\limits_{k=1}^K \sigma_k(x)\lambda^k\,.$$ 

In general, there does not exist a vector field $\gamma\in \fkX(\bbT^d)$ and a finite set of labels $X_i, Y_i\in \bbT^d$ such that for all $x\in \bbT^d$
$$
a_t^{\varepsilon}(x)=\gamma(\{\bar{g}_t^{\varepsilon}(X_i)\}_{i=1}^{N_1}, \{\Xi_t^{\varepsilon}(Y_i)\}_{i=1}^{N_2},\lambda_t^{\varepsilon}, x)\,.
$$
Thus, the lifted system does not immediately satisfy the assumptions of the finite-dimensional deterministic homogenisation literature \cite{melbourne2011note, melbourne2015correction, 10.1214/14-AOP979, kelly2017deterministic} required to pass to the limit as $\varepsilon\downarrow 0$. In \cite{CGH2017}, the authors consider a similar multi-scale  decomposition of the Lagrange-to-Euler map \eqref{eq:CoMdef}. The authors also wish to pass to the limit using deterministic homogenisation theory. However, the same aforementioned issue arises, nevertheless, they apply the standard deterministic homogenisation theory. Thus, in \cite{CGH2017}, a gap remains to pass to the limit. Our approach can be taken to close this gap, and we explain this in Section \ref{sec:otherapproach}. The key to overcoming this difficulty is to first establish $\Xi^{\varepsilon}\rightarrow_{\bbP}\Xi$ in $C^\alpha([0,T] ; \operatorname{Diff}_{C^{n}}(\bbT^d))$ as $\varepsilon\downarrow 0$ for $\alpha<1/2$.  It is worth mentioning that our case is slightly more complicated due to the $\varepsilon$-dependence on $\bar{u}^{\varepsilon}$ (see Remark \ref{rem:why_gbar_is_random} for the reason). 

As originally discovered in \cite{kelly2017deterministic}, and elaborated in \cite{10.1214/21-AIHP1202, 10.1214/21-AIHP1203, chevyrev2019multiscale}, it is convenient to recast the homogenisation problem as an application of the stability results of rough differential equations (RDEs) with respect to its driving path and a weak invariance principle. To do this, we define the canonically lifted rough path $\bB^{\varepsilon}=(B^{\varepsilon}, \bbB^{\varepsilon}) \in \clC_g^{\alpha}([0,T]; \mathbb{R}^K)$  (see Definition \ref{defRP}) by
$$B^{\varepsilon}_t= \varepsilon \int_0^{t\varepsilon^{-2}}\lambda_s^{\varepsilon}\dd s, \quad \bbB^{\varepsilon}_{st}= \int_s^t \delta B^{\varepsilon}_{su} \otimes \dd B^{\varepsilon}_u .$$  
Then \eqref{eq:Xi_eps} can be  recast as a flow of rough diffeomorphisms  (see Definition \ref{eq:rde}) given by
\begin{equation}
\dd \Xi^{\varepsilon}_t(\bar{X})=\sum\limits_{k=1}^K\sigma_k(\Xi_t^{\varepsilon}(\bar{X}))\rmd \bB^{\varepsilon;k}_t\,,  \quad \Xi_0^{\varepsilon}(\bar{X})=\bar{X}\in \bbT^d\,. \label{eq:RDE Xi}
\end{equation}
In light of Theorem \ref{thm:rough_flow_map}, which states the stability of solutions and flows of RDEs in the rough path topology, the convergence of $\Xi^{\varepsilon}$ will follow from the convergence in law of $\bB^{\varepsilon}$. We introduce the requisite weak invariance principle assumption. 
\begin{assumption}[Ergodicity and Weak Invariance Principle]\label{asm:weak_invariance}
\mbox{} 
Let $\clM$ denote a smooth compact Riemannian manifold embedded in $\bbR^K$ and $h\in \mathfrak{X}_{C^2}(\clM)$ denote a $C^2$-vector field. Let $\phi: \bbR\rightarrow \operatorname{Diff}_{C^{2}}(\mcal{M})$ denote the solution flow of the ODE $\dot{\lambda}=h(\lambda)$.
\begin{enumerate}[(i)]
\item (Ergodicity)  $\phi$ has a closed attracting set $\Omega\subset \clM$ and there exists a $\phi$-invariant ergodic SRB measure $\bbP\in \mathcal{P}(\clM)$ supported on $\Omega$ such that the following centering condition holds
$$\int_{\Omega} \lambda \bbP(\dd \lambda)=0\,,$$
where the integrand is understood to take values in the embedding space $\bbR^K$.
\item (Weak invariance principle) $\bB^{\varepsilon}\rightarrow_\bbP \bB^{\wt{\Gamma}}$ as $\varepsilon\rightarrow 0$ in $\clC_g^{\alpha}([0,T]; \mathbb{R}^K)$, where $\bB^{\wt{\Gamma}}=(B, \bbB^{\wt{\Gamma}})\in \clC_g^{\alpha}([0,T]; \mathbb{R}^K)$, $\alpha \in [\frac13,\frac12)$, is the lift of a $K$-dimensional Brownian motion supported on the probability space $(\Omega, \clF:=\clB(\clM)\cap \Omega, \bbP)$ with covariance $\Sigma \in \operatorname{Sym}_+ (\bbR^K)$ and 
$$
\bbB_{st}^{\wt{\Gamma}}= \int_s^t \delta B_{su} \otimes \circ \dd B_u + \wt{\Gamma} (t-s)\,,
$$
for an anti-symmetric matrix  $\wt{\Gamma} \in \mathfrak{s} \mathfrak{o}(K)$. Moreover, $$\Sigma=\lim_{\varepsilon\rightarrow 0}\bbE_{\bbP}[B^{\varepsilon}_1\otimes B^{\varepsilon}_1] \quad \textnormal{and} \quad \wt{\Gamma}=\lim_{\varepsilon\rightarrow 0}\bbE_{\bbP}[\bbB^{\varepsilon}_{01}]\,.$$ 
\end{enumerate}
\end{assumption}
\begin{remark}
A detailed discussion of the rough weak invariance principle can be found in Sections 2.3.3 and Theorem 4.4 (i) of \cite{chevyrev2019multiscale}. In particular, it can be shown that if:
\begin{enumerate}[(i)]
    \item (Assumption 4.1 in \cite{chevyrev2019multiscale})
 $(B^{\varepsilon}, \mathbb{B}^{\varepsilon})\rightarrow (B, \mathbb{B})$ as $\varepsilon\rightarrow 0$ in the sense of finite-dimensional distributions with respect to  $\mathbb{P}$, where $B$ is a $K$-dimensional Brownian motion and $\mathbb{B}_{0t}=\int_0^tW_s\otimes \circ \rmd B_s + \Gamma t$ for some deterministic $\tilde{\Gamma}\in \mathfrak{s} \mathfrak{o}(K)$\,,
\item (Assumption 2.2 in \cite{kelly2017deterministic}) there exists $p > 1$ and $kappa\in (0,1)$ such that for all $u,v\in C^{\kappa}(\mathcal{M})$ with zero mean under $\bbP$, the following estimates are satisfied with $K = K(u,v,p) > 0$ for all $s,t\in [0,T]$,
\[  \bbE \bigg[ \left( \int_s^t u \circ \phi_s \dd s \right)^{2p} \bigg]^\frac{1}{2p} \leq K |t-s|^{\frac12}, \quad \bbE \bigg[ \left( \int_s^t\int_s^r u \circ \phi_r v \circ \phi_l \dd l \dd r \right)^{p} \bigg]^\frac{1}{p} \leq K |t-s| \,\,, \]
\end{enumerate}
and Assumption \ref{asm:weak_invariance}(i) holds, then Assumption \ref{asm:weak_invariance}(ii) holds. 
See, also, \cite[Theorem 9.1]{10.1214/14-AOP979}, which is the first paper to recognize the application of rough paths to the problem of deterministic homogenisation. In \cite[Theorem 1.1]{kelly2017deterministic} the authors were able to relax the mixing conditions and product structure of the noise term to extend earlier work on deterministic homogenisation (e.g., \cite{melbourne2011note} and \cite{melbourne2015correction})
\end{remark}

The following theorem is then a direct consequence of Theorems \ref{thm:rough_flow_map} and \ref{thm:RDE2SDE} in the appendix (see, also,  Theorem 5.5 of \cite{10.1214/21-AIHP1203}.)
\begin{theorem}[Convergence of $\Xi^{\varepsilon}$]\label{thm:fast_flow_covergence}
Let $T>0$ and Assumption \ref{asm:weak_invariance} hold. Let $W=\Sigma^{-1/2}B$, which is an $K$-dimensional Brownian motion on $(\Omega, \clF, \bbP)$ with identity covariance. Let $\bbF$ denote the filtration generated by $W$. Let  $n\ge 2$ and assume $\sigma \in \mathfrak{X}_{C^{n+2}}(\bbT^d)^K$.  Let $\Gamma := \sqrt{\Sigma}^{-1} \wt{\Gamma} \sqrt{\Sigma}^{-T} \in \mathfrak{s} \mathfrak{o}(K)$ and $\xi := \sigma \sqrt{\Sigma}$. 
Then $\Xi^{\varepsilon}\rightarrow_\bbP  \Xi$ in $C^\alpha( [0,T] ;\operatorname{Diff}_{C^{n}}(\bbT^d))$ as $\varepsilon\rightarrow 0$ for any $\alpha<1/2$, where $\Xi: \Omega \times [0,T]\rightarrow \operatorname{Diff}_{C^{n}}(\bbT^d))$ is the $\bbF$-adapted stochastic flow of diffeomorphisms such that for each $\bar{X}\in \bbT^d$,  $\Xi(\bar{X})$ is the unique strong $\bbF$-adapted solution of the SDE
\begin{align}\label{eq:Xi_limit}
\dd {\Xi}_t(\bar{X})=\frac{1}{2}\sum \limits_{k,l=1}^K\Gamma^{kl}\qv{\xi_k}{\xi_l}(\Xi_t(\bar{X})) \dd t +\sum\limits_{k=1}^K\xi_k(\Xi_t(\bar{X}))\circ \dd W_t^k\,, \quad 
\Xi_0(\bar{X}) =\bar{X}\,.
\end{align}
\end{theorem}
\begin{proof} 
By Theorem \ref{thm:rough_flow_map}, there exists a continuous (in the driving path) flow map $$\Phi \in \operatorname{Lip}_{loc}\left(\clC_g^{\alpha}([0,T]; \mathbb{R}^K);C^\alpha( [0,T] ;\operatorname{Diff}_{C^{n}}(\bbT^d))\right)$$ such that for all $\bZ\in \clC_g^{\alpha}([0,T]; \mathbb{R}^K)$ and $X\in \bbT^d$, $\Phi(X,\bZ)$ is the unique solution of the RDE
$$
\rmd \Phi_t(\bar{X}, \bZ) = \sum\limits_{k=1}^K\sigma_k(\Phi_t(\bar{X}, \bZ))\rmd \bZ_t\,, \quad \Phi_0(\bar{X}, \bZ)=\bar{X} \in \bbT^d\, .
$$
Assumption \ref{asm:weak_invariance}(ii) asserts that $\bB^\varepsilon \rightarrow_\bbP \bB^{\wt{\Gamma}}$ as $\varepsilon\rightarrow 0$, and hence $\Phi(\cdot, \bB^{\varepsilon}) \rightarrow_{\bbP} \Phi(\cdot, \bB^{\wt{\Gamma}})$ in $C^\alpha( [0,T] ;\operatorname{Diff}_{C^{n}}(\bbT^d))$ as $\varepsilon\rightarrow 0$ by the continuous mapping theorem \cite[Theorem 2.7]{billingsley1968convergence}.

Let $\bar{X}\in \bbT^d$ be an arbitrarily coordinate.  Since solutions of RDEs agree with solutions of ODEs if the driving path $\bZ$ is smooth, $\Xi^\varepsilon_t(\bar{X}) = \Phi_t(\bar{X}, \bB^\varepsilon)$ for all $t\in [0,T]$. It remains remains to verify that $\eta(\bar{X})=\Phi(\bar{X}, \bB^{\wt{\Gamma}})$ is the solution of the SDE \eqref{eq:Xi_limit}. Applying \cite[Theorem 2]{FRIZ20093236}, we find that $\eta(\bar{X})$ is also the solution of the RDE
\begin{equation*} 
\dd \eta_t(\bar{X})= \frac{1}{2}\sum_{k,l=1}^K\wt{\Gamma}^{kl} [\sigma_k, \sigma_l](\eta_t( \bar{X} ) ) \dd t +\sum\limits_{k=1}^K\sigma_k(\eta_t(\bar{X}) \rmd \bB^{k}_t\,,  \quad \Xi_0(\bar{X} )= \bar{X} \in \bbT^d\,,
\end{equation*}
where $\bB=(B, \bbB)$ is the Stratonovich lift of $B$. Using the identity
\[ \wt{\Gamma}^{kl} [\sigma_k, \sigma_l] = (\sqrt{\Sigma}^{-1})^{\alpha l} \wt{\Gamma}^{kl}(\sqrt{\Sigma}^{-1})^{\beta l} [\xi_\alpha, \xi_\beta] =: \Gamma^{\alpha \beta} [\xi_\alpha, \xi_\beta] \]
and the standard relation between integrals against correlated Brownian motion and uncorrelated Brownian motion, we find 
\begin{equation*} 
\dd \eta_t(\bar{X})= \frac{1}{2}\sum_{k,l=1}^K\Gamma^{kl} [\xi_k, \xi_l](\eta_t( \bar{X} ) ) \dd t +\sum\limits_{k=1}^K\xi_k(\eta_t(\bar{X}) )\rmd \bW^{k}_t\,,  \quad \Xi_0(\bar{X} )= \bar{X} \in \bbT^d\,,
\end{equation*}
where $\bW=(W, \bbW)$ is the Stratonovich lift of $W$. We then complete the proof by invoking Theorem \ref{thm:RDE2SDE}.
\end{proof}

So far, we have made no structural assumption on $\bar{g}^{\varepsilon}$, or equivalently, on $\bar{u}^{\varepsilon} = \dot{\bar{g}}^\varepsilon \bar{g}^{\varepsilon; -1}$. Recall that, in general, $\bar{u}^{\varepsilon}$ depends on the initial condition of the fast dynamics random $\omega\in \Omega$. In this section, we make a general assumption that allows us to pass to the limit.

\begin{assumption}[Mean vector field assumption]\label{asm:u_bar_eps}
For a given $n\in\bbN$, there exists a \[V \in C\left(\clC_g^\alpha ([0,T]; \bbR^K ) ; C([0,T]; \fkX_{C^n}(\bbT^d)\right)\] such that $ \bar{u}^\varepsilon_t=V_t(\cdot, \bB^\varepsilon)$.
\end{assumption}

 Assumption \ref{asm:u_bar_eps} holds if $\bar{u}=\bar{u}^{\varepsilon}$ is independent of $\varepsilon$ and $\omega$ and $\bar{u}\in C([0,T]; \fkX_{C^n}(\bbT^d))$. In Section \ref{sec:VP1}, we propose a variational principle that serves as a closure model for
 $\Xi_{*}\bar{u}$ in \eqref{eq:g_homog_limit}
 that in turn closes the dynamics of $\bar{u}$. 
 The same variational principle can also be applied to $g^{\varepsilon}$ satisfying \eqref{eq:g_eps} for every $\varepsilon$. In the setting of an ideal perfect fluid (Section \ref{sec:Euler}), using \cite[Theorem 3.7]{CRISAN2022109632}, it is possible to show that  $\bar{u}^{\varepsilon}$ satisfies Assumption \ref{asm:u_bar_eps}, which establishes consistency of the assumption and the closure model. See, also, the discussion in Remark \ref{rem:why_gbar_is_random}.

Under Assumption \ref{asm:u_bar_eps}, using standard ODE flow stability estimates, one can show that there exists 
\[ \Psi\in C\left(\clC_g^\alpha ([0,T]; \bbR^K ) ; C^1([0,T]; \operatorname{Diff}_{C^n}(\bbT^d)\right)\] such that $ \bar{g}^\varepsilon_t(X)=\Psi_t(X, \bB^\varepsilon)$ for all $t\in [0,T]$ and $X\in \bbT^d$. Define $\bar{g}=\Psi(\cdot, \bB^{\wt{\Gamma}})\in C^1([0,T]; \operatorname{Diff}_{C^n}(\bbT^d))$. From this observation and Lemma \ref{lem:continuity_comp_diffo}, we deduce by the continuous mapping theorem (see \cite[Theorem 2.7]{billingsley1968convergence}) that $\bar{g}^{\varepsilon} \rightarrow_\bbP \, \bar{g}$ as $\varepsilon\rightarrow 0$ in $C^1( [0,T] ;\operatorname{Diff}_{C^n}(\bbT^d))$.

The following is the main theorem of the section, which shows that $g^{\varepsilon}$ converges to the flow of the ansatz \eqref{eq:SALTansatz} with $u_t=\Xi_{t*}\bar{u}_t+\frac12 \sum\limits_{k,l=1}^K{\Gamma}^{kl}[\xi_{k},\xi_l]$.

\begin{theorem}[Convergence of composition of maps]\label{thm:convergence}
Let Assumption \ref{asm:u_bar_eps} and the assumptions of Theorem \ref{thm:fast_flow_covergence} hold with $n \geq 2$. Then  $g^{\varepsilon}\rightarrow_\bbP g$  in $C([0,T] ; \operatorname{Diff}_{C^{n-2}}(\bbT^d))$ as $\varepsilon\rightarrow 0$, where $g: \Omega\times [0,T]\rightarrow \operatorname{Diff}_{C^{n}}(\bbT^d)$ is the $\bbF$-adapted stochastic flow of diffeomorphisms such that  for each $X\in \bbT^d$, $g(X)$ is the unique strong solution of the SDE 
\begin{align}\label{eq:g_homog_limit}
    \dd g_t(X)=  \left(\Xi_{t*}\bar{u}_t(g_t(X))+\frac12 \sum\limits_{k,l=1}^K{\Gamma}^{kl}[\xi_{k},\xi_l](g_t(X))\right)\dd t  + \sum\limits_{k=1}^K\xi_k(g_t(X))\circ \dd W_t^k\,,
\end{align}
where  $\bar{u}=V_t(\cdot, \bB)$. Moreover, $g = \Xi\circ \bar{g}$.
\end{theorem}
\begin{proof} 
Noting that $\Xi_{t*}\bar{u}: \Omega \times [0,T]\rightarrow \fkX_{C^{n-1}}(\bbT^d)$, by \cite[Theorem 3.4]{kunita1990stochastic}, there exists an $\bbF$-adapted stochastic flow of diffeomorphisms $g: \Omega\times [0,T]\rightarrow \operatorname{Diff}_{C^{n-2}}(\bbT^d)$ such that for each $X\in \bbT^d$, $g(X)$ is the unique strong solution of the SDE \eqref{eq:g_homog_limit}.

As in the proof of Theorem \ref{thm:fast_flow_covergence}, there exists \[\Phi\in C\left(\clC_g^{\alpha}([0,T]; \mathbb{R}^K); C^\alpha( [0,T] ;\operatorname{Diff}_{C^{n}}(\bbT^d))\right)\] such that 
$\Xi^{\varepsilon}=\Phi(\cdot, \bB^{\varepsilon})\,.$ Likewise from Assumption \ref{asm:u_bar_eps}, we have a map 
\[ \Psi\in C\left(\clC_g^\alpha ([0,T]; \bbR^K ) ; C^1([0,T]; \operatorname{Diff}_{C^n}(\bbT^d)\right) \]
such that $ \bar{g}^\varepsilon(X)=\Psi(X, \bB^\varepsilon)$ for all $t\in [0,T]$ and $X\in \bbT^d$. 

By considering continuity of the inclusion map, one can define new continuous maps with co-domains $C([0,T]; \operatorname{Diff}_{C^n}(\bbT^d))$ and $C([0,T]; \operatorname{Diff}_{C^{n-2}}(\bbT^d))$, which we abusively also denote by $\Phi, \Psi$ respectively. 

Owing to Lemma \ref{lem:continuity_comp_diffo}, the composition map
\begin{equation*}
\begin{aligned}
\fkC : C([0,T]; \operatorname{Diff}_{C^{n}}(\bbT^d)) \times C([0,T]; \operatorname{Diff}_{C^{n-2}}(\bbT^d)) &\rightarrow C([0,T]; \operatorname{Diff}_{C^{n-2}}(\bbT^d))\\
(g, h) &\mapsto g \circ h
\end{aligned}
\end{equation*}
is continuous and the map
\begin{align*}
\Phi\times \Psi: \clC_g^{\alpha}([0,T]; \mathbb{R}^K) &\rightarrow C( [0,T] ;\operatorname{Diff}_{C^{n}}(\bbT^d)) \times C( [0,T] ;\operatorname{Diff}_{C^{n-2}}(\bbT^d))\\
\bZ &\mapsto (\Phi(\bZ), \Psi(\bZ))
\end{align*}
is also continuous.  Together, this implies the continuity of the map $$\Upsilon= \fkC \circ (\Phi\times \Psi): \clC_g^{\alpha}([0,T]; \mathbb{R}^K)\rightarrow C( [0,T] ;\operatorname{Diff}_{C^{n-2}}(\bbT^d))$$ defined by  
$$
\Upsilon_t(X, \bZ)=\Phi_t(\Psi_t(X, \bZ), \bZ)\,, \quad (t, \bZ, X)\in [0,T]\times  \clC_g^{\alpha}([0,T]; \mathbb{R}^K) \times \bbT^d\,.
$$
Since $g^{\varepsilon}_t(X)=\Upsilon_t(X, \bB^{\varepsilon})$, combining  Assumption \ref{asm:weak_invariance} with the continuous mapping theorem for metric space valued random variables \cite[Theorem 2.7]{billingsley1968convergence}, we find that $g^\varepsilon \rightarrow_{\bbP} g$ in $C( [0,T] ;\operatorname{Diff}_{C^{n-2}}(\bbT^d))$, where 
$$
g_t(X)=\Upsilon_t(X, \bB^{\wt{\Gamma}}\,)=\Phi_t(\Psi_t(X, \bB^{\wt{\Gamma}}\,), \bB^{\wt{\Gamma}}\,)\,.
$$
In the proof of Theorem \ref{thm:fast_flow_covergence}, we showed that $\Xi=\Psi_t(X, \bB^{\wt{\Gamma}}\,)$. Thus, $g= \Xi \circ \bar{g}$. We deduce that $\mathbb{P}$-a.s. for all $t\in [0,T]$,  $g_t\in \operatorname{Diff}_{C^{n}}(\bbT^d)$ since $g_t= \Xi_t \circ \bar{g}_t$ it is the composition of two maps in  $\operatorname{Diff}_{C^{n}}(\bbT^d)$.
Applying the Stratonovich It\^o-Wentzel formula \cite{ventzel1965equations}, we find \eqref{eq:g_homog_limit}.  

\end{proof}

\begin{remark}[General configuration spaces] In more general domains $\mcal{D}$, the convergence of  $g^{\varepsilon}\rightarrow g$ and regularity of $\xi$ can be related to maps $U \subset \bbR^d \rightarrow V \subset \bbR^d$ through charts \cite{banyaga2013structure}, which can also be used to define an equivalent topology to \cite[Exercise 11.17]{friz2010multidimensional} on the rough flows. The full theory of rough differential equations has also been extended to manifolds, see \cite{https://doi.org/10.1112/jlms.12585}.
\end{remark}

\begin{remark}[Stronger notions of convergence] The convergence in the space $C([0,T] ; \operatorname{Diff}_{C^{n-2}}(\bbT^d))$ can, with further work, be improved to $C^\alpha([0,T] ; \operatorname{Diff}_{C^{n-2}}(\bbT^d))$. It is important to define this space carefully to be a separable metric space, which is crucial in order to apply the continuous mapping theorem. Note, that even the space $C^\alpha([0,1], \bbR)$ fails to be separable, but a distinguished subspace $C^{0,\alpha}([0,1], \bbR)$, definable as the closure of $C^\infty([0,1]; \bbR)$ under the $\alpha$-H\"older semi-norm, is separable. It is established in \cite[Section 13.3.1]{friz2010multidimensional} that almost sure realizations of Brownian motion belong to this space, and analogues of Theorem \ref{thm:rough_flow_map} and Lemma \ref{lem:continuity_comp_diffo} may be proven within the separable subspace of $C^\alpha([0,T] ; \operatorname{Diff}_{C^{n-2}}(\bbT^d))$.
\end{remark}

\begin{remark}[Non-product case] It is possible to extend the above analysis under the more general assumption that 
\begin{align*}
\dot{\Xi}^{\varepsilon}_t(\bar{X})&= \varepsilon^{-1}\sum\limits_{k=1}^K\xi_k(\Xi_t^{\varepsilon}(\bar{X}),\lambda_t^{\varepsilon}), \quad 
\Xi^{\varepsilon}_0(\bar{X})=\bar{X}\in \bbT^d\,,\\
\dot{\bar{g}}_t^{\varepsilon}(\omega, X) &= \bar{u}_t^{\varepsilon}(\omega,\bar{g}_t^{\varepsilon}(\omega, X)), \quad \bar{g}_0^{\varepsilon}(X)=X \in \bbT^d\,,\\[8pt]
\dot{\lambda}_t^{\varepsilon} &= \varepsilon^{-2} h(\lambda_t), \;\; \lambda_0=\omega \in \Omega\,.
\end{align*}
See, e.g.,  \cite{10.1214/21-AIHP1203, kelly2017deterministic, 10.1214/21-AIHP1202, melbourne2011note, melbourne2015correction} and particularly Section 5 of \cite{chevyrev2019multiscale}.
\end{remark}

\section{Mean flow closures via variational principles } \label{sec:vp}
The goal of this section is to derive closures for the dynamics of the mean flow velocity field $\bar{u}_t$ appearing in \eqref{eq:g_homog_limit}, and explain the modelling procedures that motivated the dynamics. We consider two closures derived from variants of the Euler--Poincar\'e variational principle with advected quantities developed in \cite{HMR1998}. The first is a stochastic closure of $\bar{u}_t$ that is equivalent to the SALT variational principle \cite{Holm2015}, albeit with a new formulation that reflects the decomposition of the flow map $g_t$. The second is a deterministic closure of $\bar{u}_t$ that resembles the Generalised Lagrangian Mean (GLM) theory \cite{andrews_mcintyre_1978, GH1996}. The stochastic closure is presented in Section \ref{sec:VP1} and the deterministic closure is presented in Section \ref{sec:det vp}. In preparation for the construction of these closures, we will recall the essential aspects of the previous section and introduce the geometric mechanics language necessary to discuss the constrained variations that are used in the variational principles in Sections \ref{sec:VP1} and \ref{sec:det vp}.

\subsection{Geometric setting and variations}
Let $(\clD, \bg)$ denote a smooth, compact, connected, oriented $d$-dimensional boundaryless Riemannian manifold with metric $\mathbf{g}$. Let $\mu_{\mathbf{g}} \in \Lambda^d(\clD)$ the associated volume form expressed in local coordinates $(x^1,\ldots,x^d)$ as
\begin{align*}
    \mu_{\mathbf{g}} = \sqrt{\left|\det g_{ij}\right|}\,dx^1\wedge\ldots \wedge dx^d\,.
\end{align*}
Let $T>0$ and $(\Omega,\mathcal{F},\bbP)$ denote a probability space supporting an $K$-dimensional Brownian motion $W=(W^1,\ldots, W^K)$ with identity covariance. Let $\bbF$ denote the filtration generated by $W$.

In the previous section, for $\mcal{D} = \mathbb{T}^d$, we have shown the homogenisation of the flow of diffeomorphisms $g^\varepsilon: \Omega \times [0,T]\rightarrow \operatorname{Diff}_{C^n}(\mcal{D})$ with fast + slow decomposition in the form of equation \eqref{eq:g_eps}. This resulted in a stochastic diffeomorphism $g : \Omega \times [0,T] \rightarrow \operatorname{Diff}_{C^n}(\mcal{D})$ as proven in Theorem \ref{thm:convergence}.  The diffeomorphism $g$ is defined through the composition $g = \Xi \circ \bar{g}$, where the constituent flow of diffeomorphisms $\Xi : \Omega \times [0,T] \rightarrow \operatorname{Diff}_{C^n}(\mcal{D})$ and $\bar{g} : \Omega \times [0,T] \rightarrow \operatorname{Diff}_{C^n}(\mcal{D})$ satisfy the following Stratonovich SDEs
\begin{align}
    \begin{split}
        \dd \Xi_t &= \left(\sum\limits_{k=1}^K\xi_k \circ \dW^k_t + \frac{1}{2}\sum \limits_{k,l=1}^K\Gamma^{kl}\left[\xi_k\,,\,\xi_l\right] \dd t\right) \Xi_t\,, \qquad \dd \bar{g}_t = \bar{u}_t \bar{g}_t \,\dt\,. 
    \end{split}\label{eqn:d Xi d gbar def}
\end{align}
Here, $\xi \in \mathfrak{X}_{C^{n+2}}(\mcal{D})^K$ is a collection of $K$ prescribed time-independent vector fields and $\bar{u} : \Omega \times [0,T] \rightarrow \mathfrak{X}_{C^n}(\mcal{D})$ is an $\mathbb{F}$-adapted $C^n$-vector field.
The tangent lifted right action of $\Xi_t, \bar{g}_t \in \operatorname{Diff}_{C^n}(\mcal{D})$ on vector fields $v\in \mathfrak{X}_{C^n}(\mcal{D})$ is denoted by concatenation from the right and it is equivalent to composition; that is, for $\phi \in \operatorname{Diff}(\clD)$ and $X \in \mathfrak{X}$, $X\phi = T_eR_{\phi} X$ where $R_{(\cdot)}: $ is composition from the right. Lastly, $\qv{\cdot}{\cdot}:\mathfrak{X}_{C^n}(\mcal{D})\times \mathfrak{X}_{C^n}(\mcal{D}) \rightarrow \mathfrak{X}_{C^{n-1}}(\mcal{D})$ denotes the commutator of vector fields. The stochastic flow of diffeomorphisms $g$ satisfies the following Stratonovich SDE 
\begin{align}
    \dd g_t= \left(\Xi_{t*}\bar{u}_t\dd t + \frac{1}{2}\sum \limits_{k,l=1}^K\Gamma^{kl}\left[\xi_k\,,\,\xi_l\right]\, \dd t + \sum \limits_{k=1}^K\xi_k \circ \dd W_t^k \right)g_t\,, \label{eqn:homoged diffeo}
\end{align}
which is equation \eqref{eq:g_homog_limit} expressed in the geometric mechanics language. 

In the remainder of the section, we will derive closure dynamics of $\bar{u}_t$. For simplicity, we assume the stochastic flow maps of the fluid particle trajectories are $C^\infty$-smooth diffeomorphisms and the vector fields are also $C^\infty$-smooth. We will use the shorthand $\operatorname{Diff}(\mcal{D})$ for $\operatorname{Diff}_{C^\infty}(\mcal{D})$ and $\mathfrak{X}(\mcal{D})$ for $\mathfrak{X}_{C^\infty}(\mcal{D})$, such that $g,\, \bar{g},\, \Xi: \Omega \times [0,T] \rightarrow \operatorname{Diff}(\mcal{D})$ and $\bar{u},\, \xi_k: \Omega \times [0,T] \rightarrow \mathfrak{X}(\mcal{D})$ for all $k=1,\ldots, K$.

The geometric notation to denote the actions of $\operatorname{Diff}(\mcal{D})$ and $\mathfrak{X}(\mcal{D})$ on $\mathfrak{X}(\mcal{D})$ and $\mathfrak{X}^*(\mcal{D})$ can be summarised as follows. The (geometric) dual space $\mathfrak{X}^*(\mcal{D})$ is identified as the space of one-form densities, i.e., $\mathfrak{X}^*(\mcal{D}) = \Lambda^1(\mcal{D})\otimes\Lambda^d(\mcal{D})$ via a (weak) duality pairing $\scp{\cdot}{\cdot} : \mathfrak{X}^*(\mcal{D})\times \mathfrak{X}(\mcal{D}) \rightarrow \mathbb{R}$ that is defined by
\begin{align*}
    \scp{m}{u} := \int_{\clD} (u \intprod \alpha) \rho\,, \quad \textrm{where} \quad m = \alpha\otimes \rho \in \mathfrak{X}^*(\clD)\,, \quad \alpha \in \Lambda^1(\clD)\,,\quad \rho \in \Lambda^d(\clD)\,,\quad u \in \mathfrak{X}(\mcal{D})\,, 
\end{align*}
where $\intprod$ is the interior product. Letting $\phi \in \operatorname{Diff}(\mcal{D})$, $u, v \in \mathfrak{X}(\mcal{D})$ and $m = \alpha \otimes \rho \in \mathfrak{X}^*(\mcal{D})$ where $\alpha \in \Lambda^1(\mcal{D})$ and $\rho\in \Lambda^d(\mcal{D})$, we have the following adjoint actions
\begin{align*}
    \begin{split}
        &\Ad: \operatorname{Diff}(\mcal{D})\times \mathfrak{X}(\mcal{D}) \rightarrow \mathfrak{X}(\mcal{D}) \,, \qquad \Ad_\phi u := \phi_* u\,,\\
        &\ad: \mathfrak{X}(\mcal{D}) \times \mathfrak{X}(\mcal{D}) \rightarrow \mathfrak{X}(\mcal{D})\,,\qquad \ad_u v := -\qv{u}{v}\,,
    \end{split}
\end{align*}
and coadjoint actions defined through the duality pairing\footnote{In certain literature, the convention to define the dual of $\Ad_\phi$ by $\Ad^*_{\phi^{-1}}$ is used. In that case, we have the representation $\Ad^*_\phi m = \phi_* m$. 
},
\begin{align*}
    \begin{split}
        &\Ad^* : \operatorname{Diff}(\mcal{D})\times \mathfrak{X}^*(\mcal{D}) \rightarrow \mathfrak{X}^*(\mcal{D}) \,, \quad \scp{\Ad_{\phi} u}{m} = \scp{u}{\Ad^*_{\phi}m} \,,\\
        &\ad^* : \mathfrak{X}^(\mcal{D})\times \mathfrak{X}^*(\mcal{D}) \rightarrow \mathfrak{X}^*(\mcal{D}) \,, \quad \scp{\ad_{u} v}{m} = \scp{v}{\ad^*_{u}m} \,.
    \end{split}
\end{align*}
Under these definitions, $\Ad^*_\phi m$ is naturally identified as the pullback $\phi^* m$ since
\begin{align*}
    \scp{m}{\Ad_{\phi}u} = \int_{\clD}({\phi_*u}\intprod \alpha)\rho = \int_{\clD}({u}\intprod \phi^*\alpha)\phi^*\rho = \scp{\phi^*\alpha\otimes \phi^*\rho}{u} = \scp{\phi^*(\alpha\otimes \rho)}{u} = \scp{\phi^*m}{u}\,,
\end{align*}
where we have used the natural properties of pullback on tensor products \cite{AMR1988}. It can be shown that the coadjoint action
$\ad^*_u m$ is given by the Lie derivative action, see e.g., \cite{MR1999}. Recall that the Lie derivative $\mcal{L}:\mathfrak{X}(\mcal{D})\rightarrow \mathcal{L}(\Lambda^k(\mcal{D}); \Lambda^k(\mcal{D}))$ is defined on differential $k$-forms by $\mcal{L}_X \kappa = \frac{d}{dt}\big|_{t=0} \psi_t^*\kappa$, where $\psi$ is the flow associated with vector field $X$. Owning to the natural property of Lie derivative over tensor products \cite{AMR1988}, we have 
\begin{align*}
    \ad^*_u m = \mcal{L}_u m = \mcal{L}_u(\alpha\otimes \rho) = \mcal{L}_u \alpha\otimes \rho + \alpha\otimes \mcal{L}_u \rho = \mcal{L}_u \alpha \otimes \rho + \operatorname{div}_{\mu_\mathbf{g}}\left(u\rho\right)\,.
\end{align*}

Fluid particles also possess intrinsic physical quantities such as mass volume and heat, which, in ideal fluid dynamics, are modelled as advected quantities. Advected quantities  lie in a dual (infinite-dimensional) vector space $V^*$ of a vector space $V$; that is, there is a (weak) duality pairing $\scp{\cdot}{\cdot}_{V\times V^*}:V\times V^* \rightarrow \mathbb{R}$ which allows us to view $V^*$ as a subset of the analytic dual of $V$. In this work, we take $V^* = \Lambda^d(\clD)\oplus \left(\oplus^N_{i=1} \Lambda^{k_i}(\clD)\right)$ for $k_i \in \{0,\ldots,d\}$ such that for $a = \{a^{(i)}\}^N_{i=0}\in V^*$ and $b = \{b^{(i)}\}^N_{i=0}\in V = \Lambda^0(\clD)\oplus \left(\oplus^N_{i=1} \Lambda^{d-k_i}(\clD)\right)$, the pairing is $\scp{\cdot}{\cdot}_{V\times V^*}$ is defined by
\begin{align*}
    \scp{b}{a}_{V\times V^*} = \scp{b^{(0)}}{a^{(0)}}_{\Lambda^0(\mcal{D})\times \Lambda^d(\mcal{D})} + \sum_{i=1}^N \scp{b^{(i)}}{a^{(i)}}_{\Lambda^{d-k_i}(\clD)\times \Lambda^{k_i}(\mcal{D})}\,,
\end{align*}
where $\scp{\beta}{\alpha}_{\Lambda^{d-k}(\clD)\times\Lambda^k(\clD)} := \int_{\clD}\beta\wedge\alpha$ for $\beta \in \Lambda^{d-k}(\clD)$ and $\alpha \in\ \Lambda^k(\clD)$ is the duality pairing between $\Lambda^{d-k}(\clD)$ and its dual $\Lambda^{k}(\clD)$. For applications to complex fluids, $V^*$ can additionally contain Lie-algebra valued $k$-forms and tensors \cite{GBR09}. 
The right actions of $\operatorname{Diff}(\mcal{D})$ and $\mathfrak{X}(\mcal{D})$ on $V^*$ are the pullback and Lie derivative, respectively, both of which we denote via concatenations on the right. We additionally define the operator $\diamond : V\times V^* \rightarrow \mathfrak{X}(\mcal{D})$ using the duality pairings on different spaces, let $u \in \mathfrak{X}(\clD)$, $b \in V$ and $a \in V^*$ defined as before,
\begin{align*}
    \scp{-b\diamond a}{u}_{\mathfrak{X}(\mcal{D})\times \mathfrak{X}^*(\mcal{D})} = \scp{b}{a u}_{V\times V^*} &= \scp{b}{\mathcal{L}_u a}_{V\times V^*} = \sum_{i=0}^N\scp{b^{(i)}}{\mathcal{L}_u a^{(i)}}_{V\times V^*}\,.
\end{align*}
In fact, the $\diamond$ operator is the cotangent lift momentum map of the cotangent bundle $T^*V \simeq V\times V^*$ induced from the right action of $\operatorname{Diff}(\mcal{D})$ by representation. Since the type of duality pairing used in self evident from its arguments, we will drop the subscripts specifying the spaces which the pairing is defined on in subsequent calculations. For an initial condition $a_0 \in V^*$, we define two advected quantities $a, \bar{a} :\Omega \times [0,T] \rightarrow V^*$ via right action by $\bar{g}$ and $g$ to have $\bar{a}_{t} = a_0 \bar{g}^{-1}_{t} = \bar{g}_{t*}a_0$ and $ a_{t} = a_0 g^{-1}_{t} = g_{t*}a_0$. Using the Kunita It\^o--Wentzel formula \cite{de2020implications}, these advected quantities satisfy the Stratonovich SDE 
\begin{align}
\begin{split}
    &\dd \bar{a}_{t} + \mcal{L}_{\bar{u}_{t}}\bar{a}_{t}\,\dt = 0 \,, \\
    &\dd {a}_{t} + \mcal{L}_{u_{t}} {a}_{t}\,\dt + \sum\limits_{k=1}^K\mcal{L}_{\xi_k} {a}_{t}\circ \dW^k_t + \frac{1}{2}\sum\limits_{k,l=1}^K\mcal{L}_{\Gamma^{kl}\qv{\xi_k}{\xi_l}} {a}_{t}\,\dt = 0\,,
\end{split}\label{eq:mean and full adv quantities}
\end{align}
for the initial condition $\bar{a}_t|_{t=0} = {a}_t|_{t=0} = a_0$. Together, fluid velocity and advected quantities form the configuration space of the fluid flow, and it is given by $\mathfrak{X}(\clD) \times V^*$. In what follows, we will refer to the variables $(\bar{u}, \bar{a}):\Omega\times[0,T]\rightarrow \mathfrak{X}(\clD) \times V^*$ and $(u, a) : \Omega\times[0,T]\rightarrow \mathfrak{X}(\clD) \times V^*$ as the mean fluid variables and fluid variables without prefixes, respectively.

To define the variational principle, we need a space of time-parameterised curves and the corresponding set of variations. In our context, we define two sets of stochastic time-parameterised curves taking values in $\operatorname{Diff}(\mcal{D})$, and then prescribe their variations.  

Let $\mathscr{S}(\mathfrak{X}(\clD))$ denote the space of $\mathfrak{X}(\clD)$-valued measurable stochastic processes adapted to the filtration $\bbF$. For any $\bar{u}\in \mathscr{S}(\mathfrak{X}(\clD))$, there exists a stochastic curve of diffeomorphisms $\bar{g}: \Omega \times [0,T] \rightarrow \operatorname{Diff}(\mcal{D})$  adapted to the filtration $\bbF$ satisfying $ \frac{d}{dt}\bar{g}_t = \bar{u}_t (\bar{g}_t)$. We collect all such stochastic curves of diffeomorphisms in a set:
\begin{align*}
    \bar{\mathscr{S}}(\operatorname{Diff}(\mcal{D})) = \left\{\bar{g} : \Omega \times [0,T] \rightarrow \operatorname{Diff}(\mcal{D})\, | \,\, \exists \,\, \bar{u} \in \mathscr{S}(\mathfrak{X}(\clD))\,\, \textnormal{s.t.} \;\; \frac{d}{dt}\bar{g}_t = \bar{u}_t (\bar{g}_t) \right\}\,.
\end{align*}
Let $\Xi : \Omega  \times [0,T] \rightarrow \operatorname{Diff}(\mcal{D})$ be the stochastic flow map of the SDE \eqref{eqn:d Xi d gbar def}. We define the set $\mathscr{S}^{\Xi}(\operatorname{Diff}(\mcal{D}))$ as the image of right translation by the stochastic flow $\Xi$ of the set $\bar{\mathscr{S}}(\operatorname{Diff}(\mcal{D}))$:
\begin{align*}
    \mathscr{S}^{\Xi}(\operatorname{Diff}(\mcal{D})) = \{g : \Omega \times [0,T] \rightarrow \operatorname{Diff}(\mcal{D}) \,|\,\, \exists \,\, \bar{g} \in \bar{\mathscr{S}}(\operatorname{Diff}(\mcal{D})\,\, \textnormal{s.t.}  \;\;  g = \Xi \bar{g}\} \,.
\end{align*}
The elements of $\mathscr{S}^{\Xi}(\operatorname{Diff}(\mcal{D}))$ are $C([0,T]; \operatorname{Diff}(\mcal{D}))$-valued stochastic process adapted to the filtration $\bbF$, which take the form \eqref{eqn:homoged diffeo}.

For any $g\in \mathscr{S}^{\Xi}(\operatorname{Diff}(\mcal{D}))$, we construct $\epsilon$-dependent deformations of $g$, $\epsilon \in \mathbb{R}$, denoted by $\wt{g}_{t, \epsilon} \in \mathscr{S}^{\Xi}(\operatorname{Diff}(\mcal{D}))$ following \cite{ACC2014, CCR2023}. Let $v: \Omega \rightarrow C^1([0, T]; \mathfrak{X}(\mcal{D}))$ be an $\mathbb{F}$-adapted stochastic process satisfying $v_0 = v_T = 0$, and for all $\epsilon \in [0,1]$, let $e_{\epsilon}: \Omega \rightarrow C^1([0, T]; \operatorname{Diff}(\mcal{D}))$ be the solution of the random ODE on $\operatorname{Diff}(\mcal{D})$ given by
\begin{align}
    \begin{split}
        \frac{d}{dt} e_{t, \epsilon} = \epsilon \frac{d}{dt} v_t e_{t, \epsilon}\,, \qquad e_{\epsilon, 0} = e\,,
    \end{split}\label{eqn:d e def}
\end{align}
where $e$ is the identity diffeomorphism. A direct consequence of this definition is that  $e_{0, t} = e$ for all $t \in [0, T]$. 
Furthermore, by \cite[Lemma 3.1]{ACC2014}, the $\epsilon$ derivatives of $e_{t, \epsilon}$ and $e^{-1}_{\epsilon, t}$ take the form
\begin{align}
    \frac{d}{d\epsilon} \Big|_{\epsilon = 0} e_{t, \epsilon} = v_t\,,\qquad \frac{d}{d\epsilon} \Big|_{\epsilon = 0} e^{-1}_{ \epsilon, t} = -v_t\,. \label{eq:eps d e def}
\end{align}
Using $e_{t, \epsilon}$, we define perturbations of $g\in \mathscr{S}^{\Xi}(\operatorname{Diff}(\mcal{D}))$ by 
\begin{align}
    \wt{g}_{t, \epsilon} := \Xi_t e_{t, \epsilon} \bar{g}_t \in \mathscr{S}^{\Xi}(\operatorname{Diff}(\mcal{D})) \,. \label{eq:perturbed g}
\end{align} 
The perturbed stochastic flow of diffeomorphisms $\wt{g}_{t, \epsilon}$ is defined as an $\epsilon$-dependent perturbation constructed to vary $\bar{g}_t$ following the flow generated by the random flow of diffeomorphisms $e_{t, \epsilon}$, whilst keeping $\Xi_t$ unchanged. This way, $\wt{g}$ retains the composition structure as a right translation by the stochastic flow $\Xi$ and it remains in the set $\mathscr{S}^{\Xi}(\operatorname{Diff}(\mcal{D}))$. 
To see that, one compute the stochastic time differential of $e_{t, \epsilon}\bar{g}_t \in \bar{\mathscr{S}}(\operatorname{Diff}(\clD))$ to have 
\begin{align}
    \dd (e_{t, \epsilon}\bar{g}_t) = \left(\epsilon \frac{d}{dt}v_t  e_{t, \epsilon} + \Ad_{e_{t,\epsilon}} \bar{u}_t e_{t, \epsilon}\right) \bar{g}_t\,,
\end{align}
and note that $\epsilon \frac{d}{dt}v_t + \Ad_{e_{t,\epsilon}} \bar{u}_t$ is $\bbF$-adapted. Thus, it is clear that $\wt{g}_{t, \epsilon} := \Xi_t e_{t, \epsilon} \bar{g}_t \in\mathscr{S}^\Xi(\operatorname{Diff}(\clD))$ due to its right translation of $\Xi$. The stochastic vector fields generated by the flow $\wt{g}$ can be computed using equations \eqref{eqn:d Xi d gbar def} and \eqref{eqn:d e def} to see that $\wt{g}_{t, \epsilon}$ satisfies the following Stratonovich SDE,
\begin{align}
    \begin{split}
        \dd \wt{g}_{t, \epsilon} &= \dd \Xi_t e_{t, \epsilon} \bar{g}_t + \Xi_t \frac{d}{dt}e_{t, \epsilon} \bar{g}_t\,\dt + \Xi_t e_{t, \epsilon} \frac{d}{dt} \bar{g}_t\,\dt \\ 
        &= \left(\dd \Xi_t \Xi_t^{-1}\right)\wt{g}_{t, \epsilon} + \Ad_{\Xi_t}\left(\frac{d}{dt}e_{t, \epsilon}e_{t, \epsilon}^{-1}\right)\wt{g}_{t, \epsilon}\,\dt + \Ad_{\Xi_t }\Ad_{ e_{t, \epsilon}}\left(\frac{d}{dt}\bar{g}_{t} \bar{g}^{-1}_{t}\right)\wt{g}_{t, \epsilon}\,\dt \\ 
        &= \left(\sum\limits_{k=1}^K\xi_k \circ \dW^k_t + \frac{1}{2}\sum \limits_{k,l=1}^K\Gamma^{kl}\qv{\xi_k}{\xi_l} \dd t\right)\wt{g}_{t, \epsilon} + \epsilon \left(\Ad_{\Xi_t} \frac{d}{dt}v \right) \wt{g}_{t, \epsilon}\,\dt \\
        & \qquad \qquad + \left(\Ad_{\Xi_t}\Ad_{e_{t, \epsilon}} \bar{u}_{t} \right)\wt{g}_{t, \epsilon} \,\dt\\
        &= \left(u_{t,\epsilon}\,\dt + \frac{1}{2}\sum \limits_{k,l=1}^K\Gamma^{kl}\qv{\xi_k}{\xi_l}\,\dt + \sum\limits_{k=1}^K\xi_k\circ \dW^k_t\right)\wt{g}_{t, \epsilon}\,, 
    \end{split}\label{eqn:perturbed dg}
\end{align}
where the perturbed drift vector field $u_{t, \epsilon} \in \mathfrak{X}(\mcal{D})$ is defined by
\begin{align*}
    u_{t,\epsilon} := \Ad_{\Xi_t} \left(\epsilon \frac{d}{dt}v_t + \Ad_{e_{t, \epsilon}} \bar{u}_{t} \right)\,.    
\end{align*}
Using the $\epsilon$-derivative properties of $e_{t,\epsilon}$ in equation \eqref{eq:eps d e def}, we have the variations induced by $e_{t, \epsilon}$ on the perturbed drift vector field $u_{t, \epsilon}$ is given by
\begin{align}
    \begin{split}
    \frac{d}{d\epsilon} \Big|_{\epsilon = 0} u_{t, \epsilon} = \Ad_{\Xi_t}\left(\frac{d}{dt} v_t + \ad_{v_t}\bar{u}_{t}\right)\,,\label{eq:delta U calc}
    \end{split}
\end{align}

The perturbation to $a$ that is consistent with perturbing the flow of diffeomorphisms $g$ is defined by $a_{t,\epsilon} := a_0 \wt{g}_{t, \epsilon}^{-1}$, where $\wt{g}_{t, \epsilon}$ is defined in \eqref{eq:perturbed g}. Then, one finds the $\epsilon$ derivative
\begin{align}
\begin{split}
    \frac{d}{d\epsilon}\Big|_{\epsilon=0} a_{t,\epsilon} = a_0 \bar{g}_{t}^{-1} \left(\frac{d}{d\epsilon}\Big|_{\epsilon=0}e^{-1}_{\epsilon, t}\right)\Xi_t^{-1} = -\left(a_0\bar{g}^{-1}_t v_t\right)\Xi_t^{-1} &= -\left(\mcal{L}_{v_t}\bar{a}_t\right)\Xi_t^{-1} \\
    &= -\left(\mcal{L}_{v_t}a_t \Xi_t\right)\Xi_t^{-1}\,.
\end{split}\label{eq:delta a calc}
\end{align}
Similarly, we define the perturbation to $\bar{a}$ by $\bar{a}_{t,\epsilon} := a_0 \bar{g}_{t,\epsilon}$ and we find the $\epsilon$ derivative
\begin{align*}
    \frac{d}{d\epsilon}\Big|_{\epsilon=0} \bar{a}_{t,\epsilon} = a_0 \bar{g}_{t}^{-1} \left(\frac{d}{d\epsilon}\Big|_{\epsilon=0}e^{-1}_{\epsilon,t}\right) = -a_0\bar{g}^{-1}_t v_t = -\mcal{L}_{v_t}\bar{a}_{t,\epsilon}\,.
\end{align*}
Note that the form of the variations of $\bar{a}_{t}$ are the same as in the standard theory of Euler--Poincar\'e reduction with advected quantities \cite{HMR1998}.

The $e_{t, \epsilon}$-induced variations formula for the drift vector fields and advected quantities in equation \eqref{eq:delta U calc} and \eqref{eq:delta a calc} respective will be used extensively in the subsequent variational principles to derive stochastic closures for $\bar{u}_{t}$ in Section \ref{sec:VP1} and deterministic closures in Section \ref{sec:det vp}.

\subsection{Stochastic mean flow closure} \label{sec:VP1}
Fix the stochastic flow of diffeomorphisms $\Xi : \Omega \times [0,T] \rightarrow \operatorname{Diff}(\mcal{D})$ and the associated set $\mathscr{S}^{\Xi}(\operatorname{Diff}(\mcal{D}))$. For $g \in \mathscr{S}^{\Xi}(\operatorname{Diff}(\mcal{D}))$, we associate the unique $\bbF$-adapted vector fields $ \bar{u} : \Omega \times [0,T] \rightarrow \mathfrak{X}(\clD)$ and $ u : \Omega \times [0,T] \rightarrow \mathfrak{X}(\clD)$ by,
\begin{align}
    \bar{g} =: \Xi^{-1} g\,,\quad \bar{u}_{t} := \dd \bar{g}_t \bar{g}^{-1}_t\,,\quad u_{t} := \Xi_{t*}\bar{u}_{t} = \Ad_{\Xi_t}\bar{u}_{t}\,. \label{eqn:g u relation}
\end{align}
To find the dynamics of $\bar{u}$ and $u$, we consider a Lagrangian $\ell : \mathfrak{X}(\mcal{D})\times V^* \rightarrow \mathbb{R}$ and an action functional $S^{\Xi}$ over the set $\mathscr{S}^{\Xi}(\operatorname{Diff}(\mcal{D}))$ with parametric dependence on $a_0$ by 
\begin{align}
    S^{\Xi}[g, a_0] := \int_{0}^{T} \ell(u_{t}, a_0g^{-1}_t) \,\dt = \int_{0}^{T} \ell(u_{t}, a_{t}) \,\dt \label{eq:EPA action stoch}\,,
\end{align}
where the relationships between $g_t$ and $u_t$ are defined in equation \eqref{eqn:g u relation}. In the action principle \eqref{eq:EPA action stoch}, we have made the modelling choice to only include the coefficients of the finite variations part of the semi-martingale $g_t$ in the Lagrangian. This is due to the Lagrangian of interests are typically quadratic in the $\mathfrak{X}(\mcal{D})$ augments and the Lagrangian would be ill-defined when the martingale parts of $g_t$ are included.
\begin{remark}[Other regularisation procedures]
    In the stochastic fluid dynamics literature, similar modelling choice are made for the ill-defined integration against the square of white noise if it arises. For example, in the derivation of the stochastic Navier Stokes equation given in \cite{MR2004}, the term ``$\int {\ddot{W}}\,\dt$'' appearing in the time derivative is balanced with pressure forces  via an application of Doob-Meyer semi-martingale decomposition theorem.
\end{remark}

Assuming that the Lagrangian $\ell$ is smooth with respect to its variables, we define the variational derivatives $\frac{\delta \ell}{\delta u} : \fkX(\clD)\times V^*\rightarrow \fkX^*(\clD)$ and $\frac{\delta \ell}{\delta a} :\fkX(\clD)\times V^* \rightarrow V$ by 
\begin{align}
    \frac{d}{d\epsilon}\bigg|_{\epsilon=0}\ell(u + \epsilon\delta u, a) = \scp{\frac{\delta \ell}{\delta u}(u,a)}{\delta u}\,, \quad \frac{d}{d\epsilon}\bigg|_{\epsilon=0}\ell(u, a + \epsilon\delta a) = \scp{\frac{\delta \ell}{\delta a}(u,a)}{\delta a}\,, \label{eq:vder lag}
\end{align}
for all $\delta u\in \mathfrak{X}(\mcal{D})$ and $\delta a\in V^*$. We will use the shorthand notation
$$\frac{\delta \ell}{\delta u_{t}} := \frac{\delta \ell}{\delta u}(u_{t}, a_t) \quad \textnormal{and}\quad \frac{\delta \ell}{\delta a_t} := \frac{\delta \ell}{\delta a}(u_t, a_t)$$
for the variational derivatives evaluated at $(u_t, a_t)$ for the rest of the paper when the context is clear. 

An application of the principle of least action results in the following proposition. 
\begin{proposition}
The flow of diffeomorphisms $g \in \mathscr{S}^{\Xi}(\operatorname{Diff}(\mcal{D}))$ is a critical point of $S^{\Xi}$ as defined in \eqref{eq:EPA action stoch} if and only if the following random-coefficient Euler--Poincar\'e equation with advected quantities 
    \begin{align}
         \dd\left(\Xi_t^*\frac{\delta \ell}{\delta u_{t}} \right) + \mcal{L}_{\bar{u}_t}\left(\Xi_t^*\frac{\delta \ell}{\delta u_{t}}\right)\, \dt = \left(\Xi_t^*\frac{\delta \ell}{\delta a_{t}} \right)\diamond \bar{a}_{t}\,\dt\,,
    \label{eq:EPA mean}
    \end{align}
    is satisfied.
\end{proposition}

\begin{proof}
Using the variations generated by $e_{t, \epsilon}$ to the $g_t$ as given in \eqref{eqn:perturbed dg} and their induced variations in $u_{t}$ and $a_{t}$, we have 
\begin{align*}
    \begin{split}
    0 = \frac{d}{d\epsilon}\Big|_{\epsilon = 0}S^{\Xi} &= \int_{0}^{T} \frac{d}{d\epsilon}\Big|_{\epsilon = 0} \ell(u_{t,\epsilon}, a_{t,\epsilon})\,\dt\\
    & = \int_{0}^{T}\scp{\frac{\delta \ell}{\delta u_{t}}}{\frac{d}{d\epsilon}\Big|_{\epsilon = 0} u_{t,\epsilon}} + \scp{\frac{\delta \ell}{\delta a_{t,\epsilon}}}{\frac{d}{d\epsilon}\Big|_{\epsilon = 0} a_{t,\epsilon}}\,\dt \\
    & = \int_{0}^{T}\scp{\frac{\delta \ell}{\delta u_{t}}}{\Ad_{\Xi_t}\left(\frac{d}{dt} v_t + \ad_{v_t}\bar{u}_{t}\right)} + \scp{\frac{\delta \ell}{\delta a_{t}}}{- \left(\mcal{L}_{v_t}\bar{a}_{t}\right)\Xi_t^{-1}}\,\dt\\
    &= \int_{0}^{T}\scp{\dd \left(\Ad^*_{\Xi_t}\frac{\delta \ell}{\delta u_{t}}\right) + \ad^*_{\bar{u}_t}\Ad^*_{\Xi_t}\frac{\delta \ell}{\delta u_{t}}\,\dt - \left(\frac{\delta \ell}{\delta a_{t}}\Xi_t\right) \diamond \bar{a}_{t}\,\dt }{ v_t} \,.
    \end{split} 
\end{align*}
In the last equality, we have used integration by parts for the time derivative and the property $v_0 = v_T = 0$ so that the boundary terms vanishes. Since $\Xi_t$ is assumed to be $\operatorname{Diff}(\mcal{D})$-valued continuous semi-martingale, using the Ito-Wentzel formula and product rule \cite{de2020implications}, one can show that $\Ad^*_{\Xi_t}\frac{\delta \ell}{\delta u_{t}}$ a $\bbF$-adapted continuous semimartingale with values in $\fkX^*(\clD)$. Thus, the integration by parts rule holds.  Moreover, since $v :\Omega \rightarrow C^1([0, T]; \mathfrak{X}(\mcal{D}))$ is arbitrary, applying the stochastic version of fundamental lemma of calculus of variations \cite{SC2019, ST2023} yields the random-coefficient Euler--Poincar\'e equation \eqref{eq:EPA mean}.
\end{proof}
\begin{corollary}\label{cor:EP equivalent}
    The random-coefficient Euler--Poincar\'e equation with advected quantities \eqref{eq:EPA mean} arises from the following Euler--Poincar\'e constrained variational principle,
    \begin{align*}
        0 = \delta S = \delta \int_{0}^{T}\ell(u_{t}, a_{t})\,\dt\,, 
    \end{align*}
    subject to constrained variations 
    \begin{align*}
        \delta u_{t} = \Ad_{\Xi_t}\left(\p_t v  - \ad_{\Ad_{\Xi_t^{-1}}u_{t}} v\right)\,,\quad \delta a_{t} = - \left(\mathcal{L}_v \left(a_{t}\Xi_t^{-1}\right)\right)\Xi_t\,,
    \end{align*}
    where $v \in C^1([0,T];\mathfrak{X}(\clD)$ is assumed to be an arbitrary variation vanishing at $t = 0,T$.
\end{corollary}
\begin{proof}
    The variations of $u_{t}$ and $a_{t}$ are obtained in \eqref{eq:delta U calc} and \eqref{eq:delta a calc} respectively. Applying the standard Euler--Poincar\'e variational procedure then yields the desired results.
\end{proof}
We remark that the random-coefficient Euler--Poincar\'e equation \eqref{eq:EPA mean} in its current form is expressed with a mixture of $\bar{u}_{t}$ and $u_{t}$. That is, the momentum $\Xi_t^*\frac{\delta \ell}{\delta u_{t}}$ is expressed in terms of $u_{t}$ and the transport vector field is $\bar{u}_{t} = \Ad_{\Xi_t^{-1}} u_{t}$. In the paragraphs that follows, we consider two equivalent forms of the Euler--Poincar\'e equation \eqref{eq:EPA mean} by expressing ${\delta \ell}/{\delta u_t}$ in terms of $\bar{u}_{t}$ and ${\delta \ell}/{\delta a_t}$ in terms of $\bar{a}_t$. 

The action functional $S^{\Xi}$ defined in \eqref{eq:EPA action stoch} features a non-random, time independent Lagrangian $\ell$. By treating $\Xi$ as an external parameter, we define a random, time-dependent Lagrangian $\ell^{\Xi}: \Omega\times[0, T]\times \mathfrak{X}(\mcal{D})\times V^*\rightarrow \mathbb{R}$ where for all $(\bar{u}, \bar{a}):\Omega\times[0,T] \rightarrow \mathfrak{X}(\clD)\times V^*$,
\begin{align}
    \ell^{\Xi}(\bar{u}, \bar{a}) := \ell(\Ad_{\Xi_t} \bar{u}, \bar{a} \Xi^{-1}_t)\,.\label{eq:EPA mean lag def}
\end{align}
Here, the Lagrangian $\ell^{\Xi}$ is smooth in the $\mathfrak{X}(\mcal{D})$ and $V^*$ variables; the randomness and time dependence of $\ell^{\Xi}$ are defined exclusively by $\Xi$. We will use the same shorthand for the variational derivatives
\begin{align}
    \frac{\delta \ell^{\Xi}}{\delta \bar{u}} : \Omega \times [0, T]\times \mathfrak{X}(\mcal{D}) \times V^* \rightarrow \mathfrak{X}^*(\mcal{D})\,,\quad \frac{\delta \ell^{\Xi}}{\delta \bar{u}_t} := \frac{\delta \ell^{\Xi}}{\delta \bar{u}}(t, \bar{u}_t, \bar{a}_t)\,, \label{eq:vder lag t dep}
\end{align}
and similarly for ${\delta \ell^{\Xi}}/{\delta \bar{a}_t}$, where the $\Omega$ dependence is suppressed as before. From the definition of $\ell^{\Xi}$, we have the equivalence of the Lagrangians $\ell(u, a) = \ell^{\Xi}(\bar{u}, \bar{a})$, evaluated on the variables $\bar{u}$, $u$, $\bar{a}$ and $a$, whose relations are given in equation \eqref{eqn:g u relation}.
From this relation, we have the $\epsilon$-derivatives
\begin{align}
\begin{split}
    \frac{d}{d\epsilon}\Big|_{\epsilon = 0}\ell^{\Xi}(\bar{u}_{t,\epsilon}, \bar{a}_{t,\epsilon}) &= \scp{\frac{\delta \ell^{\Xi}}{\delta \bar{u}_{t}}}{\frac{d}{d\epsilon}\Big|_{\epsilon = 0} \bar{u}_{t,\epsilon}} + \scp{\frac{\delta \ell^{\Xi}}{\delta \bar{a}_{t}}}{\frac{d}{d\epsilon}\Big|_{\epsilon = 0} \bar{a}_{t,\epsilon}}\,,\\
    \frac{d}{d\epsilon}\Big|_{\epsilon = 0}\ell(u_{t,\epsilon}, a_{t,\epsilon}) &= 
    \scp{\frac{\delta \ell}{\delta u_{t}}}{\frac{d}{d\epsilon}\Big|_{\epsilon = 0} u_{t,\epsilon}} + \scp{\frac{\delta \ell}{\delta a_{t}}}{\frac{d}{d\epsilon}\Big|_{\epsilon = 0} a_{t,\epsilon}} \\
    &= \scp{\frac{\delta \ell}{\delta u_{t}}}{\Ad_{\Xi_t}\left( \frac{d}{d\epsilon} \Big|_{\epsilon = 0}\bar{u}_{t,\epsilon}\right)} + \scp{\frac{\delta \ell}{\delta a_{t}}}{\left(\frac{d}{d\epsilon}\Big|_{\epsilon = 0} \bar{a}_{t,\epsilon}\right)\Xi_t}\\
    &= \scp{\Ad_{\Xi_t}^*\frac{\delta \ell}{\delta u_{t}}}{\frac{d}{d\epsilon} \Big|_{\epsilon = 0}\bar{u}_{t,\epsilon}} + \scp{\Xi_t^*\frac{\delta \ell}{\delta a_{t}}}{\frac{d}{d\epsilon}\Big|_{\epsilon = 0} \bar{a}_{t,\epsilon}}\,.\label{Xichainrulevariation}     
\end{split}
\end{align}
and we conclude,
\begin{align*}
    \frac{\delta \ell^{\Xi}}{\delta \bar{u}_{t}} = \Ad^*_{\Xi_t}\frac{\delta \ell}{\delta u_{t}}\,, \qquad \frac{\delta \ell^{\Xi}}{\delta \bar{a}_{t}} = \Xi^*_t \frac{\delta \ell}{\delta a_{t}} := \frac{\delta \ell}{\delta a_{t}}\Xi_t\,.
\end{align*}
Thus, the random-coefficient Euler--Poincar\'e equation \eqref{eq:EPA mean} can be cast into the standard Euler--Poincar\'e form with random, time-dependent Lagrangian $\ell^{\Xi}$,
\begin{align}
     \dd \frac{\delta \ell^{\Xi}}{\delta \bar{u}_{t}} + \mcal{L}_{\bar{u}}\frac{\delta \ell^{\Xi}}{\delta \bar{u}_{t}}\, \dt = \frac{\delta \ell^{\Xi}}{\delta \bar{a}_{t}} \diamond \bar{a}_{t}\,\dt\,.
\label{eq:EPA mean 2}
\end{align}
which is the Euler--Poincar\'e equation for the mean velocity vector field $\bar{u}_{t}$.

To express \eqref{eq:EPA mean 2} in terms of the fluid variables $u$ and $a$, we first state a prerequisite Lemma.
\begin{lemma}[$\Ad^*$ is natural under $\diamond$] \label{prop:Ad^* natural}
Let $V$ be a vector space and $V^*$ be its dual under a duality pairing $\scp{\cdot}{\cdot}$. Assume that $\operatorname{Diff}(\mcal{D})$ has a right representation on $V$ and an induced right representation on $V^*$, both of which are denoted by concatenation. Then, let $\phi \in \operatorname{Diff}(\mcal{D})$, $ b\in V$ and $a\in V^*$ we have 
\begin{align}
    \Ad^*_{\phi} \left(b\diamond a\right) = \left(b \phi\right)\diamond \left(a \phi\right)\,. \label{eq:Ad^* identity}
\end{align}
\end{lemma}
\begin{proof}
    Let $w \in \mathfrak{X}(\mcal{D})$ be arbitrary, we first have the result
    \begin{align}
        \mcal{L}_{\Ad_{\phi} w} a = \phi_{*}\left(\mcal{L}_w \phi^*a\right) = \left(\mcal{L}_w a\phi\right)\phi^{-1} \label{eq:Lie deriv identity}
    \end{align}
    This can be shown through the Cartan form of the Lie derivative
    \begin{align*}
        \mcal{L}_{\Ad_{\phi} w} a = \mb{d}\left(\phi_* w \intprod a\right) + \phi_* w\intprod \mb{d}a = \phi_*\left(\mb{d}\left(w \intprod \phi^*a\right) + w \intprod \mb{d} \phi^*a\right) = \phi_* \left(\mcal{L}_w \phi^*a\right)\,.
    \end{align*}
    For an arbitrary Lie group $G$ and their representation $\rho$ on arbitrary vector spaces $V$ and $V^*$, we have the analogous results expressed in through the representation and its induced representation by the Lie algebra. Using \eqref{eq:Lie deriv identity}, we have
    \begin{align*}
        \begin{split}
            \scp{\Ad^*_{\phi}\left(b\diamond a\right)}{w} = \scp{b\diamond a}{\Ad_{\phi} w} &= \scp{-b}{\mcal{L}_{\left(\Ad_{\phi} w\right)}a} \\
            &= \scp{-b}{\left(\mcal{L}_{w} \,a\phi\right)\phi^{-1}} = \scp{-b\phi}{\mcal{L}_w \,a\phi} = \scp{\left(b\phi\right)\diamond\left(a\phi\right)}{w}\,.
        \end{split}
    \end{align*}
    Since $w$ is arbitrary, we obtain \eqref{eq:Ad^* identity}. 
\end{proof}
\begin{proposition} \label{prop:random to stoch EPA}
Let $\Xi_t$ be the solution to the Stratonovich SDE given in \eqref{eqn:d Xi d gbar def}. Then, the random-coefficient Euler--Poincar\'e equation \eqref{eq:EPA mean} is equivalent to the following Stratonovich stochastic Euler--Poincar\'e equation
\begin{align}  
        \dd \frac{\delta \ell}{\delta u_{t}} + \ad^*_{u_{t}} \frac{\delta \ell}{\delta u_{t}}\,\dt + \sum\limits_{k=1}^K\ad^*_{\xi_k} \frac{\delta \ell}{\delta u_{t}} \circ \dW^k_t + \frac{1}{2}\sum\limits_{k,l=1}^K\ad^*_{\Gamma^{kl}\qv{\xi_k}{\xi_l}}\dede{\ell}{u_{t}}\,\dt = \frac{\delta \ell}{\delta a_{t}} \diamond a_{t}\,\dt\,.
    \label{eq:EPA full}
    \end{align}
\end{proposition}
\begin{proof}
Applying the Kunita-Ito-Wentzel formula for $k$-forms \cite{de2020implications} to compute the time derivative using the stochastic Lie-chain rule (see \cite{crisan2022variational} for the relevant rough path generalisation), we have,
\begin{align*}
\begin{split}
    \dd \left(\Xi_t^* \frac{\delta \ell}{\delta u_{t}}\right) &= \Xi_t^*\left(\dd \dede{\ell}{u_{t}} + \mathcal{L}_{\dd \Xi_t \Xi_t^{-1}}\dede{\ell}{u_{t}}\right) \\
    &= \Xi_t^*\left(\dd \dede{\ell}{u_{t}} + \sum\limits_{k=1}^K \mathcal{L}_{\xi^k}\dede{\ell}{u_{t}} \circ \dW^k_t + \frac{1}{2}\sum\limits_{k,l=1}^K\mcal{L}_{\Gamma^{kl}\qv{\xi_k}{\xi_l}}\dede{\ell}{u_{t}}\,\dt\right)\,.
\end{split}
\end{align*}
which has the equivalent form in the $\Ad^*$, $\ad^*$ notation as
\begin{align}
    \dd \left(\Ad^*_{\Xi_t} \frac{\delta \ell}{\delta u_{t}}\right) = \Ad^*_{\Xi_t}\left(\dd \dede{\ell}{u_{t}} + \sum\limits_{k=1}^K \ad^*_{\xi^k}\dede{\ell}{u_{t}} \circ \dW^k_t + \frac{1}{2}\sum\limits_{k,l=1}^K \ad^*_{\Gamma^{kl}\qv{\xi_k}{\xi_l}}\dede{\ell}{u_{t}}\,\dt\right)\,. \label{eq:stoch lie chain}
\end{align}
The Lie algebra automorphism property of $\Ad$ and its corresponding identity for the dual operator $\Ad^*$ state that for all $\phi \in \operatorname{Diff}(\mcal{D})$, we have 
\begin{align}
\begin{split}
    \ad_{\Ad_{\phi}u}\Ad_{\phi}v \equiv \Ad_{\phi} \ad_{u}v \,, \quad \ad^*_{\Ad_{\phi^{-1}}u} \Ad^*_{\phi}m \equiv \Ad^*_{\phi}\ad^*_{u}m\,,\quad \forall u, v \in \mathfrak{X}(\mcal{D})\,, \quad m \in \mathfrak{X}^*(\mcal{D})\,.
\end{split}\label{eq:comp transport}
\end{align}
which is sometimes known as the composite transport formula, see \cite{H2019}. Combining equations \eqref{eq:stoch lie chain}, \eqref{eq:comp transport}, we have from equation \eqref{eq:EPA mean} that
\begin{align}
    \Ad^*_{\Xi_t}\left(\dd \frac{\delta \ell}{\delta u_{t}} + \ad^*_{u_{t}} \frac{\delta \ell}{\delta u_{t}}\,\dt + \sum\limits_{k=1}^K\ad^*_{\xi_k} \frac{\delta \ell}{\delta u_{t}} \circ \dW^k_t + \frac{1}{2}\sum\limits_{k,l=1}^K\ad^*_{\Gamma^{kl}\qv{\xi_k}{\xi_l}}\dede{\ell}{u_{t}}\,\dt\right) = \left(\frac{\delta \ell}{\delta a_{t}} \Xi_t\right)\diamond \bar{a}_{t}\,\dt\,. \label{eq:EPA calc}
\end{align}
Noting that $\bar{a}_{t} = a_{t} \Xi_t$ and using Lemma \ref{prop:Ad^* natural}, we have 
\begin{align*}
    \left(\frac{\delta \ell}{\delta a_{t}} \Xi_t\right)\diamond \bar{a}_{t} = \Ad_{\Xi_t}^*\left(\frac{\delta \ell}{\delta a_{t}}\diamond a_{t}\right)\,.
\end{align*}
Applying the $\Ad^*_{\Xi_t^{-1}}$ operation to both sides of equation \eqref{eq:EPA calc}, we obtain the stochastic Euler--Poincar\'e equation \eqref{eq:EPA full}.
\end{proof}

The equivalence of the Euler--Poincar\'e equations, \eqref{eq:EPA mean 2} and \eqref{eq:EPA full} implies that stochastic Euler--Poincar\'e equations can arise from random time-dependent Lagrangians. As we will see in the example of incompressible Euler equations in Section \ref{sec:examples}, stochasticity can arise from Lagrangian consisting of time-dependent, random, pullback metrics.

We remark that the stochastic Euler--Poincar\'e equation \eqref{eq:EPA full} have been derived previously in e.g., \cite{Holm2015, DHP2023, crisan2022variational} using rough and stochastic variational principles. In the case that $\Gamma^{kl} = 0$, equation \eqref{eq:EPA full} recovers the stochastic Euler--Poincar\'e equations obtained from the SALT approach \cite{Holm2015}. More specifically, \cite{DHP2023} obtained the stochastic Euler--Poincar\'e equation \eqref{eq:EPA full} using the following stochastic Hamilton-Pontryagin variational principle defined on the Hamilton–Pontryagin bundle $HP = \operatorname{Diff}(\mcal{D})\times \left(\mathfrak{X}(\mcal{D})\oplus \mathfrak{X}^*(\mcal{D})\right)$. For $(g,u,m):\Omega\times[0,T]\rightarrow HP$, the variational principle is given by 
\begin{align*}
    0 = \delta S[g, u, m] &= \delta \int_{0}^{T} \ell(u_t)\,\dt + \scp{m_t}{\dd g_tg_t^{-1} - u_t\,\dt - \sum\limits_{k=1}^K\xi_k \circ \dW^k_t - \frac{1}{2}\sum\limits_{k,l=1}^K\Gamma^{kl}\qv{\xi_k}{\xi_l}\,\dt}\,, 
\end{align*}
where the variations $\delta u_t$, $\delta m_t$ are arbitrary and vanishing at $t=0, T$ and the variations $\delta g_t$ are constructed without using the group definition, $g_t = \Xi_t \circ \bar{g}_t$. That is, $\delta g_t = \frac{d}{d\epsilon}\big|_{\epsilon=0} e_{\epsilon,t } \circ g_t$. This class of variations were considered in e.g., \cite{ST2023, crisan2022variational}. This is in contrast to the Hamilton–Pontryagin variant of the variational principle \eqref{eq:EPA action stoch} that gives \eqref{eq:EPA full} as its stationary condition,
\begin{align*}
    0 = \delta S[\bar{g}, u, m] &= \delta\int_{0}^{T} \ell\left(u_t \right) + \scp{m_t}{\Ad_{\Xi_t}\dot{\bar{g}}_t\bar{g}^{-1}_t - u_t}\,dt \,,
\end{align*}
where the variations are taken over the mean diffeomorphism $\bar{g} \in \bar{\mathscr{S}}(\operatorname{Diff}(\mcal{D}))$ whilst leaving $\Xi_t$ unchanged using $\delta \bar{g}_t = \frac{d}{d\epsilon}\big|_{\epsilon=0} e_{\epsilon,t } \circ \bar{g}_t$. 

\begin{remark}[Continuous path dependence of the slow vector field]\label{rem:why_gbar_is_random} 
Given $\varepsilon > 0$, the variational principles in this section can be applied to derive closures of the maps $g^\varepsilon, \bar{g}^\varepsilon$ discussed in Section \ref{sec:homog}. 

Consider the case that the Euler--Poincar\'e equation is driven by a path $\bB^\varepsilon$ approximating Brownian motion and is taken to be well-posed. Using rough stability techniques, we can deduce a Wong-Zakai principle holds for these equations that necessarily implies the convergence of $g^\varepsilon$ to $g$ solving \eqref{eq:EPA full}. One may then ask if such a property holds for map $\bar{g}^\varepsilon$, and a Wong-Zakai principle to approximate the solution of \eqref{eq:EPA mean 2}. 

Suppose that for every $\alpha \in (\frac13, \frac12]$ rough path $\bZ\in \clC_g^{\alpha}([0,T]; \mathbb{R}^K)$, there exists a classical solution of the rough partial differential equation
\begin{align}\label{eq:RALT EP}
    \dd\frac{\delta \ell}{\delta u_{t}} + \ad^*_{u_{t}} \frac{\delta \ell}{\delta u_{t}}\,\dt + \sum\limits_{k=1}^K\ad^*_{\xi_k} \frac{\delta \ell}{\delta u_{t}}\dd \bZ^k_t= 0\,, 
\end{align}
that is continuous as a map on $\clC_g^{\alpha}([0,T]; \mathbb{R}^K)$ taking values in a space of $\alpha$-H\"older curves of vector fields with $C^n$ regularity, $n \in \bbN$ (see e.g., \cite[Theorem 3.7]{CRISAN2022109632}). This implies we can write solutions in the form $u_t=U_t( \cdot , \bZ)$ for a map $U \in C(\clC_g^\alpha ([0,T]; \bbR^K ) ; C^{\alpha}([0,T]; \fkX_{C^n}(\bbT^d)) )$. In Section \ref{sec:homog} the map $\Xi$ was constructed from a rough flow map $\Xi_t(\cdot) := \Phi_t(\cdot, \bZ)$ with $\Phi \in \operatorname{Lip}_{loc}\left(\clC_g^{\alpha}([0,T]; \mathbb{R}^K);C^\alpha( [0,T] ;\operatorname{Diff}_{C^{n-1}}(\bbT^d))\right)$ in the proof of Theorem \ref{thm:fast_flow_covergence}. 

Given the expression $\bar{u}_t = \Ad_{\Xi_t^{-1}} u_t$, it is natural to ask whether $\bar{u}$ inherits this continuous path dependence. This motivates the Assumption \ref{asm:u_bar_eps} that was invoked in Section \ref{sec:homog}. For the case of finite dimensional (matrix) Lie group $G$ and its Lie algebra $\mathfrak{g}$, the analogue of \eqref{eq:RALT EP} (with convex, smooth Lagrangian $\ell$) the adjoint representation $\Ad$ is smooth and we have the continuous map,
\begin{equation*}
\begin{aligned}
\Ad_{(\cdot)^{-1}} : C^\alpha([0,T]; G) \times C\left([0,T]; \mathfrak{g} \right) &\rightarrow C\left([0,T]; \mathfrak{g}\right) \\
(\Xi, u) &\mapsto \Ad_{\Xi^{-1}} u = \Xi^{-1} u \Xi =: \bar{u}\, .
\end{aligned}
\end{equation*}
Since $\bZ \mapsto (\Xi, u)$ is continuous, we obtain $\bar{u}_t = V_t(\cdot, \bZ_t )$ for $V \in C(\clC_g^\alpha ([0,T]; \bbR^K ) ; C([0,T]; \fkg ) )$ with $V = \Xi^{-1}U \Xi$.
In the case $G = \operatorname{Diff}_{C^n}(\mc{D})$, the map $\Ad_{\Xi^{-1}} = \Xi^*$ is a bounded linear map (and thus smooth).  We conjecture the continuity of these maps for the $\alpha$-H\"older continuous in time analogues, 
\begin{equation*}
\begin{aligned}
\Ad_{(\cdot)^{-1}} : C^\alpha([0,T]; \operatorname{Diff}_{C^n}(\mcal{D})) \times C^\alpha\left([0,T]; \mathfrak{X}_{C^n}(\clD)\right) &\rightarrow C^\alpha\left([0,T]; \mathfrak{X}_{C^{n-1}}(\clD)\right) \\
(\Xi, u) &\mapsto \Xi^* u \, ,
\end{aligned}
\end{equation*}
For Euler's equations, we verify this conjecture by constructing an explicit SPDE for $\bar{u}$, allowing us to invoke rough stability (see \eqref{eq:ubar dynamics} in Section \ref{sec:examples}). It follows that in the case of Euler's equation we have continuous path dependence in both $\bar{u}, u$. 
\end{remark}

\begin{remark}
    We remark that through the definitions $\Xi^*_t a_{t} = \overline{a}_{t}$ and $\Xi^*_t u_{t} = \overline{u}_{t}$, it is not typically the case that $\Xi^*_t \frac{\delta \ell}{ \delta u_t} = \frac{\delta \ell}{ \delta \overline{u}_t}$ or $\Xi^*_t \frac{\delta \ell}{ \delta a_t} = \frac{\delta \ell}{ \delta \overline{a}_t}$. Nevertheless, it is often that more direct expressions in terms of the mean quantities $\overline{u}_t, \overline{a}_t$ can be derived for specific choices of Lagrangian. For example, consider the Euler-Boussinesq Lagrangian
    $$\ell(\Xi_{t*} \bar{u}_t, \Xi_{t*}\bar{D}_t, \Xi_{t*}\bar{b}_t) = \int_{\mcal{D}} \frac12 |\Xi_{t*} \bar{u}_t|^2\Xi_{t*}\bar{D} + \mathrm{g}\Xi_{t*}\bar{D} \Xi_{t*}\bar{b}_t \hat{z} + \dd p(\bar{D} - 1) d\mu_{\mathbf{g}}$$
    where $b_t = \Xi_{t*}\bar{b}_t$ is the advected buoyancy scalar, $\mathrm{g}$ the gravity constant and $\hat{z}$ the constant $z$ axis unit vector. One has $\frac{\delta \ell}{\delta b_t} = \mathrm{g} \Xi_{t*}\bar{D}_t \hat{z}$ and thus $\Xi_t^* \frac{\delta \ell}{\delta b_t} = \mathrm{g} \bar{D}_t \hat{z}$, this calculation is in agreement with \cite[Sec. 3.3]{Holm2002}.
\end{remark}

\subsubsection{Kelvin-Noether circulation theorem}
Associated with the random-coefficient Euler--Poincar\'e equation \eqref{eq:EPA mean} and the stochastic Euler--Poincar\'e equation \eqref{eq:EPA full} are the Kelvin-Noether circulation theorems, which we will state next. 
\begin{theorem}
    Let $(\bar{u}_t, \bar{a}_{t})$ be the solution to the random-coefficient Euler--Poincar\'e equation \eqref{eq:EPA mean} and let $(u_{t}, a_{t})$ be the solution to the stochastic Euler--Poincar\'e equation \eqref{eq:EPA full} for $t \in [0, T]$. Let $g_t, \bar{g}_t \in \operatorname{Diff}(\mcal{D})$ satisfy the flow equations defined in \eqref{eqn:d Xi d gbar def} and \eqref{eqn:homoged diffeo}, respectively. Assume that there exists $\bar{D}, D: \Omega\times [0,T] \rightarrow \Lambda^d(\clD)$ that is defined by $\bar{D}_t := D_0 \bar{g}_t^{-1} = \bar{g}_{t*} D_0$ and $D_t := D_0 g_t^{-1} = g_{t*} D_0$ for an initial non zero density $D_0 \in \Lambda^d(\mcal{D})$. Then, we have the following equivalent Kelvin-Noether circulation theorems,
    \begin{align}
        \dd \oint_{\bar{c}_t} \frac{1}{\bar{D}_t} \left(\Xi^*_t\frac{\delta \ell}{\delta u_{t}}\right) = \oint_{\bar{c}_t} \frac{1}{\bar{D}_t}\left(\Xi^*_t\frac{\delta \ell}{\delta a_{t}}\right) \diamond \bar{a}_{t}\,\dt\,,\label{eqn:mean kelvin}
    \end{align}
    and 
    \begin{align}
        \dd \oint_{c_t} \frac{1}{D_t} \frac{\delta \ell}{\delta u_{t}} = \oint_{c_t} \frac{1}{D_t} \frac{\delta \ell}{\delta a_{t}} \diamond a_{t}\,\dt\,,\label{eqn:full kelvin}
    \end{align}
    where $\bar{c}_t := \bar{g}_t(c_0)$ and $c_t := g_t(c_0)$ denote the left action of $\bar{g}_t$ and $g_t$ on the embedded one-dimensional submanifold $c_0$, the initial embedding of a fluid material loop.
\end{theorem}
\begin{remark}\label{rmk:division by density}
    For $m = \alpha \otimes \rho \in \mathfrak{X}^*(\mcal{D})$ where $\alpha \in \Lambda^1(\mcal{D})$ and $\rho \in \Lambda^d(\mcal{D})$, the ``division by density'' operation appearing in equations \eqref{eqn:mean kelvin} and \eqref{eqn:full kelvin} is interpreted as
    \begin{align*}
        \frac{1}{D}m = \frac{1}{D}\left(\alpha\otimes \rho\right) := \frac{\star \rho}{\star D}\alpha \in \Lambda^1(\mcal{D})\,,
    \end{align*}
    where $\star$ denote the Hodge dual, $\star: \Lambda^{d-k}(\mcal{D}) \rightarrow \Lambda^k(\mcal{D})$ for $k = 0,\ldots,d$.
\end{remark}
\begin{proof}
    From the definition of $\bar{D}_t$ and $D_t$, they are the solution to the advection equations
    \begin{align*}
        \dd \bar{D}_t + \mcal{L}_{\bar{u}_{t}}\bar{D}_t\,\dt = 0\,,\qquad 
        \dd D_t + \mcal{L}_{u_{t}} D_t\,\dt + \sum\limits_{k=1}^K\mcal{L}_{\xi_k} D_t\circ \dW^k_t + \frac{1}{2}\sum\limits_{k,l=1}^K\mcal{L}_{\Gamma^{kl}\qv{\xi_k}{\xi_l}}D_t\,\dt = 0\,.
    \end{align*}
    Let $\alpha:\Omega \times[0,T] \rightarrow \Lambda^1(\clD)$ be such that $\frac{\delta \ell}{\delta u_t} = \alpha_t\otimes D_t$. This implies the similar decomposition of $\Xi^*_t \frac{\delta \ell}{\delta u_t} = \Xi^*_t\alpha_t \otimes \Xi^*_t D_t =: \bar{\alpha}_t \otimes \bar{D}_t$ such that $\alpha = \frac{1}{D_t}\frac{\delta \ell}{\delta u_t}$ and $\bar{\alpha}_t = \frac{1}{\bar{D}_t}\Xi^*_t\frac{\delta \ell}{\delta u_t}$, where the operation of division by density is understood in the sense of Remark \ref{rmk:division by density}.
    We first prove equation \eqref{eqn:mean kelvin} holds. Equation \eqref{eq:EPA mean} implies
    \begin{align*}
    \begin{split}
        &\left(\dd + \ad^*_{\bar{u}_t\,\dt}\right)\left(\Xi^*_t \frac{\delta \ell}{\delta u_t}\right) = \left(\dd + \mcal{L}_{\bar{u}_t\,\dt}\right)\bar{\alpha}_t \otimes \bar{D}_t + \bar{\alpha}_t \otimes \left(\dd + \mcal{L}_{\bar{u}_{t}}\right)\bar{D}_t = \left(\dd + \mcal{L}_{\bar{u}_t\,\dt}\right)\bar{\alpha}_t \otimes \bar{D}_t = \Xi^*_t\frac{\delta \ell}{\delta a_t}\diamond \bar{a}_t\,\dt\,, \\
        & \qquad \Longrightarrow \qquad \left(\dd + \mcal{L}_{\bar{u}_t\,\dt}\right)\bar{\alpha}_t = \frac{1}{\bar{D}}\left(\Xi^*_t\frac{\delta \ell}{\delta a_t}\diamond \bar{a}_t\right)\,\dt\,.
    \end{split}
    \end{align*}
    In the first line, the first equality uses the Lie-chain rule \eqref{eq:stoch lie chain} and the second equality uses the advection equation of $\bar{D}$ by $\bar{u}$. Then, we obtain 
    \begin{align*}
        \dd \oint_{\bar{g}_t c_0} \bar{\alpha}_t = \dd \oint_{c_0}\bar{g}_t^*\bar{\alpha}_t = \oint_{c_0}\bar{g}_t^*\left(\dd + \mcal{L}_{\bar{u}_t\,\dt}\right)\bar{\alpha}_t = \oint_{\bar{g}_tc_0} \left(\dd + \mcal{L}_{\bar{u}_t\,\dt}\right)\bar{\alpha}_t = \oint_{\bar{c}_t} \frac{1}{\bar{D}_t}\left(\Xi^*_t\frac{\delta \ell}{\delta a_t}\diamond \bar{a}_t\right)\,\dt\,.
    \end{align*}
    Repeating the same arguments for equation \eqref{eq:EPA full}, we obtain
    \begin{align*}
        \begin{split}
            \dd \alpha_t + \mcal{L}_{u_t} \alpha_t \,\dt + \sum\limits_{k=1}^K\mcal{L}_{\xi_k} \alpha\circ \dW^k_t + \frac{1}{2}\sum\limits_{k,l=1}^K\mcal{L}_{\qv{\xi_k}{\xi_l}} \alpha_t \,\dt = \frac{1}{D_t}\frac{\delta \ell}{\delta a_t}\diamond a_t\,\dt\,.
        \end{split}
    \end{align*}
    Direct computation of the stochastic differential of the loop integral yields equation \eqref{eqn:full kelvin} since
    \begin{align*}
        \dd \oint_{g_t c_0} \alpha_t = \dd \oint_{c_0}g_t^*\alpha_t = \oint_{c_0}g_t^*\left(\dd + \mcal{L}_{\dd g_t g_t^{-1}}\right){\alpha}_t = \oint_{g_t c_0} \left(\dd + \mcal{L}_{\dd g_t g_t^{-1}}\right)\alpha_t = \oint_{c_t} \frac{1}{D_t}\frac{\delta \ell}{\delta a_t}\diamond a_t \,\dt
    \end{align*}
    where the notation $\dd g_t g_t^{-1}$ defines the Eulerian vector field generate by $g_t$ in equation \eqref{eqn:homoged diffeo}. Since equations \eqref{eq:EPA mean} and \eqref{eq:EPA full} are equivalent by Proposition \ref{prop:random to stoch EPA}, we have the Kelvin-Noether circulation dynamics are equivalent.
\end{proof}

\subsubsection{Hamiltonian formulation}\label{subsec:ham}
On the Hamiltonian side, we define two equivalent Hamiltonians, $h: \mathfrak{X}(\clD)^* \times V^* \rightarrow \mathbb{R}$ and $h^{\Xi}:\Omega\times[0,T]\times \mathfrak{X}(\clD)^*\times V^*\rightarrow \mathbb{R}$ using the Legendre transform on the Lagrangians $\ell(u,a)$ and $\ell^{\Xi}(\bar{u},\bar{a})$ as
\begin{align}
    h^{\Xi}(\bar{m}_{t}, \bar{a}_{t}) := \scp{\bar{m}_{t}}{\bar{u}_{t}} - \ell^{\Xi}(\bar{u}_{t}, \bar{a}_{t})\,, \quad h(m_{t}, a_{t}) := \scp{m_{t}}{u_{t}} - \ell(u_{t}, a_{t})\,. \label{eq: ham defs}
\end{align}
Here, the momentums are related to the Lagrangians as
\begin{align*}
    \bar{m}_t = \frac{\delta \bar{\ell}}{\delta \bar{u}_t}\,,\quad  m_t = \frac{\delta \ell}{\delta u_t}\,.
\end{align*}
Working in the regular Lagrangian case, we additionally have the relations
\begin{align}
    \bar{u}_{t} = \dede{h^{\Xi}}{\bar{m}_{t}} \,,\quad  \dede{\ell^{\Xi}}{\bar{a}_{t}} = -\dede{h^{\Xi}}{\bar{a}_{t}}\,,\quad u_{t} = \dede{h}{m_{t}}\,, \quad  \dede{\ell}{a_{t}} = -\dede{h}{a_{t}}\,,
\end{align}
where the notations of the variational derivatives of the Hamiltonians $h$ and $h^{\Xi}$ are the analogously to the variational derivatives of the Lagrangians in equations \eqref{eq:vder lag} and \eqref{eq:vder lag t dep}, respectively.
When $u_{t} = \Xi_{t*} \bar{u}_{t} = \Ad_{\Xi_t}\bar{u}_t$ and $a_t = \bar{a}_t\Xi_t = \Xi_t^* \bar{a}_t$, we have the equivalence of the Hamiltonians 
\begin{align*}
    h^{\Xi}(\bar{m}, \bar{a}) = h(m, a)\,,
\end{align*}
by the definition of $\ell^{\Xi}$ in terms of $\ell$ and their variational derivatives given in equation \eqref{Xichainrulevariation} that defines the momentums $\bar{m}$ and $m$. 

The Euler--Poincar\'e equations given in \eqref{eq:EPA mean} and \eqref{eq:EPA full} can be cast into Lie-Poisson form where it can be written in the equivalent Lie-Poisson matrix form
\begin{align}
    \dd \begin{pmatrix}
        \bar{m}_{t} \\ \bar{a}_{t}
    \end{pmatrix}
     = -
     \begin{pmatrix}
         \ad^*_{\square} \bar{m}_{t} & \square \diamond \bar{a}_{t} \\
         \mathcal{L}_{\square}\bar{a}_{t} & 0
     \end{pmatrix}
     \begin{pmatrix}
         {\delta h^{\Xi}}/{\delta \bar{m}_{t}} \\
         {\delta h^{\Xi}}/{\delta a_{t}} 
     \end{pmatrix}
     \dt\,, \label{LPbracket1}
\end{align}
and 
\begin{align}
\begin{split}
    \dd \begin{pmatrix}
        m_{t} \\ a_{t}
    \end{pmatrix}
     = -
    \begin{pmatrix}
        \ad^*_{\square} m_{t} & \square \diamond a_{t} \\
        \mathcal{L}_{\square} a_{t} & 0
    \end{pmatrix}
    \begin{pmatrix}
        {\delta h}/{\delta m_{t}}\\
        {\delta h}/{\delta a_{t}} 
    \end{pmatrix}
    \dt &- 
    \sum\limits_{k=1}^K
    \begin{pmatrix}
        \ad^*_{\square} m_{t} & \square \diamond a_{t} \\
        \mathcal{L}_{\square} a_{t} & 0
    \end{pmatrix}
    \begin{pmatrix}
        \xi_k \\
        0
    \end{pmatrix}
    \circ \dW^k_t\\
    & \hspace{-2em} - \frac{1}{2}\sum\limits_{k,l=1}^K
    \begin{pmatrix}
        \ad^*_{\square} m_{t} & \square \diamond a_{t} \\
        \mathcal{L}_{\square} a_{t} & 0
    \end{pmatrix}
    \begin{pmatrix}
        \Gamma^{kl}\qv{\xi_k}{\xi_l} \\
        0
    \end{pmatrix}\,\dt\,.
\end{split} \label{LPbracket2}
\end{align}
Here, the product between the $2\times 2$ matrix with bi-linear operator as entries and the vector of variational derivatives should be understood as matrix multiplication. The $\square$ symbol appearing in \eqref{LPbracket1} and \eqref{LPbracket2} denotes the argument of the bi-linear operator appearing in the matrix where the vector entries are inserted to in the matrix multiplication. 

The equations \eqref{LPbracket1} and \eqref{LPbracket2} are the Lie-Poisson equations on the semidirect product Lie co-algebra $\mathfrak{s}^* = \mathfrak{X}^*(\clD)\ltimes V^*$. Let $X = C^\infty( \mathfrak{s}^*, \mathbb{R})$. The semidirect product Lie-Poisson bracket $\{\cdot, \cdot\}: X\times X \rightarrow X$ is defined by 
\begin{align}
    \{f, g\}(m, a) := - \scp{m}{\qv{\frac{\delta f}{\delta m}}{\frac{\delta g}{\delta m}}} + \scp{a}{\mcal{L}_{\frac{\delta g}{\delta m}}\frac{\delta f}{\delta a} - \mcal{L}_{\frac{\delta f}{\delta m}}\frac{\delta g}{\delta a}}\,, \label{LPdef}
\end{align}
where $f, g\in X$ and  $(m, a) \in \mathfrak{X}(\mcal{D})\times V^*$.
Using the Lie-Poisson bracket defined above, the evolution of $f$ can be expressed as 
\begin{align*}
    \begin{split}
        \dd f(\bar{m}_{t},\bar{a}_{t}) &= -\{f, h^{\Xi}\}(\bar{m}_{t},\bar{a}_{t})\,\dt\,,\\
        \dd f(m_{t}, a_{t}) &= -\{f, h\}(m_{t}, a_{t})\,\dt - \sum\limits_{k=1}^K\{f, \mathfrak{h}_k\}(m_{t}, a_{t})\circ \dW^k_t -\{f, \hslash\}(m_{t}, a_{t})\,\dt\,,
    \end{split}
\end{align*}
where
\begin{align*}
\mathfrak{h}_k := \scp{m}{\xi_k}\,, \quad \hslash := \sum\limits_{k,l=1}^K\Gamma^{kl}\scp{m}{\qv{\xi_k}{\xi_l}}\,.    
\end{align*}
The above calculation implies the following corollary. A Casimir is any element $C\in X$ such that  $\{C, f\}=0$ for all $f\in X$. 
\begin{corollary} \label{casimirlemma}
For any Casimir function $C$, the quantities $C(m_{t}, a_{t})$ and $C(\bar{m}_{t},\bar{a}_{t})$ are conserved by the dynamics of $(m_{t}, a_{t})$ and $(\bar{m}_{t},\bar{a}_{t})$, respectively.
\end{corollary}

\subsubsection{Isometries and energy conservation} 
From the Lie-Poisson systems \eqref{LPbracket1} and \eqref{LPbracket2}, we have two definitions for the energy of the Lie-Poisson systems. The mean Hamiltonian $h^{\Xi}$ evaluated at the mean variables $(\bar{m}_t,\bar{a}_t)$ defines an random and time-dependent energy for the mean system. The Hamiltonian $h$ does not explicitly dependent on time, however, the non-conservative dynamics of $h(m_t,a_t)$ are due to the stochastic Lie-Poisson structure given in \eqref{LPbracket2}. When the conditions $a_t = \Xi^*_t \bar{a}_t$ and $\bar{m}_t = \Ad^*_{\Xi_t}m_t$ hold, we have that $h(m_t,a_t) = h^{\Xi}(t, \bar{m}_t,\bar{a}_t)$ for $t \in [0,T]$. Thus, we can directly compute the evolution of $h(m_t,a_t)$ under the flow of $(m_t, a_t)$ through $h^{\Xi}(\bar{m}_t, \bar{a}_t)$ to have
\begin{align} 
\begin{split}
\dd h(m_t, a_t) = \dd h^{\Xi}(\bar{m}_{t}, \bar{a}_{t}) &= \scp{\frac{\delta h^{\Xi}}{\delta \bar{m}_{t}}}{ \dd \bar{m}_{t}} + \scp{\frac{\delta h^{\Xi}}{\delta \bar{a}_{t}}}{ \dd \bar{a}_{t}} + \scp{ \frac{ \delta h^{\Xi} }{ \delta \Xi_t} }{\dd \Xi_t} \\
& = \scp{\bar{u}_t}{-\mcal{L}_{\bar{u}_t}\bar{m}_{t}\,\dt - \frac{\delta h^{\Xi}}{\delta \bar{a}_{t}}\diamond \bar{a}_{t}\,\dt} + \scp{\frac{\delta h^{\Xi}}{\delta \bar{a}_{t}}}{ - \mcal{L}_{\bar{u}}\bar{a}_{t}\,\dt} + \scp{ \frac{ \delta h^{\Xi} }{ \delta \Xi_t} }{\dd \Xi_t} \\
& = \scp{ \frac{ \delta h^{\Xi} }{ \delta \Xi_t} }{\dd \Xi_t}\,.
\label{energyconservation}
\end{split}
\end{align} 
Note that the variational derivative with respect to diffeomorphisms $\frac{ \delta h^{\Xi} }{ \delta \Xi_t}$ is well defined in the Gateaux sense on tangent spaces \cite{HMR1998}. It follows we have the Hamiltonian $h(m_{t}, a_{t})$ is conserved in time only when the mean energy is $\Xi_t$ (and therefore time) independent.  
\begin{proposition}\label{energyconservethm}
    Let $\Xi$ be the stochastic flow satisfying \eqref{eqn:homoged diffeo}. Assume that the Lagrangian $\ell: \mathfrak{X}(\mcal{D})\times V^* \rightarrow \mathbb{R}$ is $\Xi$ invariant in the sense that for all $t\in [0,T]$, $\bar{u}\in \fkX(\mcal{D})$, and $\bar{a}\in V^*$,
    \begin{equation*}\label{eq:xi_invariance}
        \ell(\Ad_{\Xi_t} \bar{u},  \bar{a} \Xi_t^{-1}  ) = \ell ( \bar{u} ,  \bar{a} ) \,.
    \end{equation*}
    Then the corresponding stochastic Euler--Poincar\'e equations \eqref{eq:EPA full} are energy preserving.
    Furthermore, if  $\ell(\bar{u},\bar{a}) = \frac12\scp{\bar{u}}{\bar{u}}_{\mathfrak{X}(\mcal{D}) \times \mathfrak{X}^*(\mcal{D})}$ is the kinetic energy Lagrangian, then $\Xi$-invariance is equivalent to the vector fields $\{\xi_k\}_{k=0}^K$ being killing vector fields.    
\end{proposition}
\begin{proof}
    This follows from the Legendre transform,
    \begin{align*}
    \begin{split}
        h^{\Xi}(\bar{m}_{t}, \bar{a}_{t}) &:= \scp{ \bar{m}_{t} }{\bar{u}_t} - \ell(\Ad_{\Xi_t}\bar{u}_t, \bar{a}_{t}\Xi^{-1}_t)\\
        &= \scp{ \bar{m}_{t} }{\bar{u}_{t}} - \ell(\bar{u}_t, \bar{a}_{t}) = h(\bar{m}_{t}, \bar{a}_{t})\,. 
    \end{split}
    \end{align*} 
    Thus, $\scp{ \frac{ \delta h^{\Xi} }{ \delta \Xi_t} }{\dd \Xi_t} = 0$, by \eqref{energyconservation} the energy is conserved. For the kinetic energy of a fluid, one can show,
\begin{align}
    \int_{\mc{D}} \frac12 \mathbf{g}(u_t, u_t) D_t = \int_{\mc{D}} \frac12 \mathbf{g}(\Xi_{t *}\bar{u}_t, \Xi_{t *}\bar{u}_t) \Xi_{t *}\bar{D}_t = \int_{\mc{D}} \frac12 \Xi_{t*} \Big( (\Xi^*_t \mathbf{g}) (\bar{u}_t , \bar{u}_t) \bar{D}_t \Big) = \int_{\mc{D}} \frac12  (\Xi^*_t \mathbf{g}) (\bar{u}_t , \bar{u}_t) \bar{D}_t
\end{align}
    The requirement that $\Xi_t^* \mathbf{g} = \mathbf{g}$ is precisely that $\Xi_t$ acts as an isometry at all times, such flows are generated by Killing fields.
\end{proof}
In general, the vector field $\bar{u}_t$ appearing in the random-coefficient Euler--Poincar\'e equations \eqref{eq:EPA mean 2} and the stochastic Euler--Poincar\'e equations \eqref{eq:EPA full} are coupled via the stochastic momentum $\bar{m}_{t} = \Xi^*_t \frac{\delta \ell}{\delta u_t}$.  When $\Xi_t$ is an isometry, the variational principle \eqref{eq:EPA mean lag def} contains no randomness and produces a deterministic PDE. This phenomenon is illustrated in Section \ref{sec:Euler} for the incompressible Euler equations and Section \ref{rigidbodyexample} for rigid body rotation dynamics.

\subsection{Deterministic mean flow closure and averaging}\label{sec:det vp}
In Section \ref{sec:VP1}, we have introduced a $\Xi$-coupled closure for the flow generated by a mean map $\bar{g}$. The full, composite flow $g = \Xi \circ \overline{g}$ produces the Euler--Poincar\'e equations corresponding to the SALT approach when $\Xi$ generates the stochastic flow specified in equation \eqref{eqn:d Xi d gbar def}, which was motivated through homogenisation procedure presented in Section \ref{sec:homog} for a fast, chaotic flow $\Xi^\varepsilon$. 

Crucially, the closure introduced in Section \ref{sec:VP1} typically introduces randomness in $\bar{u}$, which motivated Assumption \ref{asm:u_bar_eps}. See, also Remark \ref{rem:why_gbar_is_random} and equation \eqref{eq:ubar dynamics} in the examples section. A distinguished case where $\bar{u}$ remained uncoupled to noise is examined in Proposition \ref{energyconservethm} where $\ell$ is assumed to be $\Xi$-invariant. In terms of Assumption \ref{asm:u_bar_eps}, this is the trivial case where $\bar{u}^\varepsilon \equiv \bar{u}$ is a constant function of the rough path, which is smooth. 

In this subsection, we propose a variational closure of $\bar{u}$ that can be assumed to be fully decoupled from the fast scales. This approach more closely matches the axioms of GLM and allows a weakening of the Assumption \ref{asm:u_bar_eps}, so that the limiting mean vector field $\bar{u}$ remains deterministic. The motivation for this closure is a type of \emph{averaging} (in the sense of \cite[Chapter 10]{pavliotis2008multiscale}), rather than homogenisation theory, that is applied to $\bar{u}^\varepsilon$. One can treat the theory of the fast map $\Xi^\varepsilon$ in the same manner as Section \ref{sec:homog} and adjust the assumptions on the mean variable to deduce convergence of the composition of maps. Fix $n \in \bbN$. Given $\bar{u}^{\varepsilon}: \Omega \times [0,T] \rightarrow \mathfrak{X}_{C^n}(\bbT^d)$ with a flow map $\bar{g}^{\varepsilon} :\Omega \times [0,T] \rightarrow \operatorname{Diff}_{C^n}([0,T];\bbT^d) $ 
defined by,
$$\dot{\bar{g}}_t^{\varepsilon}(X) = \bar{u}_t^{\varepsilon}(\bar{g}_t^{\varepsilon}(X)), \quad \bar{g}_0^{\varepsilon}(X)=X. 
$$
We make the following assumption.

\begin{assumption}[Alternative mean vector field assumption] 
\label{weakgbarflow} There exists a $\bar{u} : [0,T] \rightarrow \mathfrak{X}_{C^n}(\bbT^d)$ such that $u^\varepsilon \rightarrow_\bbP u$ in $C\left([0,T]; \mathfrak{X}_{C^n}(\bbT^d)\right)$. It follows that the flow map $\bar{g}^\varepsilon$ converges in law to a $\bar{g} : [0,T] : \rightarrow  \operatorname{Diff}_{C^n}(\bbT^d) $ such that,
\begin{align*}
    \dot{\bar{g}}_t(X) = \bar{u}_t(\bar{g}_t(X)), \quad \bar{g}_0(X)=X \, .
\end{align*}
Note that $\bar{u}^\varepsilon$ the notion of convergence in Assumption \ref{weakgbarflow} is arbitrary, and does not require particular direct dependence of the rough path $\bB^\varepsilon(\omega)$ (in contrast to Assumption \ref{asm:u_bar_eps}). 

The convergence in law to a deterministic element implies convergence in probability. By Slutsky's theorem\footnote{See \cite[Example 3.2]{billingsley1968convergence}. $C\left([0,T] ; \operatorname{Diff}_{C^n}(\clD) \right)$ is a Polish space \cite[Section 11.2]{friz2010multidimensional} and is thus separable.}, $(\Xi^\varepsilon, \bar{g}^\varepsilon) \rightarrow_\bbP (\Xi, \bar{g})$ converges jointly and the argument for convergence of $g^\varepsilon = \Xi^\varepsilon \circ \bar{g}^\varepsilon$ may be repeated as in Section \ref{sec:homog}.
\end{assumption}

We now seek a consistent variational closure for a deterministic flow $\bar{g}$. Let $\ell^{\Xi} : \Omega  \times [0,T] \times \mathfrak{X}(\clD) \times V^* \rightarrow \bbR$ be the (random-coefficient) Lagrangian as considered in \eqref{eq:EPA mean lag def}. We define the deterministic Lagrangian $\bar{\ell} : [0,T] \times \mathfrak{X}(\clD) \times V^* \rightarrow \bbR$ by taking the expectation with respect to the underlying probability space to have
\begin{align*}
\bar{\ell}(t, \bar{u}, \bar{a}) := \bbE \left[ \ell^{\Xi}( \bar{u},  \bar{a})\right] =  \bbE \left[ \ell(\Ad_{\Xi_t} \bar{u}_t,  \bar{a}_t \Xi^{-1}_t)\right]  \, .
\end{align*}
A variety of works have considered the averaging of fluid dynamics in different contexts and equation \eqref{eq:ELASALT EP 1} is a synthesis of several approaches. For example, an expectation average variational principle was posed in \cite{CCR2023}. The variational principle \eqref{sleep_lag} and its resulting equations resemble a probabilistic analogue of the Lagrangian Averaged Euler--Poincar\'e equations (LAEP) \cite{Holm2002,DarrylDHolm_2002}. A homogenisation theorem for McKean-Vlasov SDE, where the expectation operation $\bbE$ and more general dependence on the underlying probability measure of the solution appearing in the evolution equation has been considered in \cite{HONG2025405}.  

We may thus define an action integral and perform \emph{deterministic} variations from the theory of calculus of variations to produce an equation relating the Lagrangian $\bar{\ell}$ and $\bar{u} \in C\left([0,T]; \mathfrak{X}_{C^n}(\bbT^d)\right), \bar{a} \in C\left([0,T]; V^* \right)$. 
\begin{align}
    S[\bar{u}_t , \bar{a}_t] := \int_{0}^{T} \bar{\ell}(t, \bar{u}_t, \bar{a}_t)\dd t\,.
\label{sleep_lag}
\end{align}
The equations of motion are derived from the deterministic Euler--Poincar\'e constrained variational principle. Variations of the action \eqref{sleep_lag} are of the form,
\begin{align}
    \delta \bar{u}_{t} = \partial_t v - \ad_{\bar{u}_{t} }v, \quad \delta \bar{a}_{t} = - \mathcal{L}_v \bar{a}_{t} \,,  \label{lin}
\end{align}
where $v\in C^1([0, T]; \mathfrak{X}(\mcal{D}))$ is arbitrary, vanishing at $t = 0, T$ and non random. These are the classical Lin constraints seen in \cite{Bretherton_1970, CENDRA198763}.

\begin{proposition}\label{prop:det EPA}
The variation $\delta S = 0$ for the action \eqref{sleep_lag} under the Lin constraints \eqref{lin} imply the Euler--Poincar\'e equation,
\begin{align}
   \partial_t \frac{\delta \bar{\ell}}{\delta \bar{u}_{t}} + \mathcal{L}_{\bar{u}_{t } } \frac{\delta \bar{\ell}}{\delta \bar{u}_{t}} \, = \frac{\delta \bar{\ell}}{\delta \bar{a}_{t}} \diamond \bar{a}_{t } \,.\label{eq:ELASALT EP 1} 
\end{align}
\end{proposition}
This is a standard application of the Euler--Poincar\'e theorem of deterministic geometric mechanics \cite{HMR1998}, (see also Chapters 7 and 11 in \cite{holm2009geometric}). 

Furthermore, from the fact that $\bar{u}_t$ is deterministic, one notes,
\begin{align}
\begin{split}
    \frac{\delta \bar{\ell}}{\delta \bar{u}_{t}}(\bar{u}_t, \bar{a}_t) &= \frac{\delta }{\delta \bar{u}_{t}} \bbE \left[ \ell^{\Xi} \right](\bar{u}_t, \bar{a}_t) =  \bbE \left[ \frac{\delta \ell^{\Xi}}{\delta \bar{u}_{t}}  \right](\bar{u}_t, \bar{a}_t) = \bbE \left[ \Ad^*_{\Xi_t} \circ \frac{\delta \ell}{\delta u_{t}} \circ \Ad_{\Xi_t} \right](\bar{u}_t, \bar{a}_t)\,, \\
    \frac{\delta \bar{\ell}}{\delta \bar{a}_{t}}(\bar{u}_t, \bar{a}_t) &= \frac{\delta }{\delta \bar{a}_{t}} \bbE \left[ \ell^{\Xi} \right](\bar{u}_t, \bar{a}_t) =  \bbE \left[ \frac{\delta \ell^{\Xi}}{\delta \bar{a}_{t}}  \right](\bar{u}_t, \bar{a}_t) = \bbE \left[ {\Xi^*_t} \circ \frac{\delta \ell}{\delta a_{t}} \circ \Xi_{t *} \right](\bar{u}_t, \bar{a}_t) \,.
\label{eq:ellbarvarderivative}
\end{split}
\end{align}
The first and second equalities hold using $\bar{\ell} := \bbE [ \ell^\Xi ] $ and commuting through the variation with respect to deterministic variables. The final equality holds from the same computation seen in Equation \eqref{Xichainrulevariation}.  

One can show that if $\ell$ is a hyperregular Lagrangian, then smoothness and convexity conditions are preserved by the expectation. Consequently $\bar{\ell}$ is hyperregular and $\bar{u}$ is recoverable from the momentum $\frac{\delta \bar{\ell}}{\delta \bar{u}_{t}}$. In the case the Lagrangian is quadratic in its $\mathfrak{X}(\mcal{D})$ argument, the relation is explicit as the variational derivative reduces to $\frac{\delta \ell}{\delta u}(u_t, a_t) = L(a_t) u_t $ for some operator $L$ depending on advected quantities. It follows that, 
\begin{align*}
    \frac{\delta \bar{\ell} }{\delta \bar{u}}(t, \bar{u}_t, \bar{a}_t) = \bbE \left[ \Ad^*_{\Xi_t} \circ L(\Xi_t, \bar{a}_t) \circ \Ad_{\Xi_t} \right] (\bar{u}_t) := M(\Xi_t, \bar{a}_t ) \bar{u}_t \,.    
\end{align*}
One can then invert the (linear) operator $M$ to recover the mean velocity from its expected momentum, that is, $\bbE \left[M (\Xi, \bar{a}) \right]^{-1} \frac{\delta \bar{\ell}}{ \delta \bar{u}} = \bar{u}$.
As we shall see in Section \ref{subsec:avgeuler}, for the Euler's incompressible fluid, $\bar{a} = \bar{D}$ such that $M(\Xi, \bar{D}) = \bar{D}^{-1}\Xi^* \left( \Xi_* (\cdot) \right)^\flat$ where $\bar{D}$ is the mean volume density of Euler's incompressible fluid.
\begin{remark}
    When $\ell$ is $\Xi_t$ invariant in the sense of Proposition \ref{energyconservethm}. One has $\ell^{\Xi}(\bar{u}_t) = \ell(\bar{u}_t) = \bar{\ell}(\bar{u}_t)$. Thus, it follows that the deterministic closure for $\bar{u}_t$ coincides with the stochastic closure discussed in Chapter \ref{sec:VP1}.
\end{remark}
As in Section \ref{subsec:ham}, the Legendre transform of $\bar{\ell}$ is a time-dependent Hamiltonian and so energy need not be conserved while other conserved quantities remain. The Poisson structure discussed in Corollary \ref{casimirlemma} is maintained where now the momentum is given by $\frac{\delta \bar{\ell}}{\delta \bar{u}}$\footnote{This is not to be confused with the momentum named $\bar{m} = \frac{\delta \ell^\Xi}{\delta \bar{u}}$ in Section \ref{subsec:ham}, which is coupled to a random $\bar{u}$ dynamics.}. It follows that equation \eqref{eq:ELASALT EP 1} conserves the Casimir $C(\frac{\delta \bar{\ell}}{\delta \bar{u}} ) $, where $C$ is a Casimir for the Lie Poisson bracket \eqref{LPdef}.

A Kelvin circulation theorem is also present under this closure which can be interpreted as  the expectation of the circulation theorem for a stochastic fluid. A similar interpretation is seen in the LAEP and GLM type theories for composite flow maps $g = \Xi \circ \bar{g}$, \emph{provided the advected quantities $\bar{a}_t = \bar{g}_{t *} a_0$ are assumed invariant under average}. Importantly, this occurs in the deterministic closure for $\bar{u}$ but not the stochastic closure discussed in previous sections, as it is necessary that the expectation operation $\bbE$ to commutes through mean variables. 

We may write, with the identities \eqref{eq:ellbarvarderivative}, the Kelvin circulation theorem for $\frac{1}{\bar{D}_t}\frac{\delta \bar{\ell} }{\delta \bar{u}_t}$ in terms of an averaged circulation theorem for variables $u_{t} := \Ad_{\Xi_t} \bar{u}_{t}$, $a_{t} := \Xi_{t *}\bar{a}_{t}$, 
\begin{align*}
    \partial_t \oint_{\overline{g}_t c_0} \frac{1}{\overline{D}_t} \frac{\delta \bar{\ell}}{ \delta \bar{u}_t} = \oint_{\overline{g}_t c_0}\frac{1}{\overline{D}_t} \frac{\delta \bar{\ell}}{ \delta \bar{a}_t}  \diamond \overline{a}_t
    &\iff \partial_t \oint_{\overline{g}_t c_0} \frac{1}{\overline{D}_t} \bbE \left[\Xi_t^* \frac{\delta \ell}{ \delta u_t} \right]  = \oint_{\overline{g}_t c_0}\frac{1}{\overline{D}_t} \bbE \left[\Xi_t^* \left( \frac{\delta {\ell}}{ \delta a_t}   \diamond a_t \right)  \right] \\[10pt]
    &\iff \bbE \left[ \dd \oint_{g_t c_0} \frac{1}{D_t}  \frac{\delta \ell}{ \delta u_t}  \right] = \bbE \left[  \oint_{g_t c_0} \frac{1}{D_t} \frac{\delta \ell}{ \delta a_t} \diamond a_t \dd t \right] \, .
\end{align*}
Thus, the total circulation equals the expected circulation of a stochastic fluid generated by the composition of maps $g_t = \Xi_t \circ \bar{g}_t$, \emph{with $\bar{g}$ closed deterministically}. A similar ``statistical" Kelvin theorem was proposed in the LA-SALT theory of \cite{Drivas2020}. The major difference in LA-SALT theory to the present section is that the velocity vector field in averaged in \cite{Drivas2020}, while the momentum one-form is averaged here.

\section{Examples} \label{sec:examples}
In this Section, we will consider the illustrative example of incompressible Euler fluid equations as special cases of the random-coefficients and stochastic Euler--Poincar\'e equations discussed in Section \ref{sec:VP1}. Then, we illustrate the energy preserving properties of stochastic perturbations generated by isometries. Lastly, we turn to the averaged incompressible Euler equations as a special case of the averaged Euler--Poincar\'e equations derived in Section \ref{sec:det vp}.

\subsection{Stochastic incompressible Euler equations}\label{sec:Euler}
On a $d$-dimensional Riemannian manifold $(\mcal{D}, \mathbf{g})$ with Riemannian volume form $\mu_{\mathbf{g}}$, we consider the incompressible Euler equations where the fluid configuration manifold is $\mathfrak{X}(\clD)\times \Lambda^d(\clD)$\footnote{Of course, one may take the alternative approach by taking the configuration manifold as the manifold of divergence-free vector fields. We will use the current setup to illustrate the construction using advected densities.}. Let $\Xi : \Omega \times [0,T]\rightarrow\operatorname{Diff}(\clD)$ be the fixed stochastic flow of diffeomorphisms defined by equation \eqref{eqn:d Xi d gbar def}.
Let $\bar{u}\,, u: \Omega\times [0,T]\rightarrow \mathfrak{X}(\mcal{D})$ be the mean and drift velocity vector fields of the fluid, respectively, which are related by $u_t:= \Xi_{t*}\bar{u}_t$. Let $\bar{D}\,, D :\Omega \times [0,T]\rightarrow \Lambda^d(\mcal{D})$ be defined by $\bar{D} := \bar{\rho}\mu_{\mathbf{g}}$ and $D_t := \Xi_{t*}\bar{D}_t$ be the mean advected volume density and full advected volume density, respectively, for some $\bar{\rho}:\Omega\times[0,T] \rightarrow \Lambda^0(\clD)$. The kinetic energy Lagrangian $\ell_E: \mathfrak{X}(\clD)\times \Lambda^d(\clD) \rightarrow \mathbb{R}$ for Euler's equation can be expressed as
\begin{align}
    \ell_E(u, D) = \int_{\mc{D}} \frac{1}{2}\mathbf{g}(u, u) D = \int_{\mc{D}} \frac{1}{2}\mathbf{g}(\Xi_{t*}\bar{u}, \Xi_{t*}\bar{u}) \Xi_{t*} \bar{D} = \ell_E(\Xi_{t*}\bar{u}, \Xi_{t*}\bar{D}) = \ell_E^\Xi(\bar{u}, \bar{D}) \,, \label{kineticenergylag}
\end{align}
where the equivalent time-dependent Lagrangian $\ell_E^\Xi:\Omega \times [0,T] \times \mathfrak{X}(\clD)\times \Lambda^d(\clD) \rightarrow \mathbb{R}$ on the mean variables $\bar{u}$ and $\bar{D}$ can be expressed using a random time-dependent metric,
\begin{align}
    \ell_E^{\Xi}(\bar{u}, \bar{D}) = \int_{\clD}\frac{1}{2}\mathbf{g}(\Xi_{t*}\bar{u},\Xi_{t*}\bar{u})\Xi_{t*}\bar{D} = \int_{\clD}\frac{1}{2}(\Xi^*_{t}\mathbf{g})(\bar{u},\bar{u})\Xi_{t*}\bar{D}\,. \label{eq:Euler lag mean}
\end{align}
Here, $\Xi_t^*\mathbf{g} = \Xi_t^* \left(\mathbf{g}_{ij}(x)\, dx^i\otimes dx^j \right)$ is the pullback metric induced by $\Xi$ on $\mathbf{g}$.

Incompressibility condition of the \emph{mean} velocity field $\bar{u}$ can be obtained by imposing the volume-preserving constraint $\Xi^*_t D_t = \bar{D}_t = \mu_\mathbf{g}$.
To enforce the volume-preserving constraint in the variational principle, we introduce a scalar semi-martingale Lagrange multiplier $\dd \bar{P}$ defined by 
\begin{align}
\dd\bar{P} := \bar{p} \,\dt + \sum_{k=1}^K \bar{p}^{(k)}\circ \dW^k_t\,, \quad \bar{p}\,,\bar{p}^{(k)} :\Omega \times [0,T] \rightarrow \Lambda^0(\mcal{D})\,,  \label{eq:semimartingale pressure}
\end{align}
where $\bar{p}$ and $\bar{p}^{(k)}$ are to be interpreted as the pressure functions enforcing the constraint for each component of the $K$-dimensional Brownian motion $W = (W^1,\ldots, W^K)$.
The use of semi-martingale pressure is required due to the semi-martingale nature of the variational principle as shown in \cite{MR2004, hofmanova2019navier, SC2019, CRISAN2022109632}, where the collection of driving Brownian basis of the semi-martingale pressure is the same as the Brownian basis appearing in equation \eqref{eqn:d Xi d gbar def}.
To obtain the stochastic Euler equations, we apply the Euler--Poincar\'e variational principle defined in Corollary \ref{cor:EP equivalent} to the action
\begin{align}
    S = \int_{0}^{T} \ell(u_t, D_t)\,\dt - \scp{\dd \bar{P}}{\Xi_t^* D_t - \mu_{\mathbf{g}}}\,,  
\end{align}
with the constrained variations of $u_{t}$ and $D_t$, as well as additional \emph{free} variations of the pressure $\dd \bar{P}$.
Then, we calculate the variations of the action as
\begin{align*}
    \begin{split}
        0 = \delta S &= \delta \int_{0}^{T} \left[\int_{\mc{D}} \frac{1}{2} D_t\mathbf{g}({u}_t, {u}_t) \,\dt - \int_{\mc{D}} \dd \bar{P}\,(\Xi^*_tD_t - \mu_{\mathbf{g}}) \right]\\
        &= \int_{0}^{T} \scp{u_t^\flat \otimes D_t}{\delta u_t}\,\dt + \scp{\frac{1}{2}\mathbf{g}(u_t, u_t)\,\dt - \Xi_{t*}\left(\dd \bar{P}\right)}{\delta D_t} - \scp{\delta \dd \bar{P}}{\Xi_{t}^*D_t - \mu_{\mathbf{g}}} \\
        &= \int_{0}^{T}\scp{u_t^\flat \otimes D_t}{\Xi_{t*}\left(\p_t v_t - \ad_{\Xi^*_t u_t}v_t\right)}\,\dt + \scp{\frac{1}{2}\mathbf{g}(u_t, u_t)\,\dt - \Xi_{t*}\left(\dd \bar{P}\right)}{-\Xi_{t*}\mcal{L}_v (\Xi_t^*D_t)}\\
        &\qquad \qquad - \scp{\delta \dd \bar{P}}{\Xi_{t}^*D_t - \mu_\mathbf{g}} \\
        &= \int_{0}^{T}\scp{\left(\dd + \ad^*_{\Xi^*_t u_t\,\dt}\right)\left(\Xi^*_t \left(u_t^\flat \otimes D_t\right) \right) - \Xi^*_t D_t\otimes\mb{d}\left(\frac{1}{2}\Xi^*_t\mathbf{g}(u_t, u_t)\,\dt - \dd \bar{P}\right)}{v}\,\dt \\
        & \qquad\qquad + \scp{\delta \dd \bar{P} }{\Xi^*_t D_t - \mu_{\mathbf{g}}}\,.
    \end{split}
\end{align*}
Here, we made use of the musical isomorphism $\flat:\mathfrak{X}(\mcal{D}) \rightarrow \Lambda^1(\mcal{D})$ induced by the right-invariant (weak) Riemannian metric $\scp{\cdot}{\cdot}_{L_2}$ defined by $\scp{w}{v}_{L_2} = \int_{\clD}\mathbf{g}(w, v)\mathbf{\mu}_{\mathbf{g}} = \scp{w}{v^\flat}$, for all $w, v \in \mathfrak{X}(\mcal{D})$, as well as its dual $\sharp:\Lambda^1(\mcal{D}) \rightarrow \mathfrak{X}(\mcal{D})$. Setting the variations to zero and applying the stochastic fundamental Lemma of calculus of variations \cite{crisan2022variational, ST2023}, we obtain the random-coefficient Euler--Poincar\'e equation and the volume-preserving constraint with the advection of $\bar{D}_t$,
\begin{align}
    &\left(\dd + \mathcal{L}_{\bar{u}_t\,\dt}\right)\left( \Xi^*_t  \left(\Xi_{t*}\bar{u}_t\right)^\flat \otimes \bar{D}_t \right) = \mathbf{d} \left(\frac{1}{2}\Xi^*_t\mathbf{g}(\Xi_{t*}\bar{u}_t, \Xi_{t*}\bar{u}_t)\,\dt - \dd \bar{P} \right) \otimes  \bar{D}_t \,, \label{eq:mean euler EP} \\
    &\left(\dd + \mathcal{L}_{\bar{u}_t}\right)\bar{D}_t = 0\,, \quad \bar{D}_t = \mu_{\mathbf{g}}\,.\label{eq:mean euler Dbar advect}
\end{align}
Using the condition $\bar{D}_t = \mu_\mathbf{g}$, we obtain the divergence-free condition of the vector field $\bar{u}_t$, $\mathcal{L}_{\bar{u}_t}\mu_{\mathbf{g}} = \mathbf{d}({\bar{u}_t}\intprod \mu_{\mathbf{g}}) = 0$ from \eqref{eq:mean euler Dbar advect}. Dividing through by the advected mean density $\bar{D}$ in the sense of Remark \ref{rmk:division by density}, equation \eqref{eq:mean euler EP} can be simplified to
\begin{align}
    \left(\dd + \mathcal{L}_{\bar{u}_t\,\dt} \right)\Xi^*_t \left(\Xi_{t*}\bar{u}_t\right)^\flat = \mathbf{d}\left(\frac{1}{2}\Xi^*_t \mathbf{g}(\Xi_{t*}\bar{u}_t, \Xi_{t*}\bar{u}_t)\,\dt - \dd \bar{P}\right)\,,
\label{meaneuler}
\end{align}
which is the Euler's fluid equation for a time-dependent metric. We remark that when $\Xi$ is chosen to be stochastic flow of isometries, we have $\Xi^*_t (\Xi_{t*} \bar{u}_t )^\flat = \bar{u}_t^\flat$ and \eqref{meaneuler} reduces to the deterministic Euler equations. Furthermore, when $\Xi$ is chosen to be volume-preserving $\Xi^*\mu_{\mathbf{g}} = \mu_{\mathbf{g}}$, we have the random time-dependent Lagrangian $\ell_E^\Xi$ reduces to the autonomous kinetic energy Lagrangian $\ell_E$.

We express the mean Euler's equation \eqref{meaneuler} in local coordinates. For $X = X^i\p_i \in \mathfrak{X}(\clD)$ and $\alpha = \alpha_i\, \mathbf{d}x^i \in \Lambda^1(\clD)$, the coordinate expression of the Lie-derivative $\mcal{L}_X \alpha$ is expressed 
\begin{align*}
    \mcal{L}_{X} \alpha = X \intprod \mb{d}\alpha + \mb{d}\left(X \intprod \alpha \right) = \left(X^i\p_{i} \alpha_k + \alpha_i \p_{k} X^i\right)\,dx^k\,.
\end{align*}
Let the vector $\bar{\bu}$ denote the coefficients of the mean transport vector field $\bar{u} = \bar{\bu}\cdot\nabla$, and the vector $\wt{\bu}$ denote the coefficients of the momentum $\wt{u}_i\, \mathbf{d}x^i := \Xi^*_t\left(\Xi_{t*}\bar{u}\right)^\flat : \Omega \times [0,T] \rightarrow \Lambda^1(\clD)$. The coordinate expression of $\wt{u}_i\mathbf{d}x^i$ is given by
\begin{align*}
    \wt{u}_i \mathbf{d}x^i := \Xi^*_t\left(\Xi_{t*}\bar{u}\right)^\flat = \mathbf{g}_{ik}\frac{\p \Xi^i}{\p x^j}\frac{\p \Xi^k}{\p x^l}\bar{u}^j dx^l\,,
\end{align*}
and the mean Euler equations appearing in \eqref{meaneuler} can be written as 
\begin{align*}
    \dd \wt{\bu} + \bar{\bu}\cdot \nabla \wt{\bu} + \wt{u}^i \nabla \bar{u}^i = \nabla \left(\frac{1}{2}\mathbf{g}\left(\bar{\bu}\cdot \nabla \Xi, \bar{\bu}\cdot \nabla \Xi\right)\,\dt - \dd \bar{P}\right)\,, \quad \operatorname{div}(\bar{\bu}_t) = 0\,.
\end{align*}

To express Euler's equation in terms of the variable $u_{t}$, we  apply Propositions \ref{prop:random to stoch EPA} to the Euler--Poincar\'e form of the mean Euler equations \eqref{eq:mean euler EP} and obtain the stochastic incompressible Euler equations as the SPDEs
\begin{align}
    &\left(\dd + \mathcal{L}_{\dd gg^{-1}}\right) \left(u_{t}^\flat \otimes D \right) = \mathbf{d}\left(\frac{1}{2}\mathbf{g}(u_{t}, u_{t})\,\dt - \dd P_t\right) \otimes D\,,  \quad \dd P_t := \Xi_{t*} \dd \bar{P} \label{eq:full euler}\\
    & \qquad \textrm{where} \quad \dd gg^{-1} := u_{t}\,\dt + \frac{1}{2}\sum \limits_{k,l=1}^K\Gamma^{kl}\left[\xi_k\,,\,\xi_l\right] \dd t  + \sum \limits_{k=1}^K\xi_k \circ \dd W_t^k\,, \label{eq:full dg g inv}
\end{align}
and the stochastic advection of the full volume density $D$ with the volume constraint
\begin{align}
    \left(\dd + \mathcal{L}_{\dd gg^{-1}}\right) D = 0\,, \quad \Xi_t^*D_t = \mu_{\mathbf{g}}\,.
\end{align}
As $D$ is advected by the same vector field as the full momentum one-form density $u_{t}^\flat\otimes D$, denoted by $\dd g g^{-1}$, we can simplify the momentum equation \eqref{eq:full euler} by dividing through by $D$ and using the Kunita-It\^o-Wentzel formula to have
\begin{align}
    \dd u_{t}^\flat + \mathcal{L}_{u_{t}}u_{t}^\flat\,\dt + \sum\limits_{k=1}^K\mcal{L}_{\xi_k} u_{t}^\flat \circ \dW^i_t  + \frac{1}{2}\sum\limits_{k,l=1}^K\mcal{L}_{\Gamma^{kl}\qv{\xi_k}{\xi_l}} u_{t}^\flat\,\dt = \mathbf{d}\left(\frac{1}{2}\mathbf{g}(u_{t}, u_{t})\,\dt - \dd P_t\right)\,.
\label{salteuler}
\end{align}
In the case where $\Xi : \Omega \times [0,T]\rightarrow \operatorname{SDiff}(\mcal{D})$, the group of volume-preserving diffeomorphisms, we automatically have $\operatorname{div}_{\mu_{\mathbf{g}}}\xi_k = 0$ for $k=1,\dots, K$. Noting that $\bar{D}_t = \mu_{\mathbf{g}}$ following the volume constraint, the advection of $D$ imply the divergence-free condition of $u_{t}$.
\begin{align}
\begin{split}
    0 = \left(\dd + \mathcal{L}_{\dd gg^{-1}}\right) \left(\Xi_{t*} \mu_{\mathbf{g}}\right) &= \mathcal{L}_{u_{t}} \mu_{\mathbf{g}}\,\dt + \sum\limits_{k=1}^K\mcal{L}_{\xi_k} \mu_{\mathbf{g}}\circ \dW^k_t + \frac{1}{2}\sum\limits_{k,l=1}^K\mcal{L}_{\Gamma^{kl}\qv{\xi_k}{\xi_l}} \mu_{\mathbf{g}} \\
    &= \mathcal{L}_{u_{t}} \mu_{\mathbf{g}} \,\dt \,,
\end{split}\label{eq:div free U}
\end{align}
where in the last equality we have used the fact that $\xi_k$ and $\qv{\xi_k}{\xi_l}$ are divergence-free. The stochastic equation \eqref{salteuler} and the incompressibility condition \eqref{eq:div free U} thus combine to define a stochastic incompressible Euler fluid equation for the drift vector field $u_{t}$. In the case where $\Gamma^{kl} = 0$, these equations coincide with the SALT incompressible Euler fluid equations with transport noise derived in \cite{Holm2015}, whose analytical properties were studied in e.g., \cite{CFH19}.

To explicitly write down the SPDE for $u_t$ in coordinate expression, we make use of the well known identity linking the Lie derivative and the Levi-Civita connection $\nabla:\mathfrak{X}\times \mathfrak{X}(\mcal{D})\rightarrow \mathfrak{X}(\mcal{D})$ \cite{DR1977}
\begin{align*}
    \mcal{L}_v v^\flat - \frac{1}{2}\mb{d}\mathbf{g}(v,v) = \left(\nabla_v v\right)^\flat\,,\quad \forall v\in \mathfrak{X}(\mcal{D})\,.
\end{align*}
and the Lie derivative expression 
\begin{align*}
    \left(\mcal{L}_\xi u^\flat\right)^\sharp = \sum_{i,j,l,q = 1}^d\left(\xi^j\p_j u^q + \mathbf{g}^{iq}\xi^j u^l \p_j \mathbf{g}_{li} + \mathbf{g}^{iq}\mathbf{g}_{lj}u^l\p_i \xi^j\right)\p_q\,, \quad \forall v, \xi \in \mathfrak{X}(D)\,.
\end{align*}
Then, we have the coordinate form of equation \eqref{salteuler} as
\begin{align*}
    \dd u_{t} + \nabla_{u_{t}} u_{t}\,\dt + \sum\limits_{k=1}^K\left(\mcal{L}_{\xi_k} u_{t}^\flat\right)^\sharp \circ \dW^k_t  + \frac{1}{2}\sum\limits_{k,l=1}^K\left(\mcal{L}_{\Gamma^{kl}\qv{\xi_k}{\xi_l}} u_{t}^\flat\right)^\sharp\,\dt = - \nabla \dd P_t\,, \quad \operatorname{div}_{\mu_{\mathbf{g}}}(u_t) = 0\,.
\end{align*}
It is possible to consider the explicit evolution of the mean velocity $\bar{u}$ by utilising the operator $\ad^\dagger : \mathfrak{X}(\mcal{D}) \times \mathfrak{X}(\mcal{D}) \rightarrow \mathfrak{X}(\mcal{D})$ defined as the $L^2$ dual of the $\ad$ operator. That is, for $u, v, w \in \mathfrak{X}(\mcal{D})$, we have 
\begin{align}
    \scp{u}{v}_{L^2} := \int_{\clD} \mathbf{g}(u,v)\mu_{\mathbf{g}}\,,\quad 
    \scp{\ad_u v}{w}_{L^2} = \scp{v}{\ad^\dagger_u w}_{L^2}\,,\quad \text{such that}\,\quad \ad^\dagger_u w := \left(\mathcal{L}_u w^\flat\right)^\sharp\,. \label{addagger}
\end{align}
Then, one can write the stochastic equation for $u_{t}$ by taking $\sharp$ and applying the definition \eqref{addagger}, 
\begin{align*}
    \dd u_{t} + \ad^\dagger_{u_{t}} u_{t}\,\dt + \sum\limits_{k=1}^K\ad^\dagger_{\xi_k}u_{t} \circ \dW^k_t + \frac{1}{2}\sum\limits_{k,l=1}^K\ad^\dagger_{\Gamma^{kl}\qv{\xi_k}{\xi_l}}u_{t}\,\dt = \left[\mathbf{d}\left(\frac{1}{2}\mathbf{g}(u_{t}, u_{t})\,\dt - \dd P_t\right)\right]^\sharp\,,
\end{align*}
Noting that $u = \Xi_{t*}\bar{u}$, applying the the Kunita-It\^o-Wentzel formula for vector fields 
\begin{align*}
    \dd \left(\Xi_{t*}\bar{u}_t\right) = \Xi_{t*}\dd \bar{u}_t + \ad_{\dd \Xi_t \Xi_t^{-1}} \Xi_{t*}\bar{u}_t = \Xi_{t*}\dd \bar{u}_t - \qv{\dd \Xi_t\Xi^{-1}_t}{\Xi_{t*}\bar{u}_t}\,,
\end{align*}
we obtain the evolution of $\bar{u}_t$
\begin{align}
\begin{split}
    &\dd \bar{u}_{t} + \Xi^*_{t}\left(\ad^\dagger_{u_{t}} u_{t}\,\dt + \sum\limits_{k=1}^K\left(\ad^\dagger_{\xi_k} + \ad_{\xi_k}\right)u_{t} \circ \dW^k_t + \frac{1}{2}\sum\limits_{k,l=1}^K\left(\ad^\dagger_{\Gamma^{kl}\qv{\xi_k}{\xi_l}} + \ad_{\Gamma^{kl}\qv{\xi_k}{\xi_l}}\right) u_{t} \,\dt\right) \\
    & \qquad  = \Xi^*_{t}\left[\mathbf{d}\left(\frac{1}{2}\mathbf{g}(u_{t}, u_{t})\,\dt - \dd P_t\right)\right]^\sharp\,.
\end{split}\label{eq:ubar dynamics}
\end{align}

\begin{remark}
    One may in fact do away with the assumption that $\Xi$ is a flow volume-preserving diffeomorphisms. In such a case, the pressure martingale can be determined by taking the divergence of the $\bar{u}$ evolution equation \eqref{eq:ubar dynamics} since $\bar{u}$ is divergence-free from equation \eqref{eq:mean euler Dbar advect}. 
\end{remark}

The two equivalent formulations of the Euler equations \eqref{meaneuler} and \eqref{salteuler} also agree in their Kelvin theorems, as expected. For a given initial material loop $c_0$, one has conservation of the circulation integral of $\Xi^{*}_t ({\Xi_t}_* \overline{u}_t)^\flat$ and $u^\flat$,
\begin{align}
    &\dd \oint_{\overline{g}_t c_0} \Xi^{*}_t ({\Xi_t}_* \overline{u}_t)^\flat = \oint_{\overline{g}_t c_0} \mathbf{d}\left(\frac{1}{2}\Xi^*_t\mathbf{g}(\Xi_{t*}\bar{u}_t, \Xi_{t*}\bar{u}_t)\,\dt - \dd\bar{P}\right)= 0\,,  \label{meankelvin}\\
    &\dd \oint_{g_t c_0} u_t^\flat = \oint_{g_t c_0}  \mathbf{d}\left(\frac{1}{2}\mathbf{g}(u_{t}, u_{t})\,\dt - \dd P \right)\,\rmd t = 0\, \,.   \label{saltkelvin}
\end{align}
The equivalence of \eqref{meankelvin} and \eqref{saltkelvin} is an easy consequence of the change of variables formula.

\paragraph{Vorticity dynamics} Here, we consider the vorticity dynamics associated with the random-coefficient and stochastic homogeneous incompressible Euler equations \eqref{meaneuler} and \eqref{salteuler}. Let $q_{t} = \mb{d} u_{t}^\flat = \mb{d}(\Xi_{t*}\bar{u}_{t})^\flat \in \Lambda^2(\mcal{D})$ be the vorticity of the drift velocity one-form and let the vorticity associated with mean velocity one-form be $\bar{q}_{t} = \mb{d}\Xi^*_t u_{t}^\flat = \Xi^*_t q_{t}$. Applying the exterior derivative $\mb{d}$ to \eqref{meaneuler} and \eqref{salteuler} yields
\begin{align}
    &\left(\dd + \mathcal{L}_{\bar{u}_t\,\dt} \right)\bar{q}_{t} = 0\,,\label{vorticityequiv1} \\
    &\dd q_{t} + \mathcal{L}_{u_{t}}q_{t}\,\dt + \sum\limits_{k=1}^K\mcal{L}_{\xi_k} q_{t} \circ \dW^i_t + \frac{1}{2}\sum\limits_{k,l=1}^K\mcal{L}_{\Gamma^{kl}\qv{\xi_k}{\xi_l}} q_{t}\,\dt = 0 \,, \label{vorticityequiv2}
\end{align}
respectively, which are the evolution of the drift vorticity two-form and its pullback. The equivalence of the vorticity equations \eqref{vorticityequiv1}, \eqref{vorticityequiv2} can be verified using the stochastic Lie chain rule \eqref{eq:stoch lie chain} and the composite transport identity \eqref{eq:comp transport}. To recover $\bar{u}_t$ and $u_t$ from $\bar{q}_t$ and $q_t$, respectively, two \emph{different} Biot--Savart laws need to be used. Namely, we have 
\begin{align}
    u_t = \sharp \delta (-\triangle)^{-1} q_t\,, \quad \bar{u}_t = \Xi_{t}^* \left(\sharp \delta (-\triangle)^{-1} \Xi_{t*}\bar{q}_t\right)\,. \label{eq:random Biot--Savart}
\end{align}
That is, $\bar{u}_t$ and $\bar{q}_t$ are related by a random-coefficient Biot--Savart law whose randomness are generated by the stochastic flow of diffeomorphisms $\Xi_t$.
\begin{remark}[]
    We remark that the mean vorticity equation \eqref{vorticityequiv1} with the random-coefficient Biot--Savart law \eqref{eq:random Biot--Savart} form a closed set of random-coefficient PDE. This is in contrast however, to the mean velocity formulation which is given by the SPDE \eqref{meaneuler} where the stochastic integrals exist purely as pressure forces.
\end{remark}
\paragraph{Helicity dynamics} Working in three dimensions, $\mcal{D} \subset \mathbb{R}^3$, we define the helicity
\begin{align*}
    \Lambda(u_t^\flat) := \int_{\mcal{D}} u_{t}^\flat \wedge \mb{d}u_{t}^\flat = \int_{\mcal{D}} u_{t}^\flat \wedge q_{t}\,,
\end{align*}
which measures the linkage number of the vorticity field lines. Through direct calculations, we have 
\begin{align*}
    \dd \left(\Xi^*_t u_{t}^\flat \wedge \bar{q}_{t} \right) + \mathcal{L}_{\bar{u}_t} \left(\Xi^*_t u_{t}^\flat \wedge \bar{q}_{t} \right) \,\dt = -\mb{d}\left(\frac{1}{2}\Xi^*_t\mathbf{g}(\Xi_{t*}\bar{u}_t, \Xi_{t*}\bar{u}_t)\,\dt - \dd\bar{P}\right)\wedge \Xi^*_t \bar{q}_{t}\,,\\
    \dd \left(u_{t}^\flat \wedge q_{t} \right) + \mathcal{L}_{\dd g_t g_t^{-1}} \left( u_{t}^\flat \wedge q_{t} \right) = -\mb{d}\left(\frac{1}{2}\mathbf{g}(u_t, u_t)\,\dt - \dd P\right)\wedge q_{t}\,,
\end{align*}
where $\dd g_tg_t^{-1}$ is defined in \eqref{eq:full dg g inv}. Noting that $\mcal{L}_X \rho = \mb{d}(X \intprod \rho)$ for all $X \in \mathfrak{X}(\mcal{D})$ and $\rho \in \Lambda^d(\mcal{D})$, we have the conservation of the helicity 
\begin{align*}
    \dd \Lambda(u_t^\flat) = \dd \Lambda(\Xi_t^*u_t^\flat) = 0\,, 
\end{align*}
by the two Euler equations, \eqref{meaneuler} and \eqref{salteuler} in three-dimensions. 

In two dimensions, stochastic Euler's equations of the form \eqref{eq:full euler}-\eqref{eq:full dg g inv} were extensively studied in \cite{diamantakis2024levy} by selecting specific $\xi_k$ such that the system exhibits deterministic behaviour in a time-dependent stochastic coordinate frame. We cast the results appearing in \cite{diamantakis2024levy} through the insights of the current paper. In two dimensions, the exterior derivative operator $\mb{d}$ on velocity one-forms can be expressed by perpendicular gradients $\nabla^\perp = (-\p_y, \p_x)^T$ in local coordinates. The vorticity two-form can then be identified with a scalar function via Hodge duality and expressed as $q_{t} = \nabla^\perp \cdot \bu_{t}\,d^2x$, where $\bu_{t}$ is the coefficient of the vector field $u_{t}$. 
We consider the equivalent energies $h^{\Xi}(\Xi_{t*} q_{t}) \equiv h^{\Xi}(\bar{q}_{t}) = h(q_{t})$, formulated in mean variables as per equations \eqref{vorticityequiv1} and \eqref{eq:random Biot--Savart}, and expressed as
\begin{align}
    h(q) := \int_{\mcal{D}} - \frac{1}{2} q \Delta^{-1} q\, d^2 x\,, \qquad h^{\Xi}(\bar{q}) := \int_{\mcal{D}} - \frac{1}{2} \Xi_{t*} \bar{q} \Delta^{-1} \Xi_{t*}\bar{q} \,d^2 x\,.\label{comvorticityham}
\end{align}
Note that $t$-dependence in $h^{\Xi}$ occurs generally since $\Xi_{t}^*$ and $\Delta^{-1}$ does not generally commute. It follows that $\bar{q}_{t}$ is not dual to its stream function in a time-independent manner. 

An important special case of the mean equations of vorticity, \eqref{vorticityequiv1} and \eqref{eq:random Biot--Savart}, occurs when considering a point vortex solution ansatz $q_{t} = \sum_{\alpha} \Gamma_\alpha \delta(x;x_\alpha(t) )$. When $\Xi_t$ is an isometry at all $t$, its pullback does commute with the Laplace-de Rham operator. It follows that $h(q_{t}) = h^{\Xi}(\bar{q}_{t}) = h(\bar{q}_{t})$ becomes time independent and conserved. The calculus of distributions\footnote{To compute this we make use of the identity found in Example 6.1.3 in \cite{H_rmander_2003} and the fact that $\mcal{J}(\Xi_t^{-1}) = 1$ and $\Xi_t$ is a bijection.} implies $\bar{q}_t = \Xi^*_t q_{t} = \mcal{J}(\Xi_t^{-1})\bar{q}_{t} \circ \Xi_t = \sum_{\alpha} \Gamma_\alpha  \delta (\Xi_t(x) ,x_\alpha(t) ) = \sum_{\alpha} \Gamma_\alpha  \delta (x,\Xi^{-1}_t(x_\alpha(t)) )$.
It is well known in the literature \cite{aref2007point,poincare1893theorie,grobli1877specielle} that stable configurations exist in the noiseless case ($\Xi_t \equiv e$ for \eqref{comvorticityham}) around a center of vorticity $x_c(t)$. Through symmetry arguments, the persistence of equilateral triangles configuration of unit strength point vortices when considering $\xi_1(x_1, x_2) = (x_2,-x_1)^T Ar \exp (-\frac{r}{2} \|x - x_c(t)\|^2), \,\,\xi_2(x_1, x_2) = (-b,a)$ for given $a,b \in \bbR, A, r > 0$ is shown in \cite{diamantakis2024levy} for a stochastically rotating and translating frame.

This result may be explained through the results of this paper. The vector fields $\xi_k$ are defined with respect to coordinates of the point vortices, which remain at $\|x_\alpha(t) - x_c(t)\| = const$ if initially configured in an equilateral triangle around $x_c(0)$. In this case, $x \mapsto \xi_k(x)$ is a linear map for a fixed element of $\mathfrak{se}(2)$ when restricted to $\|x_\alpha(t) - x_c(t)\| = const$ and acts as Killing vector field of $\bbR^2$ for all $k$. It follows that the Killing fields integrated against the Stratonovich differentials $\circ \dd W^k_t$ generate an isometry-valued stochastic process of the metric. Since $\Xi_t$ is a flow of isometries, the methods of Proposition \ref{energyconservethm} imply that $\bar{x}_\alpha := \Xi_t^{-1} (x_\alpha)$ is a solution of a deterministic equation, and the stochastic dynamics are simply that of the noiseless case in the stochastic coordinates of $\Xi_t$. Energy conservation is implied by Proposition \ref{energyconservethm} and verified numerically in \cite{diamantakis2024levy}.

\subsection{Averaged incompressible Euler equations}\label{subsec:avgeuler}  
The deterministic $\bar{u}$ closure of the Euler equations, as described for the general case in Section \ref{sec:det vp}, may be derived by applying the expectation operator to the Lagrangian $\ell_E^{\Xi}$ in \eqref{eq:Euler lag mean} to obtain an averaged Lagrangian $\bar{\ell}_E : [0,T]\times \mathfrak{X}(\mcal{D})\times \Lambda^d(\clD) \rightarrow \mathbb{R}$, $\bar{\ell}_E(\bar{u}, \bar{D}) : = \bbE[\ell_E^\Xi(\bar{u},\bar{D})]$. Here, the arguments of $\bar{\ell}_E$ are assumed to be deterministic flows of vector field $\bar{u}$ and $d$-form $\bar{D}$, generated by a deterministic flow of diffeomorphisms $\bar{g}: [0,T] \rightarrow \operatorname{Diff}(\mcal{D})$ as required by Assumption \ref{weakgbarflow}. That is, $\bar{u}:[0,T]\rightarrow \mathfrak{X}(\mcal{D})$ where $u_t = \frac{\dd}{\dt}\bar{g}_t\,\bar{g}_t^{-1}$ and $\bar{D}: [0,T] \rightarrow \Lambda^d(\clD)$ where $\bar{D}_t = D_0 \bar{g}_t^{-1}$ for some positive, non-zero $D_0 \in \Lambda^d(\clD)$. Special attention is required when considering the volume-preserving constraint $\bar{D} = \mu_{\mathbf{g}}$ under expectation. In this section, we make the simplifying assumption that the Lagrange multiplier, represented by the semi-martingale pressure $\Pi$, decomposes into $\dd \Pi_t = \pi_t \,\dt + \sum_{k=1}^K \pi^{(k)}_t dW^k_t$, where $\pi, \pi^{(k)}: \Omega \times [0,T]\rightarrow \Lambda^0(\mcal{D})$. In comparison with the semi-martingale pressure $\bar{P}$ defined in equation \eqref{eq:semimartingale pressure}, the semi-martingale decomposition of $\Pi$ is defined in It\^o sense in anticipation for taking expectation. Let $\bar{\pi} : = \bbE[\pi]$. Under this choice of pressure, the volume-preserving constraint can be enforced under expectation and we define the action $\bar{S}$ as the following
\begin{align}
\begin{split}
        \bar{S}[\bar{u}_t, \bar{D}_t] &= \int_{0}^T \bar{\ell}_E (\bar{u}_t,\bar{D}_t)\,\dt - \int_0^T\bbE\left[\int_{\clD}\dd \Pi_t\left(\bar{D}_t - \mu_{\mathbf{g}}\right) \right]\\
        &= \int_0^T\bbE\left[\int_{\clD}\frac{1}{2}(\Xi^*_t\mathbf{g})(\bar{u}_t,\bar{u}_t)\Xi_{t*}\bar{D} \,\dt\right] - \int_0^T\int_{\clD}\bar{\pi}_t\left(\bar{D}_t - \mu_{\mathbf{g}}\right)\,\dt\,.
\end{split}\label{eq:averg euler action 1}
\end{align}
Variational derivatives of $\bar{\ell}_E$ can be calculated as 
\begin{align*}
    \frac{\delta \bar{\ell}_E}{\delta \bar{u}_t} = \bbE[(\Xi^*_t\mathbf{g})(\bar{u}_t,\cdot) \otimes \Xi_{t*}\bar{D}]\,,\quad \frac{\delta \bar{\ell}_E}{\delta \bar{D}_t} = \bbE[\Xi^*_t(\Xi^*_t\mathbf{g})(\bar{u}_t,\bar{u}_t)]\,,
\end{align*}
where the $(\Xi^*_t\mathbf{g})(\bar{u}_t,\cdot)\in \Lambda^1(\mathcal{D}) $ is the  flat operation of $\bar{u}_t$ under the pullback metric $\Xi^*_t\mathbf{g}$. 
Applying Hamilton's principle $\delta S = 0$ with the standard Lin constrained variations for $\bar{u}$ and $\bar{D}$, as well as free variations for the pressure $\pi$, we obtain the averaged Euler--Poincar\'e equations following Proposition \ref{prop:det EPA} from the action \eqref{eq:averg euler action 1}
\begin{align}
    \left(\p_t + \mathcal{L}_{\bar{u}_t}\right)\bbE[ (\Xi^*_t\mathbf{g})(\bar{u}_t,\cdot) \otimes \Xi_{t*}\bar{D}_t] = \bar{D}_t\mathbf{d}\left(\frac{1}{2}\bbE[\Xi^*_t(\Xi^*_t\mathbf{g})(\bar{u}_t,\bar{u}_t)] - \bar{\pi}_t\right)\,, \label{eq:mean euler}
\end{align}
together with the deterministic advection of $\bar{D}$ by $\bar{u}$ and the volume-preserving condition that implies the incompressibility of $\bar{u}$
\begin{align}
    \left(\p_t + \mathcal{L}_{\bar{u}_t}\right) \bar{D}_t = 0\,,\quad \bar{D}_t = \mu_{\mathbf{g}} \quad \Longrightarrow \quad \operatorname{div}_{\mu_{\mathbf{g}}}\bar{u}_t = 0\,.
\end{align}
Under the constraint $\bar{D}_t = \mu_{\mathbf{g}}$, we have $\Xi_{t*}\bar{D} = \mcal{J}_{\Xi^{-1}_t} \mu_{\mathbf{g}}$, where $\mcal{J}_{\Xi^{-1}_t}$ is the determinant of the Jacobian of $\Xi^{-1}_t$. Let $\wt{\mathbf{g}} := \bbE\left[\mcal{J}_{\Xi^{-1}_t}\Xi^*_t\mathbf{g}\right]$ and $\wt{\flat}: \mathfrak{X}(\clD) \rightarrow \Lambda^1(\clD)$ the musical isomorphism associated with the weighted pullback Riemannian metric $\wt{\mathbf{g}}$.
Then, \eqref{eq:mean euler} can be expressed as
\begin{align}
    \left(\p_t + \mathcal{L}_{\bar{u}_t}\right)\bar{u}_t^{\wt{\flat}} =\mathbf{d}\left(\frac{1}{2}\bbE[\Xi^*_t(\Xi^*_t\mathbf{g})(\bar{u}_t,\bar{u}_t)] - \bar{\pi}_t\right)\,. \label{eq:averaged Euler simplified}
\end{align}
Let $\bar{q}=\mathbf{d}\bar{u}_t^{\wt{\flat}} : [0,T]\rightarrow \Lambda^2(\clD)$ be the vorticity. Then, $\bar{q}$ satisfies the following vorticity equation with an averaged Biot--Savart law:
\begin{align}
    \left(\p_t + \mathcal{L}_{\bar{u}_t}\right)\bar{q}_t = 0 \,,\quad \bar{u}_t =\wt{\sharp} \delta (-\triangle)^{-1} \bar{q}_t\,, \label{eq:avged vort eq}
\end{align}
where $\wt{\sharp}: \Lambda^1(\clD) \rightarrow \mathfrak{X}(\clD)$ is the dual operator to $\wt{\flat}$. In particular, $\bar{u}_t$ and $\bar{q}_t$ are related by a Biot--Savart law, where the metric is the expectation of the weighted pullback metric $\mcal{J}_{\Xi^{-1}_t}\Xi^*_t\mathbf{g}$.

\section{Alternative modelling perspectives} \label{prevwork}
The goal of this section is to compare the homogenisation derivation of the stochastic Lagrangian particle trajectory ansatz \eqref{eq:SALTansatz} presented in this work with the homogenisation derivation proposed in \cite{CGH2017}, hereafter referred to as CGH. The authors of CGH derived the ansatz \eqref{eq:SALTansatz} by applying a deterministic homogenisation procedure to a flow of diffeomorphisms conforming to a different fast-slow decomposition from that assumed in Section \ref{sec:homog}.

In what follows, we highlight the differences in assumptions and results between CGH and the current work. We then adapt the homogenisation arguments in Section \ref{sec:homog} to reinterpret the results of CGH and use the variational closure techniques in Section \ref{sec:vp} to obtain the corresponding stochastic Euler--Poincaré equations. Finally, we provide a physical interpretation of the modelling choices made in this work and in CGH.

\subsection{Homogenisation of mean flows}
In CGH, the full flow of diffeomorphisms is assumed to have the decomposition $g = \Xi^{\varepsilon}\circ \bar{g}^{\varepsilon}$ with $\Xi^{\varepsilon}_t = g'_{t/\varepsilon}$ for some chaotic flow $g' : [0,T] \rightarrow \operatorname{Diff}(\mcal{D})$ whose time dependence is on the quotient $t/\varepsilon$. The stochastic Lagrangian trajectory ansatz is then shown to emerge from the homogenisation limit of the map $\bar{g}^\varepsilon$. Compared with the present work (see Section \ref{sec:homog}), there are two key differences.

First, in CGH, the stochastic flow ansatz \eqref{eq:SALTansatz} arises from the limit of the map $\bar{g}^\varepsilon$, whereas in our approach, the same ansatz arises as the limit of the total map $g^{\varepsilon}$.
Second, in our approach, we assume a more general fast-slow flow decomposition of the form $g^{\varepsilon} = \Xi^{\varepsilon}\circ \bar{g}^{\varepsilon}$, cf. \eqref{eq:CoMdef}, whereas in CGH, the composite map $g$ is taken to be $\varepsilon$-independent.
Moreover, in this work, we have shown that both $g^\varepsilon$ and $\bar{g}^\varepsilon$ must be $\varepsilon$-dependent for consistency in the case of stochastic closures presented in Section \ref{sec:VP1}, which yield stochastic Euler--Poincaré equations.

When all flow maps are $\varepsilon$-dependent, one can switch between the two perspectives where either $g^{\varepsilon}$ or $\bar{g}^\varepsilon$ gives rise to the stochastic Lagrangian trajectory ansatz in the homogenisation limit via a relabelling of the flow of diffeomorphisms. The modelling interpretations of the two perspectives are discussed in Section \ref{sec:otherapproach}.

To pass to the homogenisation limit, the authors of CGH impose additional structural assumptions on the fluctuation map $g'_{t/\varepsilon}$.
In particular, the authors restrict to $\mathcal{D}=\mathbb{R}^d$ and assume $g'_{t/\varepsilon}$ factorises as
\begin{align*}
    g'_{t/\varepsilon}(X) = X + \zeta_{t/\varepsilon}(X)\,,
\end{align*}
where $\zeta: [0, T]\times \bbR^d\rightarrow \bbR^d$ is a map with the initial condition $\zeta_0(X)=0$ for all $X\in \bbR^d$. Here, we note that the decomposition of the fast map into this sum only makes sense in the case of $\mathcal{D}=\mathbb{R}^d$. For a discussion on this type of flow decomposition and its limitations in constructing Generalised Lagrangian Mean theories using ensemble averages, see, e.g., \cite{GV2018}. Continuing from the decomposition above, the composite flow map becomes
\begin{align*}
    g_t(X) = g'_{t/\varepsilon}\circ \bar{g}^{\varepsilon}_t (X) = \bar{g}^{\varepsilon}_t(X) + \zeta_{t/\varepsilon} \left(\bar{g}^{\varepsilon}_t(X)\right)\,.
\end{align*}
This factorisation is a special case of \eqref{eq:CoMdef} with $\Xi^{\varepsilon}= g'_{t/\varepsilon}$ without the additional assumptions \eqref{eq:Xi_eps} and \eqref{eq:lambda_eps}. In pursuit of obtaining a homogenised limit of $\bar{g}^{\varepsilon}$, the authors of CGH apply the time derivative and the chain rule to the decomposition of $g$, obtaining for every $X\in \mathbb{R}^d$:
\begin{align}\label{eq:deriv_g_holm}
   u_t(g_t(X))= \dot{g}_t(X)= Tg'_{t/\varepsilon}  \dot{\bar{g}}^{\varepsilon}_t(X) + \frac{1}{\varepsilon}\left[\frac{\partial  g'}{\partial t}\right]_{t/\varepsilon}\left( \bar{g}^{\varepsilon}_t(X)\right) \,,
\end{align}
where $Tg'_{t/\varepsilon}: \mathbb{R}^d\rightarrow \mathbb{R}^{d\times d}$ is the total derivative of the map $g'_{t/\varepsilon}$ in the spatial position. Assuming the matrix $Tg'_{t/\varepsilon}$ is invertible for all $t\in [0,T]$, rearranging \eqref{eq:deriv_g_holm} and inverting $T g'_{t/\varepsilon}$ yields (c.f., \cite[Eq. 3.9]{CGH2017})
\begin{equation*}
    \dot{\bar{g}}^{\varepsilon}_t(X) = [Tg'_{t/\varepsilon}]^{-1}u_t\left( g_t(X)\right) - \frac{1}{\varepsilon} [Tg'_{t/\varepsilon}]^{-1}    \left[\frac{\partial  g'}{\partial t}\right]_{t/\varepsilon} ( \bar{g}^{\varepsilon}_t(X)) \,.
\end{equation*}
The following assumption is then made:
\begin{align*}
    \left[\frac{\partial  g'}{\partial t}\right]_{t/\varepsilon} (\bar{g}^{\varepsilon}_t(X)) = \sum_{k=1}^K \lambda^k_{t/\varepsilon }\sigma_k ( \bar{g}^{\varepsilon}_t(X)), \quad \dot{\lambda}_t^{\varepsilon} = \varepsilon^{-2} h(\lambda^\varepsilon_t)\,, \quad \lambda_0=\omega \in \Omega\,, \quad \forall t\in [0,T]\,,
\end{align*}
where the dynamical system for $\lambda$ satisfies the conditions given in Assumption \ref{asm:weak_invariance} with an ergodic SRB measure $\bbP$ supported on $\Omega$ and $\sigma \in \mathfrak{X}_{C^{n+2}}(\mcal{D})^K$. Substituting the assumed decomposition yields
\begin{align*}
    \dot{\bar{g}}^{\varepsilon}_t(X) = [Tg'_{t/\varepsilon}]^{-1}u_t\left(g'_{t/\varepsilon}(\bar{g}^{\varepsilon}_t(X))\right) - \frac{1}{\varepsilon}  \sum_{k=1}^K [T g'_{t/\varepsilon}]^{-1} \sigma_k ( \bar{g}^{\varepsilon}_t(X))\lambda^k_{t/\varepsilon }.
\end{align*}
The centering condition in Section \ref{sec:homog}, Assumption \ref{asm:weak_invariance}, is replaced with
\begin{align*}
    \int_\Omega[T g'_{t/\varepsilon}]^{-1} \left[\frac{\partial  g'}{\partial t}\right]_{t/\varepsilon} \bbP ( \dd \omega) = 0 \, .
\end{align*}
With these assumptions, CGH claims one can pass to the $\varepsilon \rightarrow 0$ limit using deterministic homogenisation \cite{kelly2017deterministic} to have $\bar{g}^\varepsilon \rightarrow \bar{g}$, where $\bar{g}: \Omega \times [0,T] \rightarrow \operatorname{Diff}(\mcal{D})$ satisfies an SDE, c.f., \cite[Eq. 4.8]{CGH2017}.

However, we do not see how to do this because of the $\varepsilon$-dependent inverses that appear in both the drift and the noise. One may argue this proposal with further assumptions that there indeed exists such a form which correctly depends on the slow variables, or invoke more general theories of the type \cite[Theorem 5.5]{10.1214/21-AIHP1203}. Only in such a case can one then deduce weak convergence of the corresponding Lagrangian particles $\bar{X}^\varepsilon_t \rightarrow_\bbP \bar{X}_t$ as $\varepsilon \rightarrow 0$, where $\bar{X}^\varepsilon_t = \bar{g}^\varepsilon_t(X)$ and $\bar{X}_t = \bar{g}_t(X)$ for all $X \in \mathbb{R}^d$.

Even with convergence established via the above argument, the assumption that $g_t$ is independent of the parameter $\varepsilon$ is not natural if one has a stochastic variational closure in mind, as the map $g$ invariably couples to $\bar{g}^\varepsilon$ (or its limit) in the Euler--Poincaré equation as discussed in Section \ref{sec:VP1}. A fully explicit example of this occurring in incompressible fluid flow is shown in Section \ref{sec:examples} (equation \eqref{eq:ubar dynamics}).

Let us use the homogenisation analysis presented in Section \ref{sec:homog} with slight assumption modifications to obtain a rough limit for $\bar{g}^\varepsilon_t$. This approach reaches the original aim of CGH and bypasses any technicality of skew product forms, particular forms of the map $\Xi^\varepsilon$, and deduces that the composite flow map converges. We start with the decomposition $g^\varepsilon = \Xi^\varepsilon\circ \bar{g}^\varepsilon$ where $\varepsilon$-dependence of $g^\varepsilon$ is assumed for the consistency of stochastic variational closures and convergence properties. Let $\Xi^\varepsilon$ be the random flow of diffeomorphisms 
satisfying equation \eqref{eq:Xi_eps} and Assumption \ref{asm:weak_invariance}. By Theorem \ref{thm:fast_flow_covergence}, $\Xi^\varepsilon$ converges to $\Xi$ satisfying equation \eqref{eq:Xi_limit} as $\varepsilon \rightarrow 0$. The flow of inverse diffeomorphisms $\Theta^\varepsilon = \Xi^{\varepsilon; -1}$ satisfies the RPDE
\begin{align}
    \dd \Theta^{\varepsilon}_t({X})=-\sum\limits_{k=1}^K (T\Theta^{\varepsilon}_t)(X)\cdot \sigma_k(X)\rmd \bB^{\varepsilon;k}_t\,,  \quad \Theta_0^{\varepsilon}({X}) = {X}\,,
\end{align}
where $\sigma$ and $\mathbf{B}^{\varepsilon;k}_t$ are those defined in the flow of rough diffeomorphisms for $\Xi$ \eqref{eq:RDE Xi}. It can be shown that $\Theta^\varepsilon$ converges to $\Theta := \Xi^{-1} :\Omega \rightarrow C( [0,T] ;\operatorname{Diff}_{C^{n}}(\bbT^d))$ as $\varepsilon \rightarrow 0$, which satisfies the SPDE
\begin{align}
    \dd {\Theta}_t({X})= -\frac{1}{2}\sum \limits_{k,l=1}^K\Gamma^{kl} (T\Theta_t)(X)\cdot \qv{\xi_k}{\xi_l}(X) \dd t - \sum\limits_{k=1}^K (T\Theta_t)(X)\cdot \xi_k(X)\circ \dd W_t^k\,, \quad
\Theta_0(X) = X\,.
\end{align}

Instead of assuming $\bar{g}^\varepsilon \rightarrow_{\mathbb{P}} \bar{g}$ where $\dd \bar{g}(X) = \bar{u}(\bar{g}(X))\,\dt$ for some $\mathbb{F}$-adapted $\bar{u}:\Omega \rightarrow C([0,T]; \mathfrak{X}_{\text{Lip}^n}(\mcal{D}))$ in the homogenisation limit, c.f., Assumption \ref{asm:u_bar_eps}, we postulate that the full flow map $g^\varepsilon$ satisfies the equation
\begin{align*}
    \rmd g^\varepsilon_t (X) = u^\varepsilon(g^{\varepsilon}_t(X))\,\dt,\quad g^\varepsilon_0(X) = X\,,
\end{align*}
for some $u^\varepsilon :\Omega \rightarrow C([0,T]; \mathfrak{X}_{\text{Lip}^n}(\mcal{D}))$ that converges to
\begin{align}
    \rmd g_t (X) = u(g_t(X))\,\dt,\quad g_0(X) = X\,, \label{eq:random u}
\end{align}
for some $\mathbb{F}$-adapted $u :\Omega \rightarrow C([0,T]; \mathfrak{X}_{\text{Lip}^n}(\mcal{D}))$ in the $\varepsilon\rightarrow 0$ limit. Then, we can consider the homogenisation limit as $\varepsilon\rightarrow 0$ for the flow $\bar{g}^\varepsilon$ solving
\begin{align*}
\dot{\bar{g}}^{\varepsilon}_t(X)=\Theta^{\varepsilon}_{t*}u^{\varepsilon}_t(\bar{g}_t^{\varepsilon}(X)) - {\varepsilon^{-1}}\sum\limits_{k=1}^K \Theta^{\varepsilon }_{t*}\sigma_k(\bar{g}^{\varepsilon}_t(X))\lambda_t^{\varepsilon}\,,  \quad \bar{g}_0^{\varepsilon}(X)=X\in \clD\,.
\end{align*}
That is, we have the modelling choice to swap the roles of $u^\varepsilon$ and $\bar{u}^\varepsilon$ in the analysis of Section \ref{sec:homog} by suitably replacing $\Xi^\varepsilon$ with $\Xi^{\varepsilon;-1} = \Theta^\varepsilon$. Repeating the same arguments in Theorem \ref{thm:convergence}, we have $\bar{g}^\varepsilon \rightarrow_{\mathbb{P}} \bar{g}:\Omega \rightarrow C([0,T];\operatorname{Diff}_{C^{n}}(\mcal{D}))$ as $\varepsilon \rightarrow 0$ where the flow map $\bar{g}$ satisfies the following SDE:
\begin{align}\label{eq:g_homog_limit alt}
\rmd \bar{g}_t = \Theta_{t*}u_t(\bar{g}_t(X))\,\dt - \sum\limits_{k=1}^K\Theta_{t*}\xi_k(\bar{g}_t(X)) \circ \dW^k_t - \frac{1}{2}\sum\limits_{k,l=1}^K\Gamma^{kl}\Theta_{t*}\qv{\xi_k}{\xi_l}(\bar{g}_t(X)) \,\dt\,,
\end{align}
which can be expressed in the Eulerian form as
\begin{align}
    \rmd \bar{g}_t\bar{g}_t^{-1}= \Theta_{t*}u_t\,\dt - \sum\limits_{k=1}^K \wt{\xi}_{t,k}\circ \dW^k_t - \frac{1}{2}\sum\limits_{k,l=1}^K\Gamma^{kl}\qv{\wt{\xi}_{t,k}}{\wt{\xi}_{t,l}} \,\dt\,,\label{eq:stoch gbar}
\end{align}
where $\wt{\xi}_{t,k}:= \Theta_{t*}\xi_k$ for all $k$. We remark that in the homogenisation limits that implied the SDE dynamics for $g$ and $\bar{g}$, presented in equations \eqref{eq:g_homog_limit} and \eqref{eq:stoch gbar}, respectively, the same assumptions are placed on the fast flow map $\Xi^\varepsilon$. The assumptions that dictated whether $g$ or $\bar{g}$ satisfies an SDE are the assumptions that $\bar{g}$ or $g$ is differentiable in time, respectively.

To find the variational closure dynamics for the unknown vector field $u$, the most natural method is to insert $u$ into the Lagrangian of an action principle as $u = \dd g_t\,g_t^{-1}$ is the vector field associated with the full homogenised flow map $g$. As there are no stochastic constraints on $u$ or $g$, applying the Euler--Poincaré variational principle to a Lagrangian of $u$ will yield a deterministic equation for $u$. To couple to the flow map $\bar{g}$, we consider a modified vector field $\bar{u} := \Theta_* u = \Xi^* u$ and follow a similar procedure to that considered in Section \ref{sec:VP1}. Namely, to obtain the SALT Euler--Poincaré equation for $\bar{u}$, we consider the following. Let $V^*$ be the space of advected quantities where $a_0 \in V^*$ is some initial condition. Under the flow defined by equations \eqref{eq:random u} and \eqref{eq:stoch gbar}, we define $\bar{a}_t = a_0 \bar{g}_t^{-1}$ and $a_t = a_0 g_t^{-1}$ whose dynamics satisfy the following SPDE and random-coefficient PDE:
\begin{align}
\begin{split}
    &\dd \bar{a}_{t} + \mcal{L}_{\bar{u}_{t}} \bar{a}_{t}\,\dt - \sum\limits_{k=1}^K\mcal{L}_{\wt{\xi}_{t,k}} \bar{a}_{t}\circ \dW^k_t - \frac{1}{2}\sum\limits_{k,l=1}^K\mcal{L}_{\Gamma^{kl}\qv{\wt{\xi}_{t,k}}{\wt{\xi}_{t,l}}} \bar{a}_{t}\,\dt = 0\,,\\
    &\dd a_{t} + \mcal{L}_{u_{t}}a_{t}\,\dt = 0 \,,
\end{split} \label{eq:mean and full adv quantities alt g}
\end{align}
respectively. We consider a constrained Euler--Poincaré variational principle similar to Corollary \ref{cor:EP equivalent},
\begin{align*}
    0 = \delta S = \delta \int_0^T\ell(\bar{u}_t,\bar{a}_t)\,\dt\,,
\end{align*}
subject to constrained variations that are consistent with their definitions:
\begin{align*}
    \delta \bar{u}_t = \Ad_{\Theta_t}\left(\p_t v - \ad_{\Ad_{\Theta_t^{-1}}u_t}v\right)\,,\qquad \delta \bar{a}_t = -\left(\mcal{L}_v (\bar{a}_t\Theta^{-1}_t)\right)\Theta_t\,,
\end{align*}
where $v:\Omega \rightarrow C^1([0,T], \mathfrak{X}(\mcal{D}))$ is assumed to be an arbitrary variation. Then, the stationary conditions imply the following random-coefficient Euler--Poincaré equation:
\begin{align}
    \dd\left(\Theta^*_t\frac{\delta \ell}{\delta \bar{u}_t}\right) + \mathcal{L}_{u_t} \left(\Theta^*_t\frac{\delta \ell}{\delta \bar{u}_t}\right)\,\dt = \left(\Theta^*_t\frac{\delta \ell}{\delta \bar{a}_t}\right)\diamond a_t \,\dt\,,
\end{align}
which is equivalent to the stochastic equation
\begin{align}
    \dd \frac{\delta \ell}{\delta \bar{u}_{t}} + \ad^*_{\bar{u}_{t}} \frac{\delta \ell}{\delta \bar{u}_{t}}\,\dt - \sum\limits_{k=1}^K\ad^*_{\wt{\xi}_{t,k}} \frac{\delta \ell}{\delta \bar{u}_{t}} \circ \dW^k_t - \frac{1}{2}\sum\limits_{k,l=1}^K\ad^*_{\Gamma^{kl}\qv{\wt{\xi}_{t,k}}{\wt{\xi}_{t,l}}}\dede{\ell}{\bar{u}_{t}}\,\dt = \frac{\delta \ell}{\delta \bar{a}_{t}} \diamond \bar{a}_{t}\,\dt\,.
    \label{eq:EPA full 2017 ver}
\end{align}
We remark that the dynamics of $u_t$ in \eqref{eq:EPA full} and the dynamics of $\bar{u}_t$ in \eqref{eq:EPA full 2017 ver} are extremely similar.
In fact, by interchanging the prognostic vector fields $u_t$ and $\bar{u}_t$, the advected quantities $a_t$ and $\bar{a}_t$, and the stochastic vector fields $\xi_k$ and $-\wt{\xi}_{t,k}$ respectively, one sees that \eqref{eq:EPA full} and \eqref{eq:EPA full 2017 ver} are indeed the same.
However, the modelling assumptions that preceded these variational closures assign different interpretations to each equation.

\subsection{Modelling interpretations}\label{sec:otherapproach}
Let us denote observed advected quantities in the real world by $a^R$ and let $g^\varepsilon$ denote the full multi-scale Lagrangian flow map. Under these notations, it is most natural to model the observed advected quantities as $a_t^R=a_t^{\varepsilon}=g_{t*}^{\varepsilon} a_0$ for some $\varepsilon\ll 1$. In both CGH and the present work, the full Lagrangian flow is decomposed as $g^{\varepsilon}=\Xi^{\varepsilon}\circ \bar{g}^{\varepsilon}$\footnote{In fact, in CGH, the authors consider the case where $g$ is independent of $\varepsilon$.}, albeit with different assumptions placed on $\Xi^\varepsilon$ and $\bar{g}^\varepsilon$. Defining $\bar{a}_t=\bar{g}_{t*}^{\varepsilon}a_0$, we have $a^R_t=\Xi^{\varepsilon}_{t*}\bar{a}_{t}$.

In CGH, the authors heuristically show that $\bar{g}^{\varepsilon}$ converges to a stochastic flow map $\bar{g}$, which implies that $\bar{a}_t=\bar{g}_{t*}a_0$ satisfies an SPDE. The authors explain that this motivates the original ansatz of the SALT modelling approach \cite{Holm2015}. While the original paper \cite{Holm2015} and subsequent works, e.g., \cite{CCHOS18a}, do not denote the quantities satisfying the stochastic flow and SPDE with the over-bar $\bar{(\cdot)}$ notation, with CGH in mind, one is meant to interpret them as such; indeed, see \cite[Section 3]{CCHOS18a} which uses the word ``averaged'' to describe the stochastic Lagrangian particles in their modelling approach, which is presumably the limit of $\bar{g}^{\varepsilon}$ given in CGH. Under this model, one can recover the observed advected quantities $a_t^R$ via $a_t^R = \Xi_{t*} \bar{a}_t$; thus, the observed advected quantities are modelled by a random partial differential equation given in \eqref{eq:mean and full adv quantities alt g}.

In Section \ref{sec:homog}, we provide sufficient assumptions to prove that $g^{\varepsilon}$ converges to a stochastic flow map $g$. This implies that the observables $a_t^R=a_t=g_{t*}a_0$ satisfy an SPDE, and that $\bar{a}_t=\bar{g}_{t*}a_0$ satisfies a random PDE, both of which are given in equation \eqref{eq:mean and full adv quantities}.

The decision on how to interpret observed advected quantities is ultimately up to the preferences and viewpoint of those utilising the SALT model. Specifically, if the model dictates that the observed advected quantities should satisfy a random PDE, the model defined by equations \eqref{eq:EPA full 2017 ver} and \eqref{eq:mean and full adv quantities alt g} should be used. Alternatively, if the model dictates that the observed advected quantities should satisfy a stochastic PDE, the model defined by equations \eqref{eq:EPA full} and \eqref{eq:mean and full adv quantities} should be used.

Our contribution in the present work mathematically validates \emph{both} of these viewpoints through rigorous homogenisation arguments. Starting from the decomposition $g^\varepsilon = \Xi^\varepsilon \circ \bar{g}^\varepsilon$, either model can be justified from the homogenisation procedure described in Section \ref{sec:homog}. Furthermore, our approach is capable of deriving closures from a constrained Euler--Poincaré variational principle consistent with how one chooses to interpret the SALT modelling approach, by fixing an initial choice of ``rough'' limit map $g$ or $\bar{g}$.

\section{Concluding Remarks and Future Work}\label{sec:conclusion}

In this work, we (i) obtain a stochastic flow of diffeomorphisms as the homogenisation limit of a composition of fast-slow flows of diffeomorphisms, and (ii) use variational principles to derive stochastic closures for the slow dynamics of the homogenised flow while respecting its composition structure. Through this approach, we provide new perspectives on the derivation of the stochastic Euler--Poincaré equations for ideal fluid flows, as first proposed in \cite{Holm2015}.

In what follows, we summarise the main contributions of this work by section and discuss potential future research directions arising from our results.

In Section \ref{sec:homog}, we constructed a stochastic flow of diffeomorphisms as the flow map of a rough differential equation possessing a map composition structure using an iterated weak invariance principle (WIP) \cite{kelly2017deterministic, 10.1214/21-AIHP1203, chevyrev2019multiscale}. Our starting assumption was that $g^\varepsilon$ decomposes into slow and fast timescales: $g^\varepsilon = \Xi^\varepsilon \circ \bar{g}^\varepsilon$. We proved that the fast flow $\Xi^\varepsilon$ converges in law to a solution of an SDE via identification of a rough flow map and application of the WIP. Under the assumption that the slow flow is continuously related to a driving path of the fast map, we proved the same continuous dependence for the full map $g^\varepsilon$. This allowed us to use both the continuous mapping theorem and WIP to show that a limit $g$ exists. The Kunita-Itô-Wentzel formula \cite{de2020implications, len2023geometric} then yields an SDE representation for $g$, which we identify as satisfying the stochastic Lagrangian flow ansatz \eqref{eq:SALTansatz}.

In Section \ref{sec:vp}, we considered variational closures for the dynamics of the velocity fields $\bar{u}_t$ and $u_t := \Ad_{\Xi_t} \bar{u}_t$, where $u_t$ is associated with the drift coefficient of the stochastic flow $g = \Xi \circ \bar{g}$ resulting from the analysis of Section \ref{sec:homog}. We constructed a particular class of $\epsilon$-parameterised perturbations of $u_t$ in \eqref{eqn:perturbed dg} (distinct from the $\varepsilon$-dependence of the flow maps before homogenisation) by deforming only the $\bar{g}$ component of $g$ to $\bar{g}_\epsilon$, such that the deformed flow $g_\epsilon$ retains the form $g_\epsilon = \Xi \circ \bar{g}_\epsilon$. From these constructed perturbations, we derived the random Euler--Poincaré equation with advected quantities \eqref{eq:EPA mean 2} using a random, time-dependent Lagrangian. We show that this equation is equivalent to the stochastic Euler--Poincaré equation with advected quantities \eqref{eq:EPA full} that appeared in prior work, e.g., \cite{Holm2015}. Additionally, we considered the special case where the stochastic flow $\Xi_t$ is an isometry of the metric on the Riemannian manifold for $t\in [0,T]$. In this case, we showed that the Hamiltonian is preserved along the flow of the stochastic Euler--Poincaré equation \eqref{eq:EPA full}.

Motivated by averaging theory \cite{pavliotis2008multiscale} and generalised Lagrangian mean theory \cite{andrews_mcintyre_1978, GV2018, gilbert2024geometric}, we considered deterministic variational closure dynamics for $\bar{u}_t$ in Section \ref{sec:det vp}. In this setting, we replaced Assumption \ref{asm:u_bar_eps} with Assumption \ref{weakgbarflow} on $\bar{g}$ and established the convergence of the map $g^\varepsilon$. Then, through a Lagrangian functional that takes the expectation of semimartingale arguments, application of Hamilton's principle resulted in the deterministic equation \eqref{eq:ELASALT EP 1} for the mean momentum.

To illustrate our modelling approaches, we provided applications to the incompressible Euler equations in Section \ref{sec:examples}. In Section \ref{sec:Euler}, we explicitly derived the equivalent random-coefficient and stochastic forms of the incompressible Euler equations through an Euler--Poincaré variational principle and examined the associated Kelvin circulation and vorticity dynamics. We showed that when $\Xi_t$ is an isometry for $t\in [0,T]$, the random-coefficient form of the two-dimensional Euler equation is greatly simplified. In Section \ref{subsec:avgeuler}, we obtained a variant of the incompressible Euler equation as an averaged Euler--Poincaré equation, where the fluid velocity $\bar{u}_t$ satisfies a deterministic evolution following the averaging closure discussed in Section \ref{sec:det vp}.

In Section \ref{prevwork}, we compared our homogenisation analysis and variational closures governing the dynamics of $\bar{u}$ and $u$ with previous works, namely \cite{CGH2017}. We highlighted the differences in modelling assumptions between the present work and \cite{CGH2017}, expressed their homogenisation analysis in the framework of Section \ref{sec:homog}, and formulated stochastic variational principles to derive stochastic equations of motion based on their modelling assumptions, using constructions similar to those in Section \ref{sec:vp}. The impacts and interpretations of the modelling assumptions were discussed in Section \ref{sec:otherapproach}.

\paragraph{Open problems and future work.} Following this work, several open problems should be addressed in future research. The analysis using rough path and Lie-theoretic modelling approaches taken in this paper allows us to investigate a more general class of equations than those presented. We outline such a programme below.

\begin{itemize}
    \item Construct $\Xi$ as a rough flow of group elements solving a rough differential equation (RDE) of the type $\dd \Xi_t = \xi (\Xi_t ) \dd \bZ_t $ for a rough path $\bZ_t = (Z_t, \bbZ_t)$. This can be done with simple modifications of Assumptions \ref{asm:weak_invariance} and \ref{asm:u_bar_eps} to use a generic path $\bZ$. This way, we are not restricted to the Brownian motion case, which was originally constructed as the limit of multi-time dynamics for fast eigenvalues $\lambda^{ \varepsilon}$ by homogenisation theory. It suffices to ask whether the path $\bZ_t = (Z_t, \bbZ_t)$ satisfies such a WIP. Some iterated weak invariance principles have been proven by Gehringer and Li for fractional Brownian motion (see \cite{gehringer2019homogenization}).
    
    \item Generalise the map composition structure to an arbitrary Lie group $G$ and study the convergence of $\Xi$ and its composition to a rough flow in $G$ in the Lie group topology. Formulate the associated variational principles on an arbitrary Lie group $G$ in combination with $\Xi$ being the flow of an RDE defined on the Lie algebra of $G$. Following the same variational approach as in Section \ref{sec:vp}, we believe the rough differential equations derived in \cite{crisan2022variational} will naturally follow from this approach.
\end{itemize}

Studies of the analytical properties of stochastic Euler--Poincaré equations under transport noise are a promising direction due to their connection to random-coefficient Euler--Poincaré equations via pullback. By adopting stochastic flow arguments similar to the Yudovich theorems \cite{Yudovich1963, Yudovich1995, galeati2025wellposednessnonlinearroughcontinuity} for two-dimensional Euler equations to a variety of stochastic Euler--Poincaré equations, novel well-posedness results can be obtained.

\subsection*{Acknowledgments}
We are especially grateful to G. Gottwald and G. Pavliotis for their stimulating discussions on homogenisation theory. We also wish to thank C. Cotter, D. Holm, and O. Street for several thoughtful suggestions during the course of this work, which have improved or clarified the interpretation of its results. The work of TD is supported by an EPSRC PhD Scholarship. RH is grateful for support from the Office of Naval Research (ONR) grant award N00014-22-1-2082, Stochastic Parameterization of Ocean Turbulence for Observational Networks. JML acknowledges support received from the US AFOSR Grant FA8655-21-1-7034.

\appendix
\section{Appendix}\label{app}
\subsection{Rough paths and function spaces}
\label{app:1}
Given a path $Z:[0,T]\rightarrow \bbR^d$, we define its increment $\delta Z:[0,T]^2\rightarrow \bbR^d$ by $\delta Z_{st}=Z_t-Z_s$. Moreover, given a two-index map $\bbZ:[0,T]^2\rightarrow \bbR^d$, we define its increment $\delta Z:[0,T]^3\rightarrow \bbR^d$ by $$\delta \bbZ_{sut}=\bbZ_{st}-\bbZ_{su}-\bbZ_{ut}\,.$$

\begin{definition} \label{defRP}
Let $K\in \bbN$, $\alpha \in (\frac13 ,\frac12 ]$, and $T\in \bbR_+$. A pair $\bZ=(Z,\mathbb{Z})$ such that $Z:[0,T]\rightarrow \mathbb{R}^K$ and $\mathbb{Z}:[0,T]^2\rightarrow \mathbb{R}^{K\times K}$ is called an $\alpha-$H\"older continuous rough path on the interval $[0,T]$ if 
\begin{equation*}
[Z]_{\alpha} := \sup_{s\ne t\in [0,T]} \frac{\left|\delta Z_{st}\right|}{|t-s|^\alpha}<\infty\,, \quad [\bbZ]_{2\alpha} := \sup_{s\ne t\in [0,T]} \frac{\left|\bbZ_{st}\right|}{|t-s|^{2\alpha}}<\infty \,
\end{equation*}
and Chen's relation holds:
$$
\delta \bbZ_{sut} =\delta Z_{su}\otimes \delta Z_{ut} \quad \forall (s,t)\in [0,T]^2\,.
$$
We denote by $\clC^{\alpha}([0,T]; \mathbb{R}^K)$ the complete metric space of $\alpha-$H\"older continuous rough paths on the interval $[0,T]$, endowed with the metric
$$
d(\mb{Z}^1, \mb{Z}^2):=[Z^1-Z^2]_{\alpha} + [\bbZ^1-\bbZ^2]_{2\alpha}.
$$
We also define the the following norm-like function on $\clC^{\alpha}([0,T]; \mathbb{R}^K)$:  $$\| \mb{Z} \|_{\clC_g^{\alpha}([0,T]; \mathbb{R}^K)} := [Z]_\alpha + \sqrt{[\bbZ]_{2\alpha}}.$$
  Moreover, we let $\clC_g^{\alpha}([0,T]; \mathbb{R}^K)$ denote the closure of 
\[ \left\{(Z, \mathbb{Z}) \in C^1([0,T];\bbR^K) \oplus C^1([0,T]^2;\bbR^{K \times K}) : \mathbb{Z}_{st}=\int_s^t \delta Z_{su} \otimes \dd {Z}_u \right\}  \]
in $\clC^{\alpha}([0,T]; \mathbb{R}^K)$, which is separable \cite[Exercise 2.8]{friz2020course}.
\end{definition}

Given $d\in \bbN$, let $\bbT^d$ denote the flat rational torus. Given $n\in \bbN$, let $\fkX_{C^n}(\bbT^d)$ denote the Banach space of $n$-times continuously differentiable vector fields on $\bbT^d$. 

\begin{definition}
Let $\alpha \in (\frac13, \frac12]$,   $\bZ=(Z,\bbZ)\in \clC^{\alpha}([0,T]; \mathbb{R}^K)$, $d\in \bbN$, $X\in \bbT^d$, $b\in \fkX_{C^0}(\bbT^d)$, and $\xi \in \fkX_{C^2}(\bbT^d)^K$. We say that a path $Y: [0,T]\rightarrow \bbT^d$ is a solution of the rough differential equation
\begin{align}\label{eq:rde}
dX_t = b(X_t)\rmd t + \sum_{k=1}^K\xi_k(X_t) \rmd \bZ_t^k, \quad t\in (0,T]\,, \quad X_0=X\in \bbT^d,
\end{align}
if $Y_0=X$ and $R:[0,T]^2\rightarrow \bbT^d$ defined by
$$
R_{st}=\delta X_{st}- \int_0^t b(X_s)\rmd s -  \sum_{k=1}^K\xi_k(X_t) \delta Z_{st}^k - \sum_{k,l=1}^K \xi_k[\xi_l](X_s)\bbZ^{lk}_{st} 
$$
satisfies $$[R]_{3\alpha}=\sup_{s\ne t\in [0,T]} \frac{|R_{st}|}{|t-s|^{3\alpha}}<\infty\,.$$
\end{definition}

\begin{definition}[$C^n$ Diffeomorphisms] Given $n \in \bbN$ and a smooth compact boundaryless manifold $\mc{D}$, we define $\operatorname{Diff}_{C^n}(\mcal{D})$ to be space of $C^n$-diffeomorphisms. We endow $\operatorname{Diff}_{C^n}(\mcal{D})$ with the Whitney topology \cite{hirsch1976differential}; that is, $g^n \rightarrow g$ in $\operatorname{Diff}_{C^n} (\mcal{D})$ if and only if for any pair of charts $(\varphi, U), (\psi, V)$ on $\mcal{D}$ such that $g^n(U)$, $ g(U)$ are subsets of $V$, the maps $\wt{g}^{\, n} := \psi \circ g^n \circ \varphi^{-1}, \wt{g} :=  \psi \circ g\circ \varphi^{-1}  : \varphi(U) \subset \bbR^d \rightarrow \psi(V) \subset \bbR^d$ are such that for all multi-indices $|\beta|\le n$, 
\begin{equation*}
\begin{aligned}
& \sup_{X\in \bbR^d}\left \vert \partial^{\beta}_X\wt{g}^{\,n} (X)-\partial^{\beta}_ X\wt{g}(X) \right\vert \rightarrow 0 \quad  \textnormal{and} \quad  \sup_{X\in \bbR^d}\left \vert \partial^{\beta}_X\wt{g}^{\, n;-1} (X)-\partial^{\beta}_ X\wt{g}^{-1}(X) \right\vert \rightarrow 0\, .
\end{aligned}
\end{equation*}
Under this topology, $\operatorname{Diff}_{C^n}(\mcal{D})$ becomes a topological group with the composition operation \cite[Thm. 2.3.2]{marsden_ebin_fischer_1972}. This topological group is complete and separable, may be endowed with the structure of a Polish group \cite{BAXENDALE1984}.
\label{def:diffeo}
\end{definition}

\begin{definition} 
Let $T>0$, $\alpha\in (\frac13,\frac12]$, and $n\in \bbN$. We denote by $C^{\alpha}([0,T]; \operatorname{Diff}_{C^n}(\bbT^d))$ the set of time-dependent diffeomorphisms $g: [0,T]\rightarrow  \operatorname{Diff}_{C^n}(\bbT^d)$ satisfying 
    $$
     \sup_{X\in \clD, \beta: |\beta|\le n} [\partial^{\beta}g_{\cdot}
     (X)]_{\alpha} <\infty  \quad \textnormal{and}\quad  \sup_{X\in \clD, \beta: |\beta|\le n} [\partial^{\beta}g^{-1}_{\cdot}(X)]_{\alpha}<\infty \,.
    $$
The space  $C^{\alpha}([0,T]; \operatorname{Diff}_{C^n}(\bbT^d))$ is a polish space when endowed with the metric  \cite[Ex.\ 11.17]{friz2010multidimensional}
\[d(g,h) := \sup_{\beta : \, \, |\beta| \leq n, \, \, X \in \bbT^d } |\partial^{\beta}g(X) - \partial^{\beta}h(X)|_{\alpha}   + \sup_{\beta : \, \, |\beta| \leq n, \, \, X \in \bbT^d } |\partial^{\beta}g^{-1}(X) - \partial^{\beta}h^{-1}(X)|_{\alpha}   \, .\]
\end{definition}

\begin{definition} Let $(E,d_E)$ and $(F,d_F)$ be two metric spaces. We say $f:E \rightarrow F$ is locally Lipschitz if for all $x \in E$, there exists a $\delta>0$ and $K>0$ such that for all $y,z\in B_{\delta}(x)$,
\[ d_F(f(y), f(z)) \le K d_E (y,z) \, .\]

We denote the space of such maps by $\operatorname{Lip}_{loc}(E, F)$.
\end{definition}

Given $\gamma>0$, we write 
$
\gamma = \lfloor \gamma \rfloor + \{\gamma\},
$
where $\lfloor \gamma \rfloor \in \mathbb{N}$ and $\{\gamma\}\in (0,1]$. Let $\mathfrak{X}_{\operatorname{Lip}^{\gamma}}(
\bbT^d)$ denote the space of $\lfloor \gamma \rfloor$-times differentiable vector-fields $\xi : \bbT^d\rightarrow \bbR^d$ such that $\partial^{\beta}\xi$ is $\{\gamma\}$-H\"older for all multi-indices $|\beta| = \lfloor \gamma \rfloor$.

\begin{theorem}[Theorem 8.15 in \cite{friz2020course}, Proposition 11.1, Theorem 11.(2-3) in \cite{friz2010multidimensional}]\label{thm:rough_flow_map}
Let $K\in \bbN$, $\alpha \in (\frac13, \frac12]$, $T\in \bbR_+$, $n\in \bbN$. Assume that $b\in \fkX_{\operatorname{Lip}^{n+\delta}}(\bbT^d)$ for $\delta>0$ and $\xi \in \mathfrak{X}_{\operatorname{Lip}^{n+\gamma - 1}}(
\bbT^d)^K$ for $\gamma> \alpha^{-1}$. Then for all initial conditions $X\in \bbT^d$ and rough paths $\bZ\in \clC_g^{\alpha}([0,T]; \mathbb{R}^K)$, there exists a unique solution of \eqref{eq:rde}. Furthermore, there exists a one-parameter flow map $$\Phi \in \operatorname{Lip}_{loc}\left(\clC_g^{\alpha}([0,T]; \mathbb{R}^K);C^{\alpha}( [0,T] ;\operatorname{Diff}_{C^n}(\bbT^d))\right)$$ satisfying $X_t=\Phi_t(X,\bZ)$ for all $t\in [0,T]$ and  $X\in \bbT^d$.
\end{theorem} 

\begin{theorem}[Theorem 9.1 in \cite{friz2020course}] \label{thm:RDE2SDE}
Let the assumptions of Theorem \ref{thm:rough_flow_map} hold and $(\Omega,\mathcal{F}, \bbF=\{\clF_t\}_{t\le T},\bbP)$ denote a complete filtered probability space supporting an $\bbR^K$-dimensional  Wiener processes $W=(W^1,\ldots, W^K)$ with independent components. Let $\bW=(W,\bbW)\in \clC_g^{\alpha}([0,T]; \mathbb{R}^K)$ denote the Stratonovich lift of $W$. Then  the unique solution of the rough differential equation
$$
dX_t = b(X_t)\rmd t + \sum_{k=1}^K\xi_k(X_t) \rmd \bW_t^k, \quad t\in (0,T]\,, \quad X_0=X\in \bbT^d,
$$
is the strong solution of the Stratonovich stochastic differential equation 
\begin{align}\label{eq:sde}
dX_t = b(X_t)\rmd t + \sum_{k=1}^K\xi_k(X_t) \circ \rmd W_t^k, \quad t\in (0,T]\,, \quad X_0=X\in \bbT^d\,.
\end{align}
Moreover, the stochastic flow maps \cite{kunita1990stochastic} induced by \eqref{eq:sde} and denoted $\phi: \Omega \times [0,T]\rightarrow \operatorname{Diff}_{C^n}(\bbT^d)$ satisfies $\bbP$-a.s., $\phi=\Phi(\cdot, \bW)$.
\end{theorem}

The following lemma is a straightforward, but cumbersome, application of the mean-value theorem and Fa\'a di Bruno's formula. We state it without proof.

\begin{lemma}\label{lem:continuity_comp_diffo}
    Let $T>0$ and  $n, m \in \bbN$ be given with $m \geq 2$. Then, the composition map
    \begin{equation*}
    \begin{aligned}
    \fkC : C([0,T]; \operatorname{Diff}_{C^{n+m}}(\bbT^d)) \times C([0,T]; \operatorname{Diff}_{C^{n}}(\bbT^d)) &\rightarrow C([0,T]; \operatorname{Diff}_{C^{n}}(\bbT^d))\\
    (g, h) &\mapsto g \circ h
    \end{aligned}
    \end{equation*}
    is locally Lipschitz.
\end{lemma} 

\subsection{Rigid body rotations and Kubo oscillator} \label{rigidbodyexample}
In this subsection, we give a finite-dimensional example of the stochastic and random-coefficient Euler--Poincar\'e equations, \eqref{eq:EPA full} and \eqref{eq:EPA mean}, respectively, for a left invariant Lagrangian. The example we consider is rigid body rotations expressed in the body frame that may be seen as a gyroscopic analogue of Euler's fluid equation on the special orthogonal group $SO(3)$. 

We may embed $SO(3) \hookrightarrow \operatorname{Diff}(\bbR^3)$ by identifying $O \in SO(3)$ with the linear diffeomorphism $x \mapsto Ox$. We state without proof that Proposition \ref{prop:random to stoch EPA} holds under the restriction map and the variational structure developed in this paper also applies to the matrix Lie group $SO(3)$. In fact, one may restrict the variations used in equations \eqref{eqn:d e def} and \eqref{eqn:perturbed dg} to $SO(3)$ and show that the perturbation diffeomorphisms also embedded to $SO(3)$.  

The left invariance requires one to consider the left translated vector field $\Xi^{-1}_t \dd \Xi_t$. However, we may reduce this to right-invariant case by considering the inverse of all maps considered in Proposition \ref{prop:random to stoch EPA}, which swaps the order in the composition of maps. 

On the Lie algebra level, $\mathfrak{s}\mathfrak{o}(3)$ is intrinsically defined as the space of skew symmetric $3 \times 3$ matrices. These may be represented by vectors in $\bbR^3$ through the hat map isomorphism $\widehat{(\cdot)} : \bbR^3 \rightarrow \mathfrak{s}\mathfrak{o}(3)$, $\bbR^3 \ni \Omega_i \mapsto \varepsilon^i_{jk} \Omega_j =: \widehat{\Omega}$ and may be embedded into $\mathfrak{X}(\bbR^3)$ as linear vector fields $x \mapsto \widehat{\Omega} x$. 

\begin{lemma}
The left adjoint actions of $SO(3)$ and the algebra $\mathfrak{s}\mathfrak{o}(3)$ may be obtained from restriction of the diffeomorphism group, let $\Xi_t \in SO(3)$ and $\widehat{\Omega}_t, \widehat{\Omega}^\prime_t \in \mathfrak{s}\mathfrak{o}(3)$,
\begin{align*}
\begin{split}
\Ad_{\Xi_t \cdot x} {\widehat{\Omega}}_t \cdot x &:= (\Xi_t \cdot x)_* {\widehat{\Omega}}_t \cdot x := T(\Xi_t \cdot x) {\widehat{\Omega}}_t \Xi^{-1}_t\cdot  x = \Xi_t {\widehat{\Omega}}_t \Xi^{-1}_t\cdot  x\,, \\
\ad_{\widehat{\Omega}_t \cdot x} \widehat{\Omega}^\prime_t \cdot x &:= \left[ \widehat{\Omega}_t \cdot x, \widehat{\Omega}^\prime_t \cdot x\right]_{\mathfrak{X}(\bbR^3)} := \widehat{\Omega}_t   \widehat{\Omega}^\prime_t \cdot x  - \widehat{\Omega}^\prime_t   \widehat{\Omega}_t \cdot x  \,.
\end{split}
\end{align*}
\end{lemma}
By suppressing the spacial coordinate $x$, we recover the $SO(3)$ representatives. Using the convention $\Ad^*_{\Xi_t} = \Xi_t^* = \Xi^{-1}_{t *} $, we deduce the dual representations of $\mathfrak{so}(3)$ in the $\bbR^3$ representation as   
\begin{align*}
    \Ad_{\Xi_t } {\Omega}_t = \Xi_t \Omega_t, \quad \Ad^*_{\Xi_t} \Pi_t := \Xi_t^{-1} \Pi_t, \quad \ad^*_{\Omega_t }\Pi_t := \Pi_t \times \Omega_t, \quad \Xi_t \in SO(3), \,\,  \Omega_t, \Pi_t \in \bbR^3
\end{align*}

The rigid body Lagrangian $\ell_{\operatorname{RB}} : \mathfrak{s}\mathfrak{o}(3) \simeq \bbR^3 \rightarrow \bbR$ takes the form $\ell_{\operatorname{RB}}(\Omega) = \frac{1}{2}\Omega \cdot \bbI \Omega$ for an inertia matrix $\bbI : \mathfrak{so}(3) \rightarrow \mathfrak{so}(3)^*$. Applying the composition of maps assumption, 
\begin{align*}
    O_t = \bar{O}_t \Xi_t \in SO(3), \quad \xi \circ \dd W_t = \Xi_t^{-1} \dd \Xi_t\,,
\end{align*}
one can rewrite the Lagrangian $\ell_{\operatorname{RB}}(\Omega) $ in terms of these two maps, 
\begin{align}
    \ell_{\operatorname{RB}}(\Omega) = \ell_{t,\operatorname{RB}}^{\Xi} (\bar{\Omega}_t) := \ell_{\operatorname{RB}}( \Ad_{\Xi_t} \overline{\Omega}_t) = \frac{1}{2} (\Xi_t {\overline{\Omega}}_t ) \cdot \bbI (\Xi_t {\overline{\Omega}}_t )\,.
\end{align}
The stochastic rigid body equations following from the Lagrangian $\ell_{RB}$ can be written as
\begin{align}
    \dd \Pi_t = \Pi_t \times (\Omega_t \dd t + \Xi^{-1}_t \dd \Xi_t ), \quad \Pi_t = \bbI \Omega_t, \quad O^{-1}_t\dd{O}_t = \Omega_t, \quad \Xi^{-1}_t\dd \Xi_t = \Gamma_{ij} \xi_i \times \xi_j \dd t + \xi \circ \dd W_t \,.  \label{stochrb}
\end{align}
One may transform this to a random-coefficient ODE by applying the operator $\Ad^*_{\Xi_t}$,
\begin{align*}
\begin{split}
\Ad_{\Xi_t}^* \dd \Pi_t &= \Xi_t^{-1} \dd \Pi_t = \dd (\Xi_t^{-1} \Pi_t) - \dd( \Ad^*_{\Xi_t} ) \Pi_t = \dd (\Xi_t^{-1} \bbI {\Xi_t \bar{\Omega}_t }) + \Xi_t^{-1}\ad^*_{\Xi_t^{-1} \dd \Xi_t } \Pi_t \\
&= \dd (\Xi_t^{-1} \bbI {\Xi_t \bar{\Omega}_t }) + \Xi_t^{-1} \bbI {\Xi_t \bar{\Omega}_t } \times \Xi_t^{-1} \Xi^{-1}_t\dd \Xi_t \, , \\
\\
\Ad_{\Xi_t}^* (\Pi_t \times \Omega_t )\dd t &= {\Xi^{-1}_t} (\Pi_t \times \Omega_t)\dd t = ({\Xi^{-1}_t}\bbI \Xi_t \bar{\Omega}_t \times \bar{\Omega}_t )\dd t  \, ,\\
\\
\Ad_{\Xi_t}^* (\Pi_t \times \Xi^{-1}_t \dd \Xi_t ) &= {\Xi^{-1}_t} (\Pi_t \times \Xi^{-1}_t \dd \Xi_t ) = {\Xi^{-1}_t}\bbI \Xi_t \bar{\Omega}_t \times  \Xi_t^{-1} \Xi^{-1}_t\dd \Xi_t \, .
\end{split}
\end{align*}
We see that the ${\Xi^{-1}_t}\bbI \Xi_t \bar{\Omega}_t \times  \Xi_t^{-1} \Xi^{-1}_t\dd \Xi_t $ term appear on both sides of the equation and can be cancelled, thus leading to the equation, 
\begin{align}
    \dd (\Xi_t^{-1} \bbI {\Xi_t \bar{\Omega}_t }) = ({\Xi^{-1}_t}\bbI \Xi_t \bar{\Omega}_t \times \bar{\Omega}_t )\dd t  \,.\label{stochrb2}
\end{align}
From equation \eqref{stochrb} we have the conservation of the total angular momentum $\frac12 \| \Pi_t \|^2 = \frac12 \|\bbI \Xi_t \bar{\Omega}_t \|^2$ as a Casimir invariant. Since elements of $SO(3)$ are norm preserving this is equal to $\frac12 \| \Xi_t^{-1} \bbI \Xi_t \bar{\Omega}_t \|^2 = \frac12 \| \Ad^*_{\Xi_t} \Pi_t\|^2$, the conserved momentum corresponding to equation \eqref{stochrb2}.

It is known (see for example, \cite{Arnaudon2018}) that when\footnote{Note that such a choice of linearly dependent $\widehat{\xi}_k$ possess vanishing commutators, regardless of $\Gamma$ the noise is of Stratonovich type.} $\widehat{\xi}_k \equiv \sigma \widehat{e}_3 \in \mathfrak{s}\mathfrak{o}(3)$, $\sigma \in \bbR$ and $\bbI = \operatorname{diag}(I_1, I_1, I_3)$, the resulting equations reduce to the energy conserving Kubo oscillator. In the composition of maps language, this is due to the fact that $\widehat{\xi} \circ \dd W_t$, with this choice of $\xi$, defines a stochastic curve $\Xi_t \in SO(3)$ that commutes with $\bbI$,
\begin{align}
     \Xi_t = \left(\begin{array}{ccc}
 \sqrt{\sigma }\cos(W_t) & - \sqrt{\sigma }\sin(W_t) & 0 \\
 \sqrt{\sigma }\sin(W_t) &  \sqrt{\sigma }\cos(W_t) & 0 \\
0 & 0 &  {\sigma^{-1} }
\end{array}\right) \in SO(3), \quad \Xi_t^{-1} \dd \Xi_t  = \left(\begin{array}{ccc}
0 & -\sigma  & 0 \\
\sigma  & 0 & 0 \\
0 & 0 & 0
\end{array}\right) \circ \dd W_t \in \mathfrak{s}\mathfrak{o}(3) \, .\label{so3isom}
\end{align}
The commutation of these specific choices of $\Xi_t$ and $\bbI$ imply the following $\Xi_t$ invariance of $\ell_{\operatorname{RB}}$ and so energy preservation from Proposition \ref{energyconservethm} applies,
\begin{align*}
\begin{split}\ell_{\operatorname{RB}}( \Ad_{\Xi_t} \overline{\Omega}_t) &= \frac{1}{2} (\Xi_t {\overline{\Omega}}_t ) \cdot \bbI (\Xi_t {\overline{\Omega}}_t ) := \frac12 \left \langle \Xi_t \bar{\Omega}_t , \bbI \Xi_t \bar{\Omega}_t  \right \rangle \stackrel{\eqref{so3isom}}{=} \frac12 \left \langle \Xi_t \bar{\Omega}_t , \Xi_t \bbI \bar{\Omega}_t  \right \rangle  \\&= \frac12 \left \langle \Xi^T_t\Xi_t \bar{\Omega}_t ,  \bbI \bar{\Omega}_t  \right \rangle = \frac12 \left \langle \bar{\Omega}_t ,  \bbI \bar{\Omega}_t  \right \rangle = \ell_{\operatorname{RB}}(\bar{\Omega}_t) \, .
\end{split}
\end{align*}
In fact one can identify this choice of $\xi$ as a Killing field for a metric on $\bbR^3$ with coefficients $\mathbf{g}_{ij} = \bbI_{ij}$, generating an isometry $SO(\bbI)$ (orthogonal with respect to the inertia tensor, $\Xi_t \bbI \Xi_t^T = \bbI $).

The term $\Xi_t^{-1} \bbI \Xi_t$ in equation \eqref{stochrb2} is analogous to $\Xi_t^* \mathbf{g}$ seen in Section \ref{sec:examples} and may be averaged in the same manner as Section \ref{subsec:avgeuler}. This leads to the deterministic closure model for $\bar{\Omega}$ given by
\begin{align}
    \frac{\dd}{\dd t} \left(  \bbE \left[\Xi_t^{-1} \bbI \Xi_t \right] \bar{\Omega}_t \right) = \left( \bbE \left[\Xi_t^{-1} \bbI \Xi_t \right] \bar{\Omega}_t \right) \times \bar{\Omega}_t \,. \label{stochrb3}
\end{align}
This equation conserves the Casimir $\frac12 \| \bbE \left[\Xi_t^{-1} \bbI \Xi_t \right] \bar{\Omega}_t\|^2$. Using particular choices of isometry and inertia matrix $\bbI$ constructed in \eqref{so3isom} we observe that $\Xi_t^{-1} \bbI \Xi_t = \bbI $ and it follows that both \eqref{stochrb2}, \eqref{stochrb3} reduce to the rigid body ODE in $\bar{\Omega}_t$ variables,
\begin{align*}
    \bbI \dot {\bar{\Omega}}_t + \bbI \bar{\Omega}_t \times \bar{\Omega}_t = 0\,.
\end{align*}
This is a classical equation (a variant of the Euler top) that conserves energy.

\bibliographystyle{plain}
\bibliography{main.bib}

@article{CGH2017,
author = {Cotter, C. J. and Gottwald, G. A.  and Holm, D. D. },
title = {{Stochastic partial differential fluid equations as a diffusive limit of deterministic Lagrangian multi-time dynamics}},
journal = {Proceedings of the Royal Society A: Mathematical, Physical and Engineering Sciences},
volume = {473},
number = {2205},
pages = {20170388},
year = {2017},
doi = {10.1098/rspa.2017.0388},

URL = {https://royalsocietypublishing.org/doi/abs/10.1098/rspa.2017.0388},
eprint = {https://royalsocietypublishing.org/doi/pdf/10.1098/rspa.2017.0388},
}

@article{flandoli2022additive,
  title={From additive to transport noise in 2d fluid dynamics},
  author={Flandoli, F. and Pappalettera, U.},
  journal={Stochastics and Partial Differential Equations: Analysis and Computations},
  volume={10},
  number={3},
  pages={964--1004},
  year={2022},
  publisher={Springer}
}

@article{debussche2024second,
  title={Second order perturbation theory of two-scale systems in fluid dynamics},
  author={Debussche, A. and Pappalettera, U.},
  journal={Journal of the European Mathematical Society},
  year={2024}
}

@article{hofmanova2019navier,
  title={On the Navier--Stokes equation perturbed by rough transport noise},
  author={Hofmanov{\'a}, M. and Leahy, J-M. and Nilssen, T.},
  journal={Journal of Evolution Equations},
  volume={19},
  number={1},
  pages={203--247},
  year={2019},
  publisher={Springer}
}

@book{kunita1990stochastic,
  title={{Stochastic Flows and Stochastic Differential Equations}},
  author={Kunita, H.},
  isbn={9780521599252},
  lccn={89070813},
  series={Cambridge Studies in Advanced Mathematics},
  year={1990},
  publisher={Cambridge University Press}
}

@article{Holm2015,
  title={{Variational principles for stochastic fluid dynamics}},
  author={D. D. Holm},
  journal={Proceedings of the Royal Society A: Mathematical, Physical and Engineering Sciences},
  volume={471},
  number={2176},
  pages={20140963},
  year={2015},
  publisher={The Royal Society Publishing}
}

@article{de2020implications,
  title={{Implications of Kunita--It{\^o}--Wentzell formula for k-forms in stochastic fluid dynamics}},
  author={de Leon, A. B. and Holm, D. D. and Luesink, E. and Takao, S.},
  journal={Journal of Nonlinear Science},
  volume={30},
  pages={1421--1454},
  year={2020},
  publisher={Springer}
}

@article{crisan2022variational,
  title={{Variational principles for fluid dynamics on rough paths}},
  author={Crisan, D. and Holm, D. D. and Leahy, J-M. and Nilssen, T.},
  journal={Advances in Mathematics},
  volume={404},
  pages={108409},
  year={2022},
  publisher={Elsevier}
}

@article{kelly2017deterministic,
  title={{Deterministic homogenization for fast--slow systems with chaotic noise}},
  author={Kelly, D. and Melbourne, I.},
  journal={Journal of Functional Analysis},
  volume={272},
  number={10},
  pages={4063--4102},
  year={2017},
  publisher={Elsevier}
}

@article{Holm2002,
  author={D. D. Holm},
  title="{Lagrangian averages, averaged Lagrangians, and the mean effects of fluctuations in fluid dynamics}",
  journal={Chaos: An Interdisciplinary Journal of Nonlinear Science},
  volume={12},
  number={2},
  pages={518-530},
  year={2002},
  month={05},
  issn={1054-1500},
  doi={10.1063/1.1460941}
}

@article{andrews_mcintyre_1978, title={{An exact theory of nonlinear waves on a Lagrangian-mean flow}}, volume={89}, DOI={10.1017/S0022112078002773}, number={4}, journal={Journal of Fluid Mechanics}, publisher={Cambridge University Press}, author={Andrews, D. G. and Mcintyre, M. E.}, year={1978}, pages={609–646}}

@book{friz2020course,
    AUTHOR = {Friz, P. K. and Hairer, M.},
     TITLE = {A course on rough paths},
    SERIES = {Universitext},
      NOTE = {With an introduction to regularity structures,
              Second edition of [ 3289027]},
 PUBLISHER = {Springer, Cham},
      YEAR = {[2020] \copyright 2020},
     PAGES = {xvi+346},
      ISBN = {978-3-030-41556-3; 978-3-030-41555-6},
   MRCLASS = {60Lxx (34F05 35R60 60Hxx 93E03)},
  MRNUMBER = {4174393},
MRREVIEWER = {Fabrice Baudoin},
       DOI = {10.1007/978-3-030-41556-3},
       URL = {https://doi.org/10.1007/978-3-030-41556-3},
}

@article{melbourne2011note,
  title={A note on diffusion limits of chaotic skew-product flows},
  author={Melbourne, I. and Stuart, A.},
  journal={Nonlinearity},
  volume={24},
  number={4},
  pages={1361},
  year={2011},
  publisher={IOP Publishing}
}

@article{melbourne2015correction,
  title={{Correction to: A note on diffusion limits of chaotic skew-product flows}},
  author={Melbourne, I. and Stuart, A.},
  year={2015},
  publisher={Citeseer},
    journal={Nonlinearity}
}

@inproceedings{chevyrev2019multiscale,
  title={Multiscale systems, homogenization, and rough paths},
  author={Chevyrev, I. and Friz, P. K. and Korepanov, A. and Melbourne, I. and Zhang, H.},
  booktitle={Probability and Analysis in Interacting Physical Systems: In Honor of SRS Varadhan, Berlin, August, 2016},
  pages={17--48},
  year={2019},
  organization={Springer}
}

@article{diamantakis2024levy,
author = {Diamantakis, T. and Woodfield, J.},
title = {{Lévy Areas, Wong–Zakai Anomalies in Diffusive Limits of Deterministic Lagrangian Multitime Dynamics}},
journal = {SIAM Journal on Applied Dynamical Systems},
volume = {24},
number = {1},
pages = {836-893},
year = {2025},
doi = {10.1137/24M1637271},

URL = { https://doi.org/10.1137/24M1637271},
eprint = {  https://doi.org/10.1137/24M1637271}
}

@article{aref2007point,
  title={Point vortex dynamics: a classical mathematics playground},
  author={Aref, H.},
  journal={Journal of mathematical Physics},
  volume={48},
  number={6},
  pages={065401},
  year={2007},
  publisher={American Institute of Physics}
}

@book{poincare1893theorie,
  title={Th{\'e}orie des tourbillons: le{\c{c}}ons profess{\'e}es pendant le deuxi{\`e}me semestre 1891-1892},
  author={Poincar{\'e}, H.},
  year={1893},
  publisher={Gauthier-Villars}
}

@book{grobli1877specielle,
  title={{Specielle Probleme {\"u}ber die Bewegung geradliniger paralleler Wirbelf{\"a}den}},
  author={Gr{\"o}bli, W.},
  volume={8},
  year={1877},
  publisher={Druck von Z{\"u}rcher und Furrer}
}

@Article{Arnaudon2018,
author={Arnaudon, A.
and De Castro, A. L.
and Holm, D. D.},
title={{Noise and Dissipation on Coadjoint Orbits}},
journal={Journal of Nonlinear Science},
year={2018},
month={Feb},
day={01},
volume={28},
number={1},
pages={91-145},
issn={1432-1467},
doi={10.1007/s00332-017-9404-3},
url={https://doi.org/10.1007/s00332-017-9404-3}
}

@article{ACC2014,
   author = {M. Arnaudon and X. Chen and A. B. Cruzeiro},
   doi = {10.1063/1.4893357},
   issn = {00222488},
   issue = {8},
   journal = {Journal of Mathematical Physics},
   title = {{Stochastic {E}uler-{P}oincaré reduction}},
   volume = {55},
   year = {2014},
}

@article{CCR2023,
   author = {X. Chen and A. B. Cruzeiro and T. S. Ratiu},
   doi = {10.1007/s00332-022-09846-1},
   issn = {14321467},
   issue = {1},
   journal = {Journal of Nonlinear Science},
   month = {2},
   pages = {5},
   publisher = {Springer},
   title = {{Stochastic Variational Principles for Dissipative Equations with Advected Quantities}},
   volume = {33},
   url = {https://link.springer.com/10.1007/s00332-022-09846-1},
   year = {2023},
}

@article{HMR1998,
   author = {D. D. Holm and J. E. Marsden and T. S. Ratiu},
   doi = {10.1006/aima.1998.1721},
   issn = {00018708},
   issue = {1},
   journal = {Advances in Mathematics},
   month = {7},
   pages = {1-81},
   title = {{The {E}uler–{P}oincaré Equations and Semidirect Products with Applications to Continuum Theories}},
   volume = {137},
   url = {https://linkinghub.elsevier.com/retrieve/pii/S0001870898917212},
   year = {1998},
}

@Article{Drivas2020,
author={Drivas, T. D.
and Holm, D. D.
and Leahy, J-M.},
title={{Lagrangian Averaged Stochastic Advection by Lie Transport for Fluids}},
journal={Journal of Statistical Physics},
year={2020},
month={Jun},
day={01},
volume={179},
number={5},
pages={1304-1342},
issn={1572-9613},
doi={10.1007/s10955-020-02493-4},
url={https://doi.org/10.1007/s10955-020-02493-4}
}

@book{H_rmander_2003, title={{The Analysis of Linear Partial Differential Operators I}}, ISBN={9783642614972}, ISSN={1431-0821}, url={http://dx.doi.org/10.1007/978-3-642-61497-2}, DOI={10.1007/978-3-642-61497-2}, journal={Classics in Mathematics}, publisher={Springer Berlin Heidelberg}, author={Hörmander, L.}, year={2003} }

@article{H2019,
   author = {D. D. Holm},
   doi = {10.1007/s00332-019-09565-0},
   issn = {14321467},
   issue = {6},
   journal = {Journal of Nonlinear Science},
   pages = {2987-3031},
   publisher = {Springer US},
   title = {{Stochastic Closures for Wave–Current Interaction Dynamics}},
   volume = {29},
   url = {https://doi.org/10.1007/s00332-019-09565-0},
   year = {2019},
}

@article{DHP2023,
   author = {T. Diamantakis and D. D. Holm and G. A. Pavliotis},
   doi = {10.1137/22M1522164},
   issn = {1536-0040},
   issue = {2},
   journal = {SIAM Journal on Applied Dynamical Systems},
   month = {6},
   pages = {1182-1218},
   title = {{Variational Principles on Geometric Rough Paths and the Lévy Area Correction}},
   volume = {22},
   year = {2023},
}

@book{banyaga2013structure,
  title={{The Structure of Classical Diffeomorphism Groups}},
  author={Banyaga, A.},
  isbn={9781475768008},
  series={Mathematics and Its Applications},
  url={https://books.google.co.uk/books?id=ArEJCAAAQBAJ},
  year={2013},
  publisher={Springer US}
}

@article{GV2018,
title = "Geometric generalised {L}agrangian mean theories",
author = "Gilbert, {A. D.} and J. Vanneste",
year = "2018",
month = mar,
day = "25",
doi = "10.1017/jfm.2017.913",
language = "English",
volume = "839",
pages = "95--134",
journal = "Journal of Fluid Mechanics",
issn = "0022-1120",
publisher = "Cambridge University Press",
}

@misc{gilbert2024geometric,
      title={{Geometric approaches to Lagrangian averaging}}, 
      author={A. D. Gilbert and J. Vanneste},
      year={2024},
      eprint={2405.04394},
      archivePrefix={arXiv},
      primaryClass={physics.flu-dyn}
}

@article{DarrylDHolm_2002,
doi = {10.1088/0305-4470/35/3/313},
url = {https://dx.doi.org/10.1088/0305-4470/35/3/313},
year = {2002},
month = {jan},
publisher = {},
volume = {35},
number = {3},
pages = {679},
author = {D. D. Holm},
title = {{Variational principles for Lagrangian-averaged fluid dynamics}},
journal = {Journal of Physics A: Mathematical and General}
}

@article{SC2019,
   author = {O. D. Street and D. Crisan},
   doi = {10.1098/rspa.2020.0957},
   issn = {1364-5021},
   issue = {2247},
   journal = {Proceedings of the Royal Society A: Mathematical, Physical and Engineering Sciences},
   month = {3},
   pages = {20200957},
   publisher = {Royal Society Publishing},
   title = {Semi-martingale driven variational principles},
   volume = {477},
   url = {https://royalsocietypublishing.org/doi/10.1098/rspa.2020.0957},
   year = {2021},
}

@article{CFH19,
   author = {D. Crisan and F. Flandoli and D. D. Holm},
   doi = {10.1007/s00332-018-9506-6},
   issn = {0938-8974},
   issue = {3},
   journal = {Journal of Nonlinear Science},
   month = {6},
   pages = {813-870},
   title = {{Solution Properties of a 3D Stochastic Euler Fluid Equation}},
   volume = {29},
   url = {http://link.springer.com/10.1007/s00332-018-9506-6},
   year = {2019},
}

@article{BTB2015,
   author = {P. Bauer and A. Thorpe and G. Brunet},
   doi = {10.1038/nature14956},
   issn = {0028-0836},
   issue = {7567},
   journal = {Nature},
   month = {9},
   pages = {47-55},
   title = {The quiet revolution of numerical weather prediction},
   volume = {525},
   year = {2015},
}

@article{BMP1999,
   author = {R. Buizza and M. Milleer and T. N. Palmer},
   doi = {10.1002/qj.49712556006},
   issn = {00359009},
   issue = {560},
   journal = {Quarterly Journal of the Royal Meteorological Society},
   month = {8},
   pages = {2887-2908},
   title = {{Stochastic representation of model uncertainties in the ECMWF ensemble prediction system}},
   volume = {125},
   year = {1999},
}

@inproceedings{HP2023,
   author = {R. Hu and S. Patching},
   city = {Cham},
   editor = {Bertrand Chapron and Dan Crisan and D. D. Holm and Etienne Mémin and Anna Radomska},
   isbn = {978-3-031-18988-3},
   journal = {Stochastic Transport in Upper Ocean Dynamics},
    booktitle = {Stochastic Transport in Upper Ocean Dynamics},
   pages = {135-158},
   publisher = {Springer International Publishing},
   title = {{Variational Stochastic Parameterisations and Their Applications to Primitive Equation Models}},
   year = {2023},
}

@article{Arnold1966,
   author = {V. I. Arnold},
   doi = {10.5802/aif.233},
   issn = {0373-0956},
   issue = {1},
   journal = {Annales de l’institut Fourier},
   pages = {319-361},
   title = {Sur la géométrie différentielle des groupes de Lie de dimension infinie et ses applications à l'hydrodynamique des fluides parfaits},
   volume = {16},
   year = {1966},
}

@article{gehringer2019homogenization,
  title={Homogenization with fractional random fields},
  author={Gehringer, J. and Li, X-M.},
  journal={arXiv preprint arXiv:1911.12600},
  year={2019}
}

@article{MR2004,
   author = {R. Mikulevicius and B. L. Rozovskii},
   doi = {10.1137/S0036141002409167},
   issn = {0036-1410},
   issue = {5},
   journal = {SIAM Journal on Mathematical Analysis},
   month = {1},
   pages = {1250-1310},
   title = {{Stochastic Navier--Stokes Equations for Turbulent Flows}},
   volume = {35},
   year = {2004},
}

@article{GH18a,
   author = {F. Gay-Balmaz and D. D. Holm},
   doi = {10.1007/s00332-017-9431-0},
   issn = {14321467},
   issue = {3},
   journal = {Journal of Nonlinear Science},
   month = {6},
   pages = {873-904},
   publisher = {Springer New York LLC},
   title = {{Stochastic Geometric Models with Non-stationary Spatial Correlations in Lagrangian Fluid Flows}},
   volume = {28},
   year = {2018},
}

@article{Memin2014,
   author = {E. Mémin},
   doi = {10.1080/03091929.2013.836190},
   issn = {0309-1929},
   issue = {2},
   journal = {Geophysical \& Astrophysical Fluid Dynamics},
   month = {3},
   pages = {119-146},
   title = {{Fluid flow dynamics under location uncertainty}},
   volume = {108},
   year = {2014},
}

@article{H2021,
   author = {D. D. Holm},
   doi = {10.1007/s00332-020-09665-2},
   issn = {0938-8974},
   issue = {1},
   journal = {Journal of Nonlinear Science},
   month = {2},
   pages = {4},
   title = {{Stochastic Variational Formulations of Fluid Wave–Current Interaction}},
   volume = {31},
   url = {http://link.springer.com/10.1007/s00332-020-09665-2},
   year = {2021},
}

@article{CCHOS18a,
   author = {C. J. Cotter and D. Crisan and D. D. Holm and W. Pan and I. Shevchenko},
   doi = {10.1137/18M1167929},
   issn = {1540-3459},
   issue = {1},
   journal = {Multiscale Modeling \& Simulation},
   month = {1},
   pages = {192-232},
   title = {{Numerically Modeling Stochastic Lie Transport in Fluid Dynamics}},
   volume = {17},
   url = {https://epubs.siam.org/doi/10.1137/18M1167929},
   year = {2019},
}

@article{CCHOS18,
   author = {C. J. Cotter and D. Crisan and D. D. Holm and W. Pan and I. Shevchenko},
   doi = {10.3934/fods.2020010},
   issn = {2639-8001},
   issue = {2},
   journal = {Foundations of Data Science},
   pages = {173-205},
   title = {Modelling uncertainty using stochastic transport noise in a 2-layer quasi-geostrophic model},
   volume = {2},
   url = {http://aimsciences.org//article/doi/10.3934/fods.2020010},
   year = {2020},
}

@article{CCHPS2021,
   author = {C. J. Cotter and D. Crisan and D. D. Holm and W. Pan and I. Shevchenko},
   doi = {10.1137/19M1277606},
   issn = {21662525},
   issue = {4},
   journal = {SIAM-ASA Journal on Uncertainty Quantification},
   pages = {1446-1492},
   title = {{A particle filter for Stochastic Advection by Lie Transport: A case study for the damped and forced incompressible two-dimensional Euler equation}},
   volume = {8},
   year = {2021},
}

@article{CL2023,
   author = {D. Crisan and O. Lang},
   doi = {10.1007/s10884-022-10243-1},
   issn = {1040-7294},
   journal = {Journal of Dynamics and Differential Equations},
   month = {1},
   title = {{Well-Posedness Properties for a Stochastic Rotating Shallow Water Model}},
   year = {2023},
}

@misc{CLLLP2023,
      title={Noise calibration for the stochastic rotating shallow water model}, 
      author={D. Crisan and O. Lang and A. Lobbe and P. Jan van Leeuwen and R. Potthast},
      year={2023},
      eprint={2305.03548},
      archivePrefix={arXiv},
      primaryClass={math.DS}
}

@article{HHS2023,
   author = {D. D. Holm and R. Hu and O. D. Street},
   doi = {10.1016/j.physd.2023.133847},
   issn = {01672789},
   journal = {Physica D: Nonlinear Phenomena},
   month = {7},
   pages = {133847},
   title = {{Lagrangian Reduction and Wave Mean Flow Interaction}},
   url = {https://linkinghub.elsevier.com/retrieve/pii/S0167278923002014},
   year = {2023},
}

@article{H2001,
   author = {D. D. Holm},
   doi = {10.1007/0-387-21791-6_4},
   journal = {Geometry, Mechanics, and Dynamics},
   month = {3},
   pages = {169-180},
   title = {{Euler-Poincar\'e Dynamics of Perfect Complex Fluids}},
   url = {http://arxiv.org/abs/nlin/0103041},
   year = {2001},
}

@book{holm2009geometric,
  title={Geometric mechanics and symmetry: from finite to infinite dimensions},
  author={Holm, D. D. and Schmah, T. and Stoica, C.},
  volume={12},
  year={2009},
  publisher={Oxford University Press}
}

@book{friz2010multidimensional,
  title={Multidimensional stochastic processes as rough paths: theory and applications},
  author={Friz, P. K. and Victoir, N. B.},
  volume={120},
  year={2010},
  publisher={Cambridge University Press}
}

@article{https://doi.org/10.1112/jlms.12585,
author = {Armstrong, J. and Brigo, D. and Cass, T. and Rossi Ferrucci, E.},
title = {Non-geometric rough paths on manifolds},
journal = {Journal of the London Mathematical Society},
volume = {106},
number = {2},
pages = {756-817},
doi = {https://doi.org/10.1112/jlms.12585},
url = {https://londmathsoc.onlinelibrary.wiley.com/doi/abs/10.1112/jlms.12585},
eprint = {https://londmathsoc.onlinelibrary.wiley.com/doi/pdf/10.1112/jlms.12585},
year = {2022}
}

@book{hirsch1976differential,
  title={{Differential Topology}},
  author={Hirsch, M.W.},
  isbn={9780387901480},
  lccn={97205726},
  series={Graduate texts in mathematics},
  url={https://books.google.co.uk/books?id=uETvAAAAMAAJ},
  year={1976},
  publisher={Springer}
}

@article{debussche2024variational,
  title={Variational principles for fully coupled stochastic fluid dynamics across scales},
  author={Debussche, A. and M{\'e}min, E.},
  journal={arXiv preprint arXiv:2409.12654},
  year={2024}
}

@article{FP2021,
   author = {F. Flandoli and U. Pappalettera},
   doi = {10.1007/s00332-021-09681-w},
   issn = {0938-8974},
   issue = {1},
   journal = {Journal of Nonlinear Science},
   month = {2},
   pages = {24},
   title = {{2D Euler Equations with Stratonovich Transport Noise as a Large-Scale Stochastic Model Reduction}},
   volume = {31},
   year = {2021},
}

@misc{DH2023,
      title={Rough analysis of two scale systems}, 
      author={A. Debussche and M. Hofmanová},
      year={2023},
      eprint={2306.15781},
      archivePrefix={arXiv},
      primaryClass={math.PR},
      url={https://arxiv.org/abs/2306.15781}, 
}

@article{DR1977,
   author = {M. J. Dupré and S. I. Rosencrans},
   doi = {10.1063/1.523861},
   issn = {00222488},
   issue = {7},
   journal = {Journal of Mathematical Physics},
   pages = {1532-1535},
   title = {{Classical and relativistic vorticity in a semi-Riemannian manifold}},
   volume = {19},
   year = {1977},
}

@article{10.1214/21-AIHP1203,
author = {I. Chevyrev and P. Friz and A. Korepanov and I. Melbourne and H. Zhang},
title = {{Deterministic homogenization under optimal moment assumptions for fast–slow systems. Part 2}},
volume = {58},
journal = {Annales de l'Institut Henri Poincaré, Probabilités et Statistiques},
number = {3},
publisher = {Institut Henri Poincaré},
pages = {1328 -- 1350},
year = {2022},
doi = {10.1214/21-AIHP1203},
URL = {https://doi.org/10.1214/21-AIHP1203}
}

@article{10.1214/21-AIHP1202,
author = {A. Korepanov and Z. Kosloff and I. Melbourne},
title = {{Deterministic homogenization under optimal moment assumptions for fast-slow systems. Part 1}},
volume = {58},
journal = {Annales de l'Institut Henri Poincaré, Probabilités et Statistiques},
number = {3},
publisher = {Institut Henri Poincaré},
pages = {1305 -- 1327},
year = {2022},
doi = {10.1214/21-AIHP1202},
URL = {https://doi.org/10.1214/21-AIHP1202}
}

@book{AMR1988,
   author = {R. Abraham and J. E. Marsden and T. Ratiu},
   city = {New York, NY},
   doi = {10.1007/978-1-4612-1029-0},
   isbn = {978-1-4612-6990-8},
   publisher = {Springer New York},
   title = {{Manifolds, Tensor Analysis, and Applications}},
   volume = {75},
   year = {1988},
}

@book{MR1999,
   author = {J. E. Marsden and T. S. Ratiu},
   city = {New York, NY},
   doi = {10.1007/978-0-387-21792-5},
   isbn = {978-1-4419-3143-6},
   publisher = {Springer New York},
   title = {{Introduction to Mechanics and Symmetry}},
   volume = {17},
   url = {http://link.springer.com/10.1007/978-0-387-21792-5},
   year = {1999},
}

@inbook{marsden_ebin_fischer_1972, title={Diffeomorphism groups, hydrodynamics and relativity}, ISBN={9780919558038}, booktitle={Proceedings of the 13th Biennial Seminar of Canadian Mathematical Congress}, publisher={Canadian Mathematical Congress}, author={Marsden, J. E. and Ebin, D. G. and Fischer, A. E.}, year={1972}, pages={135–279} }

@article{CRISAN2022109632,
title = {{Solution properties of the incompressible Euler system with rough path advection}},
journal = {Journal of Functional Analysis},
volume = {283},
number = {9},
pages = {109632},
year = {2022},
issn = {0022-1236},
doi = {https://doi.org/10.1016/j.jfa.2022.109632},
url = {https://www.sciencedirect.com/science/article/pii/S002212362200252X},
author = {D. Crisan and D. D. Holm and J-M. Leahy and T. Nilssen}
}

@book{pavliotis2008multiscale,
  title={{Multiscale Methods: Averaging and Homogenization}},
  author={Pavliotis, G. and Stuart, A.},
  year={2008},
  publisher={Springer Science \& Business Media}
}

@book{billingsley1968convergence,
  title={{Convergence of Probability Measures}},
  author={Billingsley, P.},
  isbn={9780471072423},
  lccn={68239229},
  series={Wiley Series in Probability and Mathematical Statistics},
  year={1968},
  publisher={Wiley}
}

@article{Bretherton_1970, title={{A note on Hamilton’s principle for perfect fluids}}, volume={44}, DOI={10.1017/S0022112070001660}, number={1}, journal={Journal of Fluid Mechanics}, author={Bretherton, F. P.}, year={1970}, pages={19–31}}

@article{Yudovich1995,
   author = {V. I. Yudovich},
   doi = {10.4310/MRL.1995.v2.n1.a4},
   issn = {10732780},
   issue = {1},
   journal = {Mathematical Research Letters},
   pages = {27-38},
   title = {{Uniqueness Theorem for the Basic Nonstationary Problem in the Dynamics of an Ideal Incompressible Fluid}},
   volume = {2},
   year = {1995},
}

@article{kunita1996stochastic,
  title={Stochastic differential equations with jumps and stochastic flows of diffeomorphisms},
  author={Kunita, H.},
  journal={It{\^o}’s stochastic calculus and probability theory},
  pages={197--211},
  year={1996},
  publisher={Springer}
}

@article{Yudovich1963,
   author = {V. I. Yudovich},
   doi = {10.1016/0041-5553(63)90247-7},
   issn = {00415553},
   issue = {6},
   journal = {USSR Computational Mathematics and Mathematical Physics},
   month = {1},
   pages = {1407-1456},
   title = {Non-stationary flow of an ideal incompressible liquid},
   volume = {3},
   year = {1963},
}

@inproceedings{ventzel1965equations,
  title={On equations of theory of conditional Markov processes},
  author={A.D. Ventzel},
  booktitle={Theory of probability and its applications, USSR},
  volume={10},
  pages={357},
  year={1965}
}

@article{SOWARD_ROBERTS_2010, title={{The hybrid Euler–Lagrange procedure using an extension of Moffatt’s method}}, volume={661}, DOI={10.1017/S0022112010002867}, journal={Journal of Fluid Mechanics}, author={Soward, A. M. and Roberts, P. H.}, year={2010}, pages={45–72}}

@article{10.1214/14-AOP979,
author = {D. Kelly and I. Melbourne},
title = {{Smooth approximation of stochastic differential equations}},
volume = {44},
journal = {The Annals of Probability},
number = {1},
publisher = {Institute of Mathematical Statistics},
pages = {479 -- 520},
year = {2016},
doi = {10.1214/14-AOP979},
URL = {https://doi.org/10.1214/14-AOP979}
}

@article{FRIZ20093236,
title = {{Rough path limits of the Wong–Zakai type with a modified drift term}},
journal = {Journal of Functional Analysis},
volume = {256},
number = {10},
pages = {3236-3256},
year = {2009},
issn = {0022-1236},
doi = {https://doi.org/10.1016/j.jfa.2009.02.010},
url = {https://www.sciencedirect.com/science/article/pii/S0022123609000858},
author = {P. Friz and H. Oberhauser}
}

@misc{ST2023,
      title={{Semimartingale driven mechanics and reduction by symmetry for stochastic and dissipative dynamical systems}}, 
      author={O. D. Street and S. Takao},
      year={2023},
      eprint={2312.09769},
      archivePrefix={arXiv},
      primaryClass={math-ph},
      url={https://arxiv.org/abs/2312.09769}, 
}

@article{CENDRA198763,
title = {{Lin constraints, Clebsch potentials and variational principles}},
journal = {Physica D: Nonlinear Phenomena},
volume = {27},
number = {1},
pages = {63-89},
year = {1987},
issn = {0167-2789},
doi = {https://doi.org/10.1016/0167-2789(87)90005-4},
url = {https://www.sciencedirect.com/science/article/pii/0167278987900054},
author = {H. Cendra and J. E. Marsden}
}

@misc{len2023geometric,
    title={{A Geometric Extension of the Itô-Wentzell and Kunita's Formulas}},
    author={A. B. de León and S. Takao},
    year={2023},
    eprint={2311.04439},
    archivePrefix={arXiv},
    primaryClass={math.PR}
}

@ARTICLE{BAXENDALE1984,
  author={Baxendale, P.},
  title={{Brownian motions in the diffeomorphism group 1}},
  journal={Compositio mathematica},
  year={1984},
  affiliation={Univ. Aberdeen, dep. mathematics, Aberdeen AB9 2TY, United Kingdom},
  descriptors={Mouvement brownien; Processus stochastique; Variété mathématique ; Groupe difféomorphisme},
  subject={Théorie des probabilités, processus stochastiques et statistique[001B00B50]},
  issn={0010-437X},
  language={English},
  document_type={Article},
  inist_number={9014271} 
}

@article{HONG2025405,
title = {{Diffusion approximation for multi-scale McKean-Vlasov SDEs through different methods}},
journal = {Journal of Differential Equations},
volume = {414},
pages = {405-454},
year = {2025},
issn = {0022-0396},
doi = {https://doi.org/10.1016/j.jde.2024.09.012},
url = {https://www.sciencedirect.com/science/article/pii/S0022039624005928},
author = {W. Hong and S. Li and X. Sun}
}

@article{GH1996,
    title = {{Self-consistent Hamiltonian dynamics of wave mean-flow interaction for a rotating stratified incompressible fluid}},
    year = {1996},
    journal = {Physica D: Nonlinear Phenomena},
    author = {I. Gjaja and D. D. Holm},
    number = {2-4},
    pages = {343--378},
    volume = {98},
    doi = {10.1016/0167-2789(96)00104-2},
    issn = {01672789}
}

@article{GBR09,
    title = {{The geometric structure of complex fluids}},
    year = {2009},
    journal = {Advances in Applied Mathematics},
    author = {Gay-Balmaz, F. and Ratiu, T. S.},
    number = {2},
    month = {2},
    pages = {176--275},
    volume = {42},
    doi = {10.1016/j.aam.2008.06.002},
    issn = {01968858},
    arxivId = {0903.4294}
}

@article{HP2025,
   author = {R. Hu and L. Peng},
   doi = {10.1111/sapm.70112},
   issn = {0022-2526},
   issue = {3},
   journal = {Studies in Applied Mathematics},
   month = {9},
   title = {{Stochastic Multisymplectic PDEs and Their Structure‐Preserving Numerical Methods}},
   volume = {155},
   url = {https://onlinelibrary.wiley.com/doi/10.1111/sapm.70112},
   year = {2025}
}

@misc{galeati2025wellposednessnonlinearroughcontinuity,
      title={On the well-posedness of (nonlinear) rough continuity equations}, 
      author={L. Galeati and J-M. Leahy and T. Nilssen},
      year={2025},
      eprint={2502.04982},
      archivePrefix={arXiv},
      primaryClass={math.AP},
      url={https://arxiv.org/abs/2502.04982}, 
}
\end{document}